\documentclass[11pt]{article}
\usepackage{amsfonts}
\usepackage{graphics,amssymb,amsthm,amsmath,epsf,rotating}
\usepackage{graphicx}
\usepackage{subfigure}

\usepackage[round]{natbib}

\newtheorem{theorem}{Theorem}[section]

\newtheorem{prop}{Proposition}

\numberwithin{equation}{section}
\newtheorem{remark}{Remark}[section]

\allowdisplaybreaks

\begin{document}

\date{}
\title{Generalized Gompertz-power series distributions}
\author{ S. Tahmasebi$^1$ , A. A. Jafari$^{2,}$\thanks{Corresponding: aajafari@yazd.ac.ir} \\
{\small $^1$Department of Statistics, Persian Gulf University, Bushehr, Iran}\\
{\small $^{2}$Department of Statistics, Yazd University, Yazd, Iran}
}
\date{}
\maketitle
\begin{abstract}
In this paper, we introduce the generalized Gompertz-power series class of distributions
which is obtained by compounding generalized Gompertz and power series distributions. This compounding procedure follows same way that was previously carried out by
\cite{si-bo-di-co-13}
and
\cite{ba-de-co-11}
 in introducing the compound class of extended Weibull-power series distribution and the Weibull-geometric distribution, respectively. This distribution contains several lifetime models such as generalized Gompertz, generalized Gompertz-geometric, generalized Gompertz-poisson, generalized Gompertz-binomial  distribution, and generalized  Gompertz-logarithmic  distribution as special cases. The hazard rate function of the new class of distributions can be increasing, decreasing and bathtub-shaped. We obtain several properties of this distribution such as its probability density function, Shannon entropy, its mean residual life and failure rate functions, quantiles and moments. The maximum likelihood estimation procedure via a EM-algorithm is presented, and sub-models of the distribution are studied in details.
\end{abstract}

{\bf Keywords:} EM algorithm, Generalized  Gompertz distribution, Maximum likelihood estimation, Power series distributions.

\section{Introduction}

The exponential distribution is commonly used in many applied problems, particularly in lifetime data analysis \citep{lawless-03}. A generalization of this distribution is the Gompertz distribution. It is a lifetime distribution and is often applied to describe the distribution of adult life spans by actuaries and demographers.
 The Gompertz distribution is considered for the analysis of survival in some sciences such as  biology, gerontology, computer, and marketing science.  Recently, \cite{gu-ku-99} defined the generalized exponential distribution
 and in similar manner,
  \cite{el-al-al-13} introduced the generalized Gompertz (GG) distribution. A random variable $X$ is said to have a GG distribution
 denoted by ${\rm GG}(\alpha,\beta,\gamma)$, if its
cumulative distribution function (cdf) is
\begin{eqnarray}\label{eq.FG}
G(x)=[1-e^{-\frac{\beta}{\gamma}(e^{\gamma x}-1)}]^{\alpha},\;\;\alpha,\beta>0,\;\;\gamma>0; \;\;x\geq0.
\end{eqnarray}
and the probability density function (pdf) is
\begin{eqnarray}\label{eq.fG}
{\rm g}(x)= \alpha \beta e^{\gamma x}e^{-\frac{\beta}{\gamma}(e^{\gamma x}-1)}[1-e^{\frac{-\beta}{\gamma}(e^{\gamma x}-1)}]^{\alpha-1}.
\end{eqnarray}

The GG distribution  is a flexible distribution that can be skewed to the right and to the left, and 
the well-known distributions are special cases of this distribution: the generalized exponential  proposed by \cite{gu-ku-99} when $\gamma\rightarrow0^+$, the Gompertz distribution when $\alpha=1$, and the exponential distribution when $\alpha=1$ and $\gamma\rightarrow0^+$.

In this paper, we compound the generalized Gompertz  and power series distributions,  and  introduce a new class of distribution. This procedure follows similar way that was previously carried out by some authors: The exponential-power series distribution is introduced by
\cite{ch-ga-09}
which is concluded the exponential-geometric
\citep{ad-di-05, ad-lo-98},
 exponential-Poisson
 \citep{kus-07},
 and exponential-logarithmic
 \citep{ta-re-08}
  distributions; the Weibull-power series distributions  is introduced by
  \cite{mo-ba-11}
   and is a generalization of the exponential-power series distribution;
the generalized exponential-power series distribution is introduced by
\cite{ma-ja-12}
which is concluded the Poisson-exponential
\citep{ca-lo-fr-ba-11,lo-ca-ba-11}
complementary exponential-geometric
\citep{lo-ro-ca-11},
and the complementary exponential-power series
\citep{fl-bo-ca-13}
distributions; linear failure rate-power series distributions \citep{ma-ja-14}.

The remainder of our paper is organized as follows: In Section \ref{sec.dis}, we give the probability density and failure rate functions of the new distribution. Some properties such as quantiles, moments, order statistics, Shannon entropy and mean residual life  are given in Section \ref{sec.pro}. In Section \ref{sec.spe}, we consider four special cases of this new distribution. We discuss estimation by maximum likelihood and provide an expression for Fisher's information matrix in Section \ref{sec.est}.
A simulation study is performed in Section \ref{sec.sim}.
An application is given in the Section \ref{sec.exa}.

\vspace{2cm}

\section{The generalized Gompertz-power series model}
\label{sec.dis}

A discrete random variable, $N$  is a member of power series distributions (truncated at zero) if its probability mass function is given by
\begin{equation}\label{eq.ps}
p_n=P(N=n)=\frac{a_{n}\theta^{n}}{C(\theta)},\;\; n=1,2,\dots,
\end{equation}
where $a_{n}\geq0$ depends only on $n$, $C(\theta)=\sum_{n=1}^{\infty}a_{n}\theta^{n}$, and $\theta\in(0,s)$ ($s$ can be $\infty$) is such  that $C(\theta)$ is finite.
Table \ref{tab.ps}
summarizes some particular cases of the truncated (at zero) power series distributions (geometric, Poisson, logarithmic and binomial). Detailed properties of power series distribution can be found in
\cite{noack-50}.
Here, $C'(\theta )$, $C''(\theta)$ and $C'''(\theta)$ denote the first, second and third derivatives of $C(\theta)$ with respect to $\theta$, respectively.

\begin{table}[ht]

\begin{center}
\caption{Useful quantities for some power series distributions.}\label{tab.ps}
\begin{tabular}{|l| c c c c c c| }\hline
Distribution
& $a_n$ & $C(\theta)$ & $C^{\prime}(\theta)$ & $C^{\prime\prime}(\theta)$ & $C^{\prime\prime\prime}(\theta)$ &   $s$ \\ \hline
Geometric & $1$ & $\theta (1-\theta)^{-1}$ & $(1-\theta)^{-2}$ & $2(1-\theta)^{-3}$ & $6(1-\theta)^{-4}$ &  $1$ \\
Poisson & $n!^{-1}$ & $e^{\theta}-1$ & $e^{\theta}$ & $e^{\theta}$ & $e^{\theta}$ &  $\infty$ \\
Logarithmic & $n^{-1}$ & $-\log(1-\theta)$ & $(1-\theta)^{-1}$ & $(1-\theta)^{-2}$ & $2(1-\theta)^{-3}$  & $1$ \\
Binomial & $\binom {m} {n}$
& $(1+\theta)^m-1$ & $\frac{m}{(\theta+1)^{1-m}}$ & $\frac{m(m-1)}{(\theta+1)^{2-m}}$ & $\frac{m(m-1)(k-2)}{(\theta+1)^{3-m}}$ &  $\infty$ \\
\hline
\end{tabular}

\end{center}
\end{table}

 We define  generalized Gompertz-Power Series (GGPS) class of distributions denoted as ${\rm GGPS}(\alpha,\beta,\gamma,\theta)$ with cdf
 \begin{equation}\label{FGP}
F(x)=\sum_{n=1}^{\infty} \frac{a_{n}(\theta G(x))^{n}}{C(\theta)}=
\frac{C(\theta G(x))}{C(\theta)}=\frac{C(\theta t^{\alpha})}{C(\theta)}, \;\; x>0,
\end{equation}
where $t=1-e^{-\frac{\beta}{\gamma}(e^{\gamma x}-1)}$.  The pdf of this distribution  is given by
\begin{equation}\label{fGP}
f(x)
=
\frac{\theta \alpha \beta}{C(\theta)}e^{\gamma x}(1-t)t^{\alpha-1}C'\left(\theta t^{\alpha}\right).
\end{equation}

This class of distribution is obtained by compounding the Gompertz distribution and power series class of distributions as follows.
Let $N$ be a  random variable denoting the number of failure causes which it is a member of power series distributions (truncated at zero). For given $N$, let $X_{1},X_{2},\dots,X_{N}$ be a independent random sample of size $N$ from a ${\rm GG}(\alpha,\beta,\gamma)$ distribution.
Let $X_{(N)}=\max_{1\leq i\leq N}X_{i}$. Then, the conditional cdf of $X_{(N)}\mid N=n$ is given by
 $$
G_{X_{(N)}\mid N=n}(x)=[1-e^{-\frac{\beta}{\gamma}(e^{\gamma x}-1)}]^{n\alpha},
$$
which  has ${\rm GG}(n\alpha,\beta,\gamma)$ distribution.
Hence, we obtain
$$
P(X_{(N)}\leq x,N=n)=\frac{a_{n}(\theta G(x))^{n}}{C(\theta)}=
\frac{a_{n}\theta^{n}}{C(\theta)}[1-e^{-\frac{\beta}{\gamma}(e^{\gamma x}-1)}]^{n\alpha}.
$$
Therefore, the marginal cdf of $X_{(N)}$  has GGPS distribution.  This class of distributions can be applied to reliability problems. Therefore,  some of its properties are investigated in the following.

\begin{prop}
 The pdf of GGPS class can be expressed as infinite linear combination of pdf of order distribution, i.e. it can be written as
\begin{eqnarray}
f(x)=\sum_{n=1}^{\infty} p_n \ {\rm g}_{(n)}(x;n\alpha,\beta,\gamma),
\end{eqnarray}
where ${\rm g}_{(n)}(x;n\alpha,\beta,\gamma)$ is the pdf of ${\rm GG}(n\alpha,\beta,\gamma)$.
\end{prop}

\begin{proof}
Consider $t=1-e^{-\frac{\beta}{\gamma}(e^{\gamma x}-1)}$. So
\begin{eqnarray*}
f(x)&=&\frac{\theta \alpha \beta}{C(\theta)}e^{\gamma x}(1-t)t^{\alpha-1}C'\left(\theta t^{\alpha}\right)
=\frac{\theta \alpha \beta}{C(\theta)}e^{\gamma x}(1-t)t^{\alpha-1}
\sum\limits_{n=1}^{\infty}n a_{n}(\theta t^{\alpha})^{n-1}\\
&=&\sum\limits_{n=1}^{\infty}\frac{a_{n}\theta^{n}}{C(\theta)}n\alpha\beta (1-t)e^{\gamma x} t^{n\alpha-1}
=\sum\limits_{n=1}^{\infty} p_{n} {\rm g}_{(n)}(x;n\alpha,\beta,\gamma).
\end{eqnarray*}
\end{proof}

\begin{prop}\label{prop.1}
The limiting distribution of ${\rm GGPS}(\alpha,\beta,\gamma,\theta)$ when $\theta\rightarrow 0^{+}$ is
\[ {\mathop{\lim }_{\theta \rightarrow 0^{+}} F(x)}=[1-e^{-\frac{\beta}{\gamma}(e^{\gamma x}-1)}]^{c\alpha},\]
which is a GG distribution with parameters $c\alpha$, $\beta$, and $\gamma$, where $c=\min\{n\in \mathbb{N}: a_{n}>0\}$.
\end{prop}

\begin{proof}
Consider $t=1-e^{-\frac{\beta}{\gamma}(e^{\gamma x}-1)}$. So
\begin{eqnarray*}
 {\mathop{\lim }_{\theta\rightarrow 0^{+}} F(x)}&=&\mathop{\lim }_{\theta\rightarrow 0^{+}}\frac{C(\lambda t^\alpha)}{C(\theta)}={\mathop{\lim }_{\lambda\rightarrow 0^{+}}\frac{\sum\limits_{n=1}^{\infty}a_{n}\theta^{n}t^{n\alpha}}{\sum\limits_{n=1}^{\infty}a_{n}\theta^{n}}} \\
&=&{\mathop{\lim }_{\theta\rightarrow 0^{+}}\frac{a_{c}t^{c\alpha}+\sum\limits_{n=c+1}^{\infty}a_{n}\theta^{n-c}t^{n\alpha}}{a_{c}+\sum\limits_{n=c+1}^{\infty}a_{n}\theta^{n-c}}}
=t^{c\alpha}.
\end{eqnarray*}
\end{proof}

\begin{prop}
 The limiting distribution of ${\rm GGPS}(\alpha,\beta,\gamma,\theta)$ when $\gamma\rightarrow 0^+$ is
\[ \lim_{\gamma \rightarrow 0^{+}}  F(x)=\frac{C(\theta(1-e^{-\beta x})^{\alpha})}{C(\theta)},\]
i.e.  the cdf of the generalized exponential-power series class of distribution introduced by
\cite{ma-ja-12}.
\end{prop}

\begin{proof}
When $\gamma\rightarrow 0^+$, the generalized Gompertz distribution becomes to generalized exponential distribution. Therefore, proof is obvious.
\end{proof}

\begin{prop}
The hazard rate function of the GGPS class of distributions is
\begin{eqnarray}\label{hGP}
 h(x)= \frac{\theta \alpha \beta e^{\gamma x}(1-t)t^{\alpha-1} C'(\theta t^{\alpha})}{C(\theta)-C(\theta t^{\alpha})},
\end{eqnarray}
where $t=1-e^{\frac{-\beta}{\gamma}(e^{\gamma x}-1)}$.
\end{prop}

\begin{proof}
Using \eqref{FGP}, \eqref{fGP} and definition of  hazard rate function as $h(x)=f(x)/(1-F(x))$, the proof is obvious.
\end{proof}

\begin{prop}
For the pdf in \eqref{fGP}, we have
$$
\mathop{\lim }_{x \rightarrow 0^+}  f(x)=\left\{\begin{array}{ll}
\infty &   0<\alpha<1 \\
\frac{C'(0)\theta\beta }{C(\theta)} & \alpha=1 \\
0 & \alpha>1, \\  \end{array}
\right. \;\;\;\; \;\;\;  \lim_{x\rightarrow \infty}f(x)= 0.
$$
\end{prop}

\begin{proof}
The proof is a forward calculation using the following limits
$$
\mathop{\lim }_{x \rightarrow 0^+}  t^{\alpha-1}=\left\{\begin{array}{ll}
\infty &   0<\alpha<1 \\
1 & \alpha=1 \\
0 & \alpha>1, \\  \end{array}
\right. \;\;\;\; \;\;\;  \lim_{x\rightarrow 0^+} t^\alpha= 0, \;\;\;\; \;\;\;  \lim_{x\rightarrow \infty} t=1.
$$
\end{proof}

\begin{prop}
For the hazard rate function  in \eqref{hGP}, we have
\[
\mathop{\lim }_{x \rightarrow 0^+}  h(x)=\left\{\begin{array}{ll}
\infty &   0<\alpha<1 \\
\frac{C'(0)\theta\beta }{C(\theta)} & \alpha=1 \\
0 & \alpha>1, \\  \end{array}
\right. \;\;\;\; \;\;\;\mathop{\lim }_{x \rightarrow \infty}
h(x)=\left\{\begin{array}{ll}
\infty &   \gamma>0 \\
\beta & \gamma\rightarrow 0
\end{array}
\right.
  \]
\end{prop}

\begin{proof}
Since $\lim_{x\rightarrow 0^+} (1-F(x))=1$, we have
$\lim_{x\rightarrow 0^+} h(x)=\lim_{x\rightarrow 0^+} f(x)$.

For $\lim_{x\rightarrow \infty} h(x)$, the proof is satisfied using the limits
\begin{eqnarray*}
&&\lim_{x\rightarrow \infty} C'(\theta t^\alpha)=C'(\theta),\ \ \ \  \  \ \lim_{x\rightarrow \infty} t^{\alpha-1}=1, \\
&& \lim_{x\rightarrow \infty}
\frac{ e^{\gamma x}(1-t)}{C(\theta)-C(\theta t^{\alpha})}=
\lim_{x \rightarrow \infty}\frac{e^{\gamma x}(1-t)[\beta e^{\gamma x}-\gamma]}{\theta\beta\alpha C'(\theta)e^{\gamma x}(1-t)}
=\left\{\begin{array}{ll}
\infty &   \gamma>0 \\
\frac{1}{\theta\alpha C'(\theta)} & \gamma\rightarrow 0.
\end{array}
\right.
\end{eqnarray*}
\end{proof}

As a example, we consider $C\left(\theta \right)=\theta +{\theta }^{20}$. The plots of  pdf and  hazard rate function of GGPS for parameters $\beta=1, \gamma=.01, \theta=1.0$, and $\alpha=0.1, 0.5, 1.0, 2.0$ are given in Figure \ref{dhbim}. This pdf is bimodal when $\alpha=2.0$, and the values  of modes are 0.7 and 3.51.

\begin{figure}
\centering
\includegraphics[width=7.75cm,height=6.5cm]{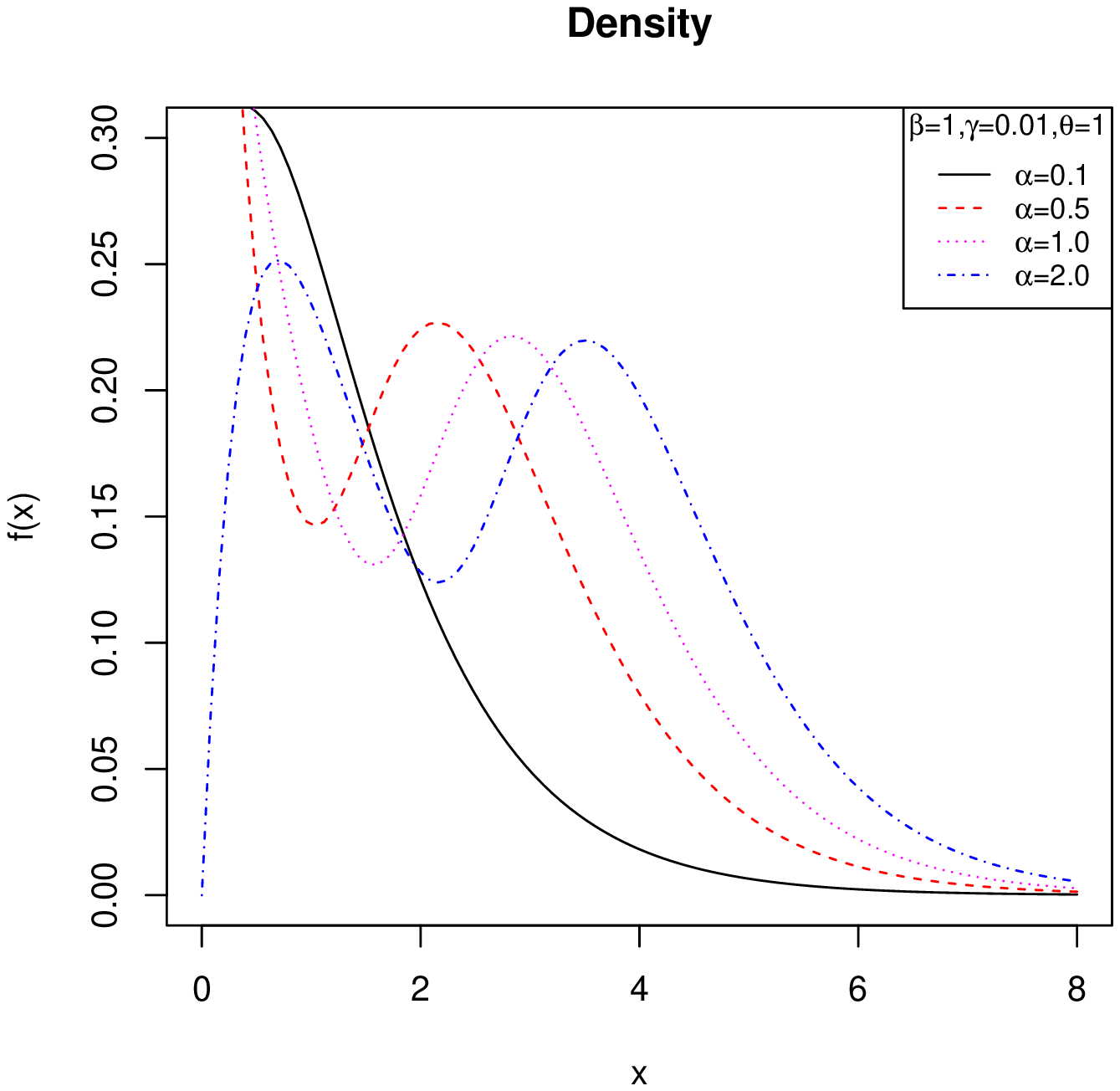}
\includegraphics[width=7.75cm,height=6.5cm]{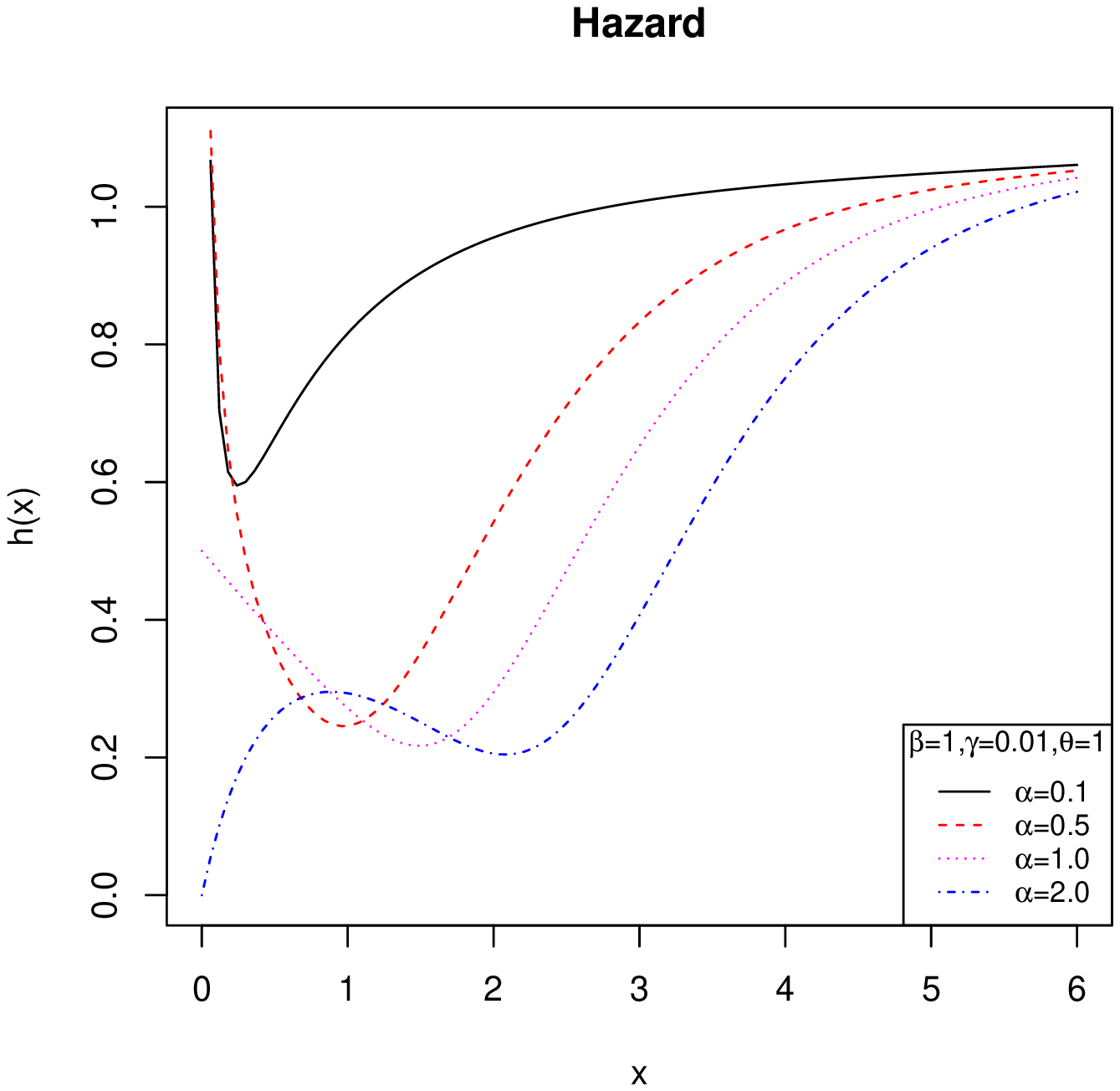}
\vspace{-0.8cm}
\caption{Plots of pdf and hazard rate functions of GGPS with $C\left(\theta \right)=\theta +{\theta }^{20}$. } \label{dhbim}
\end{figure}

\section{Statistical properties}
\label{sec.pro}
In this section, some properties of GGPS distribution such as quantiles, moments, order statistics, Shannon entropy and mean residual life  are obtained.

\subsection{Quantiles and Moments}
The quantile $q$ of GGPS is given by
$$
x_{q}=G^{-1}(\frac{C^{-1}(qC(\theta))}{\theta}),\;\;\;\;\;0<q<1,
$$
where $G^{-1}(y)=\frac{1}{\gamma}\log[1-\frac{\gamma \log(1-y^{\frac{1}{\gamma}})}{\beta}]$ and $C^{-1}(.)$ is the inverse function of $C(.)$.
This  result helps in simulating data from the GGPS distribution with generating  uniform distribution data.

For checking the consistency of the simulating data set form GGPS distribution, the histogram for a generated data set with size 100
and the exact pdf of GGPS   with $C\left(\theta \right)=\theta +{\theta }^{20}$, and  parameters $\alpha=2$, $\beta=1$, $\gamma=0.01$, $\theta=1.0$,  are displayed in Figure \ref{Fig.gd}
(left). Also, the empirical cdf and the exact cdf are given in Figure \ref{Fig.gd} (right).

\begin{figure}
\centering
\includegraphics[width=7.75cm,height=6.5cm]{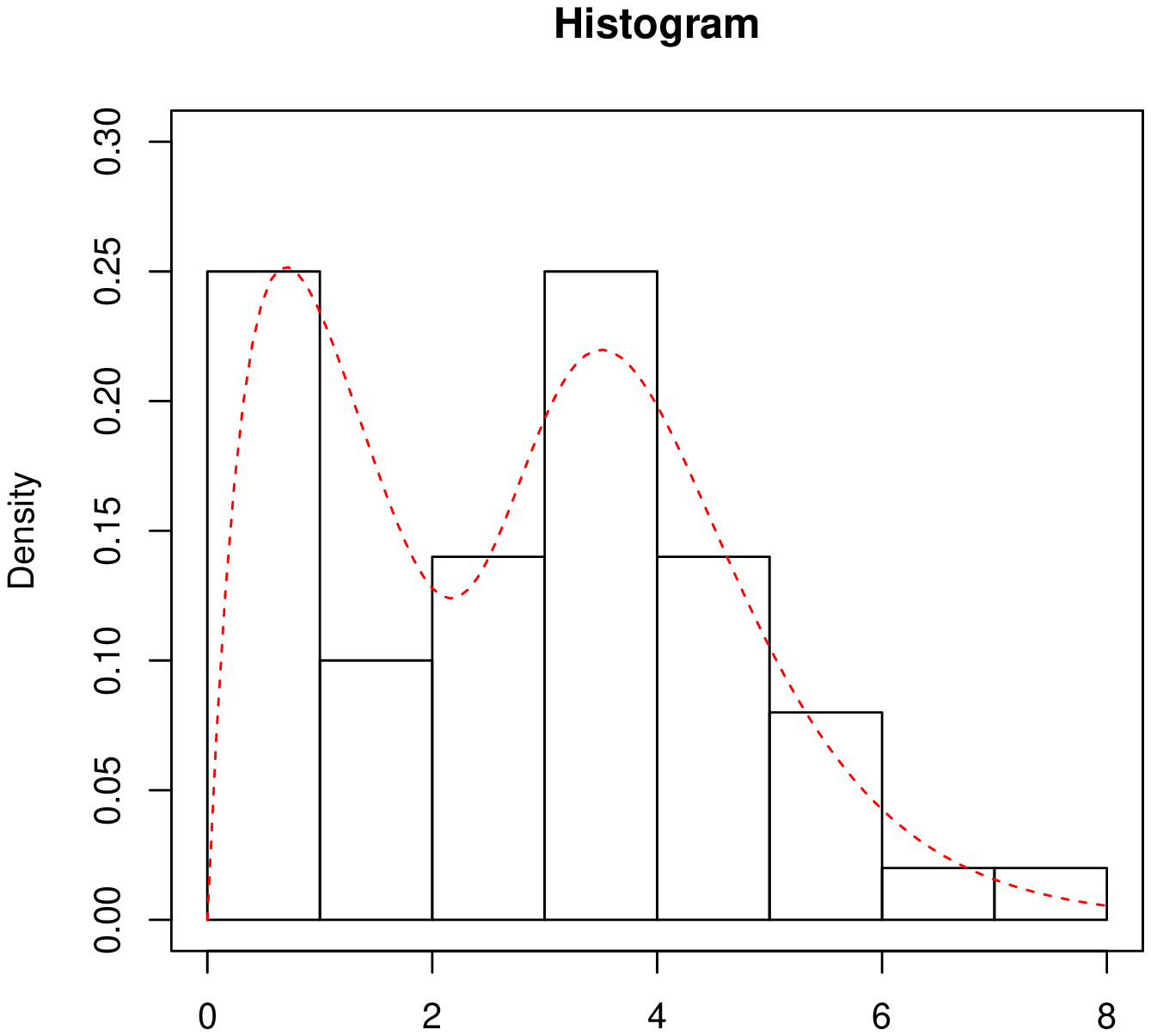}
\includegraphics[width=7.75cm,height=6.5cm]{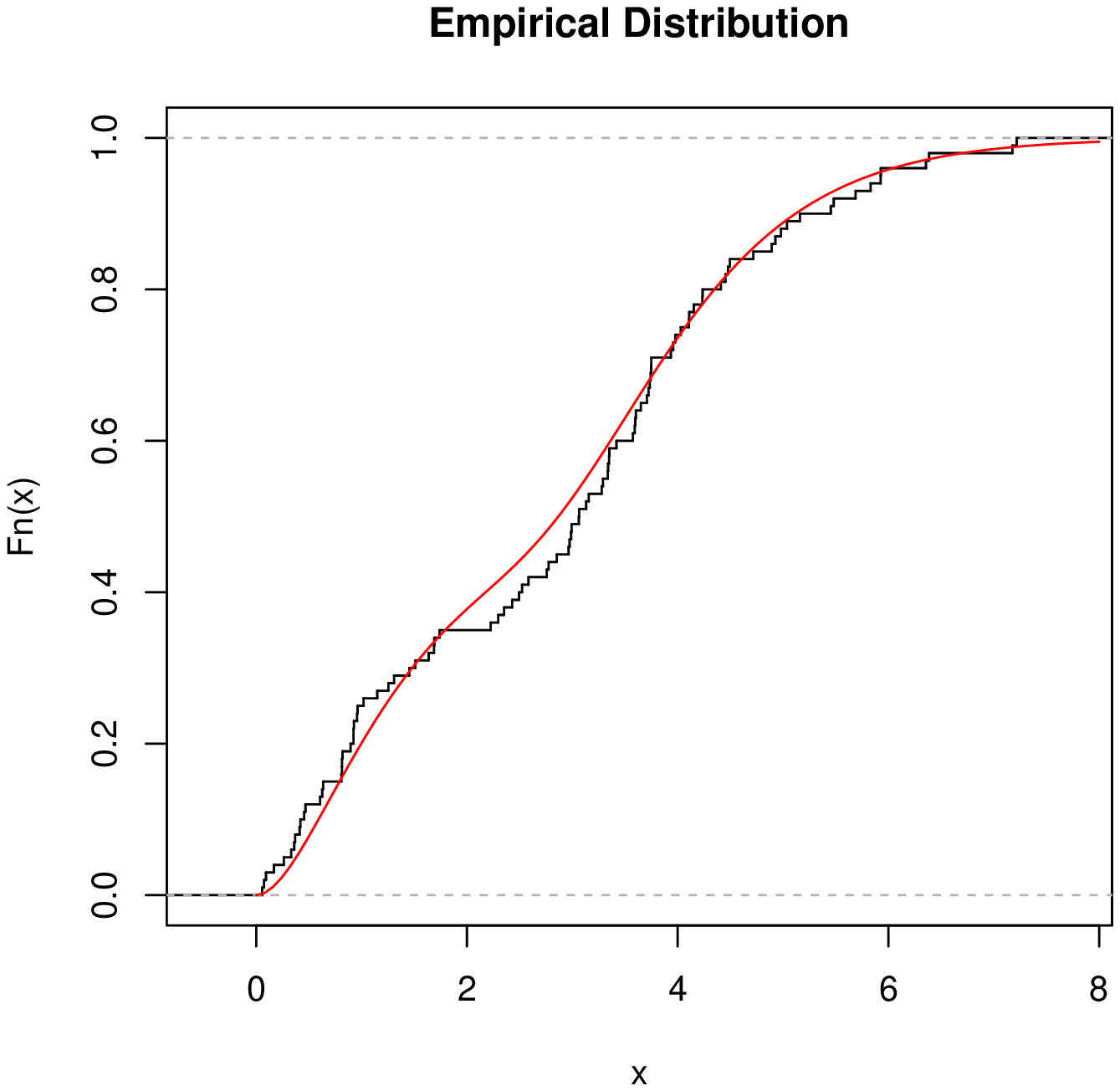}
\vspace{-0.8cm}
\caption{The histogram of a generated data set with the  pdf (left) and the empirical
 cdf with cdf (right) of GGPS distribution. } \label{Fig.gd}
\end{figure}

Consider $X\sim {\rm GGPS}(\alpha,\beta,\gamma,\theta)$. Then the Laplace transform of the GGPS class can be expressed as
\begin{eqnarray}
L(s)=E(e^{-sX})=\sum_{n=1}^{\infty} P(N=n)L_{n}(s),
\end{eqnarray}
where $L_{n}(s)$ is the Laplace transform of ${\rm GG}(n\alpha,\beta,\gamma)$ distribution given as
\begin{eqnarray}
L_{n}(s)&=&\int_{0}^{+\infty} e^{-sx}n \alpha \beta e^{\gamma x}e^{-\frac{\beta}{\gamma}(e^{\gamma x}-1)}[1-e^{\frac{-\beta}{\gamma}(e^{\gamma x}-1)}]^{n\alpha-1}dx \nonumber\\
&=&n\alpha\beta\int_{0}^{+\infty} e^{(\gamma-s)x}e^{-\frac{\beta}{\gamma}(e^{\gamma x}-1)}\sum_{j=0}^{\infty}\binom{n\alpha-1}{j}(-1)^{j}e^{\frac{-\beta}{\gamma}j(e^{\gamma x}-1)}dx \nonumber
\\
&=&n\alpha\beta\sum_{j=0}^{\infty}\binom{n\alpha-1}{j}(-1)^{j}e^{\frac{\beta}{\gamma}(j+1)} \int_{0}^{+\infty} e^{(\gamma-s)x}e^{\frac{-\beta}{\gamma}(j+1)e^{\gamma x}}dx\nonumber\\
&=&n\alpha\beta\sum_{j=0}^{\infty}\binom{n\alpha-1}{j}(-1)^{j}e^{\frac{\beta}{\gamma}(j+1)} \int_{0}^{+\infty} e^{(\gamma-s)x}\sum_{k=0}^{\infty}\frac{(-1)^{k}(\frac{\beta}{\gamma}(j+1))^{k}e^{\gamma kx}}{\Gamma(k+1)}dx\nonumber\\
&=&n\alpha\beta\sum_{j=0}^{\infty}\sum_{k=0}^{\infty}\binom{n\alpha-1}{j}\frac{(-1)^{j+k}e^{\frac{\beta}{\gamma}(j+1)}
[\frac{\beta}{\gamma}(j+1)]^{k}}{\Gamma(k+1)(s-\gamma-\gamma k)},\;\; \;s>\gamma.
\end{eqnarray}
Now, we obtain the moment generating function of GGPS.
\begin{eqnarray}
M_{X}(t)&=&E(e^{tX})
=\sum_{n=1}^{\infty} P(N=n)L_{n}(-t)\nonumber\\
&=&\alpha\beta \sum_{n=1}^{\infty}
\frac{a_{n}\theta^{n}}{C(\theta)}\sum_{k=0}^{\infty}\sum_{j=0}^{\infty}\frac{n\binom{n\alpha-1}{j}(-1)^{j+k+1}e^{\frac{\beta}{\gamma}(j+1)}(\frac{\beta}{\gamma}(j+1))^{k}}{\Gamma(k+1)(t+\gamma+\gamma k)}\nonumber\\
&=&\alpha\beta E_N[\sum_{k=0}^{\infty}\sum_{j=0}^{\infty}\frac{N \binom{N\alpha-1}{j}(-1)^{j+k+1}e^{\frac{\beta}{\gamma}(j+1)}(\frac{\beta}{\gamma}(j+1))^{k}}{\Gamma(k+1)(t+\gamma+\gamma k)}],
\end{eqnarray}
where $N$ is a random variable from the power series family with the probability mass function in \eqref{eq.ps} and $E_N[U]$ is expectation of $U$ with respect to random variable $N$.

We can use $M_{X}(t)$ to obtain the non-central moments, $\mu_{r}=E[X^{r} ]$. But from the direct calculation, we have
\begin{eqnarray}\label{eq.mur}
\mu_{r}
&=&\sum_{n=1}^{\infty}\frac{a_{n}\theta^{n}}{C(\theta)}\sum_{k=0}^{\infty}\sum_{j=0}^{\infty}\frac{n\alpha\beta \binom{n\alpha-1}{j}(-1)^{j+k+r+1}e^{\frac{\beta}{\gamma}(j+1)}(\frac{\beta}{\gamma}(j+1))^{k}\Gamma(r+1)}{\Gamma(k+1)(\gamma+\gamma k)^{r+1}}\nonumber\\
&=&\alpha\beta E_N[\sum_{k=0}^{\infty}\sum_{j=0}^{\infty}\frac{N \binom{N\alpha-1}{j}(-1)^{j+k+r+1}e^{\frac{\beta}{\gamma}(j+1)}(\frac{\beta}{\gamma}(j+1))^{k}\Gamma(r+1)}{\Gamma(k+1)(\gamma+\gamma k)^{r+1}}].
\end{eqnarray}

\begin{prop}
For non-central moment function in \ref{eq.mur}, we have
$$
\lim_{\theta \rightarrow 0^{+}}\mu_{r}=E[Y^{r}],
$$
where $Y$ has ${\rm GG}(c\alpha,\beta,\gamma)$ and
$c=\min\{n\in \mathbb{N}: a_{n}>0\}$.

\end{prop}
\begin{proof}
If $Y$ has ${\rm GG}(c\alpha,\beta,\gamma)$, then
$$
E[Y^{r}]=\sum\limits_{k=0}^{\infty}\sum\limits_{j=0}^{\infty}\frac{c\alpha\beta \binom{c\alpha-1}{j}(-1)^{j+k+r+1}e^{\frac{\beta}{\gamma}(j+1)}(\frac{\beta}{\gamma}(j+1))^{k}\Gamma(r+1)}{\Gamma(k+1)(\gamma+\gamma k)^{r+1}}.
$$
Therefore,
\begin{eqnarray*}
\lim_{\theta \rightarrow 0^{+}}\mu_{r}&=&\lim_{\theta \rightarrow 0^{+}}\frac{\sum\limits_{n=1}^{\infty}a_{n}\theta^{n}E[Y^{r}]}{\sum\limits_{n=1}^{\infty}a_{n}\theta^{n}}\\
&=&\lim_{\theta \rightarrow 0^{+}}\frac{a_{c}E[Y^{r}]+\sum\limits_{n=c+1}^{\infty}a_{n}\theta^{n-c}E[Y^{r}]}{a_{c}+\sum\limits_{n=c+1}^{\infty}a_{n}\theta^{n-c}}\\
&=&E[Y^{r}].
\end{eqnarray*}
\end{proof}

\subsection{Order statistic}

Let $X_{1},X_{2},\dots,X_{n}$  be an independent random sample of size $n$ from ${\rm GGPS}(\alpha,\beta,\gamma,\theta)$. Then, the pdf of the $i$th order statistic, say $X_{i:n}$, is given by
$$
f_{i:n}(x)=\frac{n!}{(i-1)!(n-i)!}f(x)[\frac{C(\theta t^{\alpha})}{C(\theta)}]^{i-1}[1-\frac{C(\theta t^{\alpha})}{C(\theta)}]^{n-i},
$$
where $f$ is the pdf given in (\ref{fGP}) and $t=1-e^{-\frac{\beta}{\gamma}(e^{\gamma x}-1)}$. Also, the cdf of $X_{i:n}$ is given by
$$
F_{i:n}(x)=\frac{n!}{(i-1)!(n-i)!}\sum_{k=0}^{n-i}\frac{(-1)^{k} \binom{n-i}{k}}{k+i+1}[\frac{C(t^{\alpha})}{C(\theta)}]^{k+i}.
$$
An analytical expression for  $r$th non-central moment of  order statistics $X_{i:n}$ is obtained as
\begin{eqnarray*}
E[X_{i:n}^{r}]&=& r\sum_{k=n-i+1}^{n}(-1)^{k-n+i-1}\binom{k-1}{n-i}\binom{n}{k}\int_{0}^{+\infty}x^{r-1}S(x)^{k}dx\\
&=& r\sum_{k=n-i+1}^{n}\frac{(-1)^{k-n+i-1}}{[C(\theta)]^{k}}\binom{k-1}{n-i}\binom{n}{k}\int_{0}^{+\infty}x^{r-1}[C(\theta)-C(\theta t^{\alpha})]^{k}dx,
\end{eqnarray*}
where $S(x)=1-F(x)$ is the survival function of GGPS distribution.

\subsection{Shannon entropy and mean residual life }
 If $X$ is a none-negative continuous random variable with pdf $f$, then Shannon's entropy of $X$ is defined  by
 \cite{shan-48}
  as
\begin{equation*}
H(f)=E[-\log f(X)]=-\int_{0}^{+\infty} f(x)\log (f(x))dx,
\end{equation*}
and this  is usually referred to as the continuous entropy (or differential entropy). An explicit expression of
Shannon entropy for GGPS  distribution is obtained as
\begin{eqnarray}
H(f)&=&E\{-\log[\frac{\theta \alpha \beta}{C(\theta)}e^{\gamma X}(e^{-\frac{\beta}{\gamma}(e^{\gamma X}-1)})(1-e^{-\frac{\beta}{\gamma}(e^{\gamma X}-1)})^{\alpha-1}C'\left(\theta (1-e^{-\frac{\beta}{\gamma}(e^{\gamma X}-1)})^{\alpha}\right)]\}\nonumber\\
&=&-\log[\frac{\theta\beta\alpha}{C(\theta)}]-\gamma E(X)+\frac{\beta}{\gamma}E(e^{\gamma X})-\frac{\beta}{\gamma}\nonumber\\
&&-(\alpha-1)E[\log(1-e^{-\frac{\beta}{\gamma}(e^{\gamma X}-1)})]-E[\log(C'\left(\theta (1-e^{-\frac{\beta}{\gamma}(e^{\gamma X}-1)})^{\alpha}\right))]\nonumber\\
&=&-\log[\frac{\theta\beta\alpha}{C(\theta)}]-\gamma\mu_{1}+\frac{\beta}{\gamma}M_{X}(\gamma)-\frac{\beta}{\gamma}
-(\alpha-1)\sum_{n=1}^{\infty}P(N=n)\int_{0}^{1}\ n\alpha t^{n\alpha-1}\log(t) dt\nonumber\\ &&-\sum_{n=1}^{\infty}P(N=n)\int_{0}^{1}nu^{n-1}\log(C'(\theta u))du\nonumber\\
&=&-\log[\frac{\theta\beta\alpha}{C(\theta)}]-\gamma\mu_{1}+\frac{\beta}{\gamma}M_{X}(\gamma)-\frac{\beta}{\gamma}
+\frac{(\alpha-1)}{\alpha}E_{N}[\frac{1}{N}] -E_{N}[A(N,\theta)],
\end{eqnarray}
where $A(N,\theta)=\int_{0}^{1}Nu^{N-1}\log(C'(\theta u))du$, $N$ is a random variable from the power series family with the probability mass function in \eqref{eq.ps}, and $E_N[U]$ is expectation of $U$ with respect to random variable $N$. In reliability theory and survival analysis, $X$ usually denotes a duration such as the lifetime.
The residual lifetime of the system when it is still operating at time $s$, is $X_{s} =X-s\mid X>s$ which has pdf
$$f(x;s)=\frac{f(x)}{1-F(s)}=\frac{\theta g(x)C'(\theta G(x))}{C(\theta)-C(\theta G(s))}, \;\; x\geq s>0.  $$
 Also, the mean residual lifetime  of $X_{s}$ is given by
\begin{eqnarray*}
m(s)=E[X-s|X>s]&=&\frac{\int_{s}^{+\infty}(x-s)f(x)dx}{1-F(s)}\\
&=&\frac{\int_{s}^{+\infty}x f(x)dx}{1-F(s)}-s\\
&=& \frac{C(\theta)E_N[Z(s,N)]}{C(\theta)-C(\theta[1-e^{-\frac{\beta}{\gamma}(e^{\gamma s}-1)}]^{\alpha})}-s,
\end{eqnarray*}
where $Z(s,n)=\int_{s}^{+\infty}x {\rm g}_{(n)}(x;n\alpha,\beta,\gamma) dx$,  and ${\rm g}_{(n)}(x;n\alpha,\beta,\gamma)$ is the pdf of ${\rm GG}(n\alpha,\beta,\gamma)$.

\section{Special cases of GGPS distribution}
\label{sec.spe}
In this Section, we consider four special cases of the GGPS distribution. To simplify, we consider
$t=1-e^{-\frac{\beta}{\gamma}(e^{\gamma x}-1)}$, $x>0$, and $A_j=\binom{n\alpha-1}{j}$.

\subsection{Generalized  Gompertz-geometric distribution}
The geometric distribution (truncated at zero) is a special case of power series distributions with $a_{n}=1$ and $C(\theta)=\frac{\theta}{1-\theta} \ (0<\theta<1)$. The  pdf and hazard rate function of generalized  Gompertz-geometric (GGG) distribution is given respectively by
\begin{eqnarray}
f(x)&=& \frac{(1-\theta)\alpha \beta e^{\gamma x}(1-t)t^{\alpha-1}}{(\theta t^{\alpha}-1)^{2}}, \ \ \ x>0, \label{eq.fGG}\\
h(x)&=&\frac{(1-\theta)\alpha \beta e^{\gamma x}(1-t)t^{\alpha-1}}{(1-\theta t^{\alpha})(1-t^{\alpha})}, \ \ \ x>0. \label{eq.hGG}
\end{eqnarray}

\begin{remark} Consider
\begin{eqnarray}\label{eq.fGGe}
f_M(x)= \frac{\theta^* \alpha \beta e^{\gamma x}(1-t)t^{\alpha-1}}{((1-\theta^*)t^{\alpha}-1)^{2}}, \ \ \ x>0,
\end{eqnarray}
where $\theta^{\ast}=1-\theta$. Then $f_M(x)$  is pdf for all $\theta^*>0$ \citep[see][]{ma-ol-97}.
Note that when $\alpha=1$ and $\gamma\rightarrow 0^+$, the pdf of extended exponential geometric (EEG) distribution
\citep{ad-di-05}
 is concluded from \eqref{eq.fGGe}.
  The EEG hazard function is monotonically increasing for $\theta^*>1$; decreasing for $0<\theta^*<1$ and constant for $\theta^*=1$.
\end{remark}

\begin{remark}
If $\alpha=\theta^*=1$, then the pdf in  \eqref{eq.fGGe} becomes the pdf of Gompertz distribution. Note that the hazard rate function of Gompertz distribution is $h(x)=\beta e^{\gamma x}$ which is increasing.
\end{remark}

The plots of pdf and hazard rate function of GGG for different values of $\alpha$, $\beta$, $\gamma$ and $\theta^*$  are given in Figure \ref{fig.GG}.

\begin{theorem}  Consider the GGG hazard function in
\eqref{eq.hGG}. Then, for $\alpha\geq 1$, the hazard function is increasing and for $0<\alpha<1$, is
decreasing and bathtub shaped.
\end{theorem}

\begin{proof} See Appendix A.1.
\end{proof}

The first and second non-central moments of GGG are given by
\begin{eqnarray*}
&&E(X)=\alpha\beta (1-\theta)\sum_{n=1}^{\infty} n\theta^{n-1}\sum_{k=0}^{\infty}\sum_{j=0}^{\infty}\frac{ A_j(-1)^{j+k}e^{\frac{\beta}{\gamma}(j+1)}(\frac{\beta}{\gamma}(j+1))^{k}}{\Gamma(k+1)(\gamma+\gamma k)^{2}},
\\
&&E(X^{2})=2\alpha\beta (1-\theta)\sum_{n=1}^{\infty} n\theta^{n-1}\sum_{k=0}^{\infty}\sum_{j=0}^{\infty}\frac{ A_j(-1)^{j+k+3}e^{\frac{\beta}{\gamma}(j+1)}(\frac{\beta}{\gamma}(j+1))^{k}}{\Gamma(k+1)(\gamma+\gamma k)^{3}}.
\end{eqnarray*}

\begin{figure}
\centering
\includegraphics[scale=0.35]{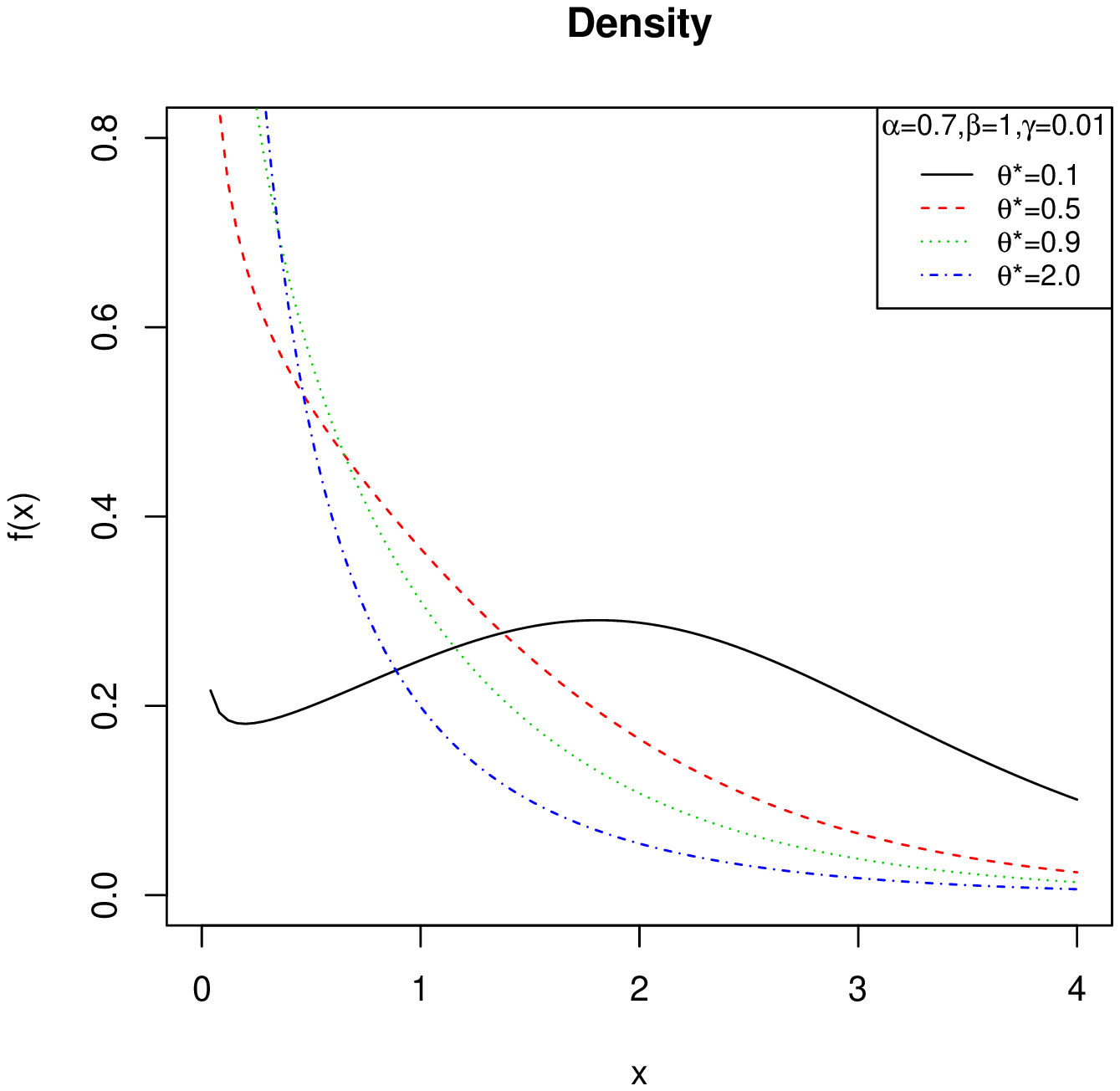}
\includegraphics[scale=0.35]{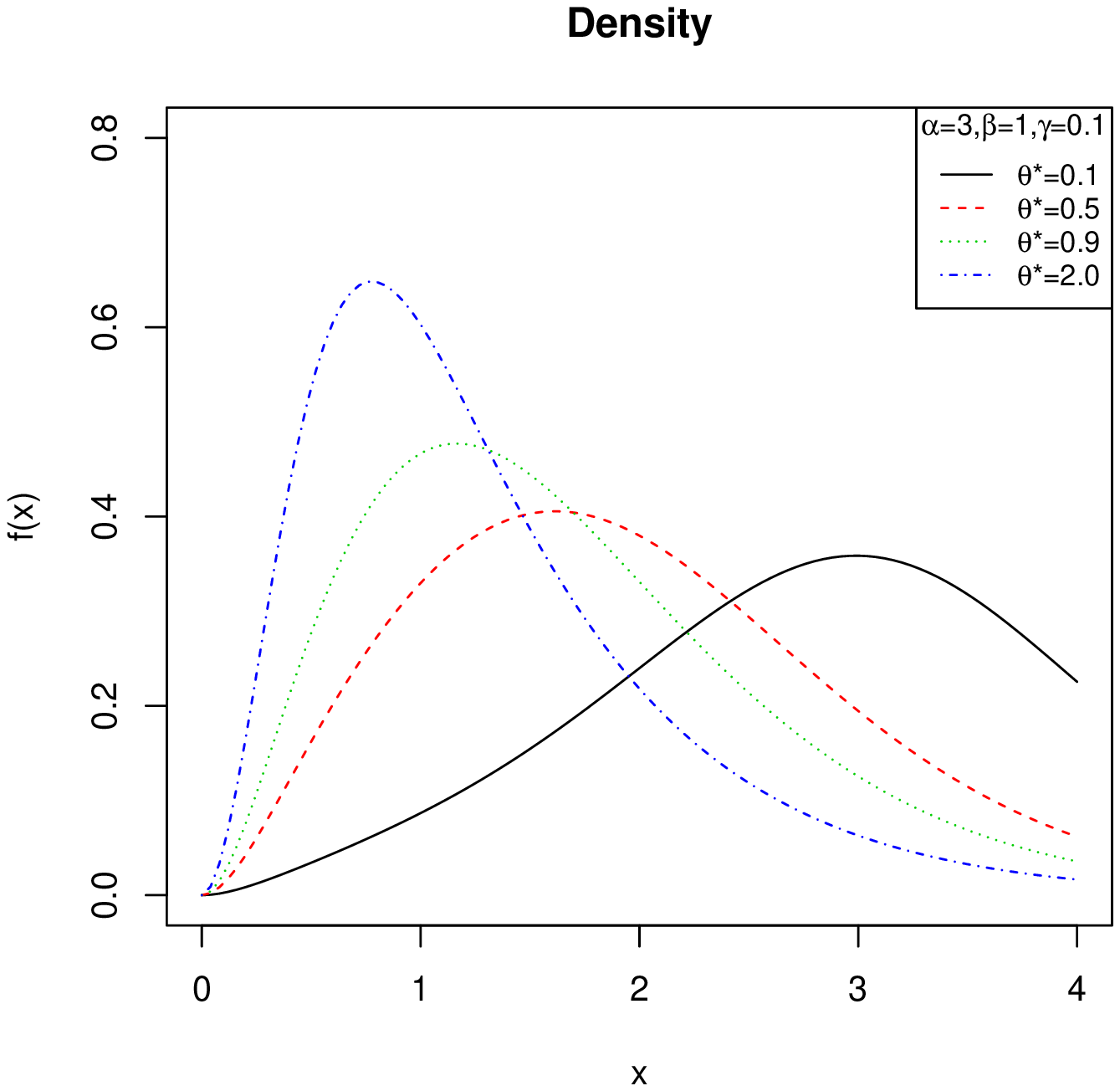}
\includegraphics[scale=0.35]{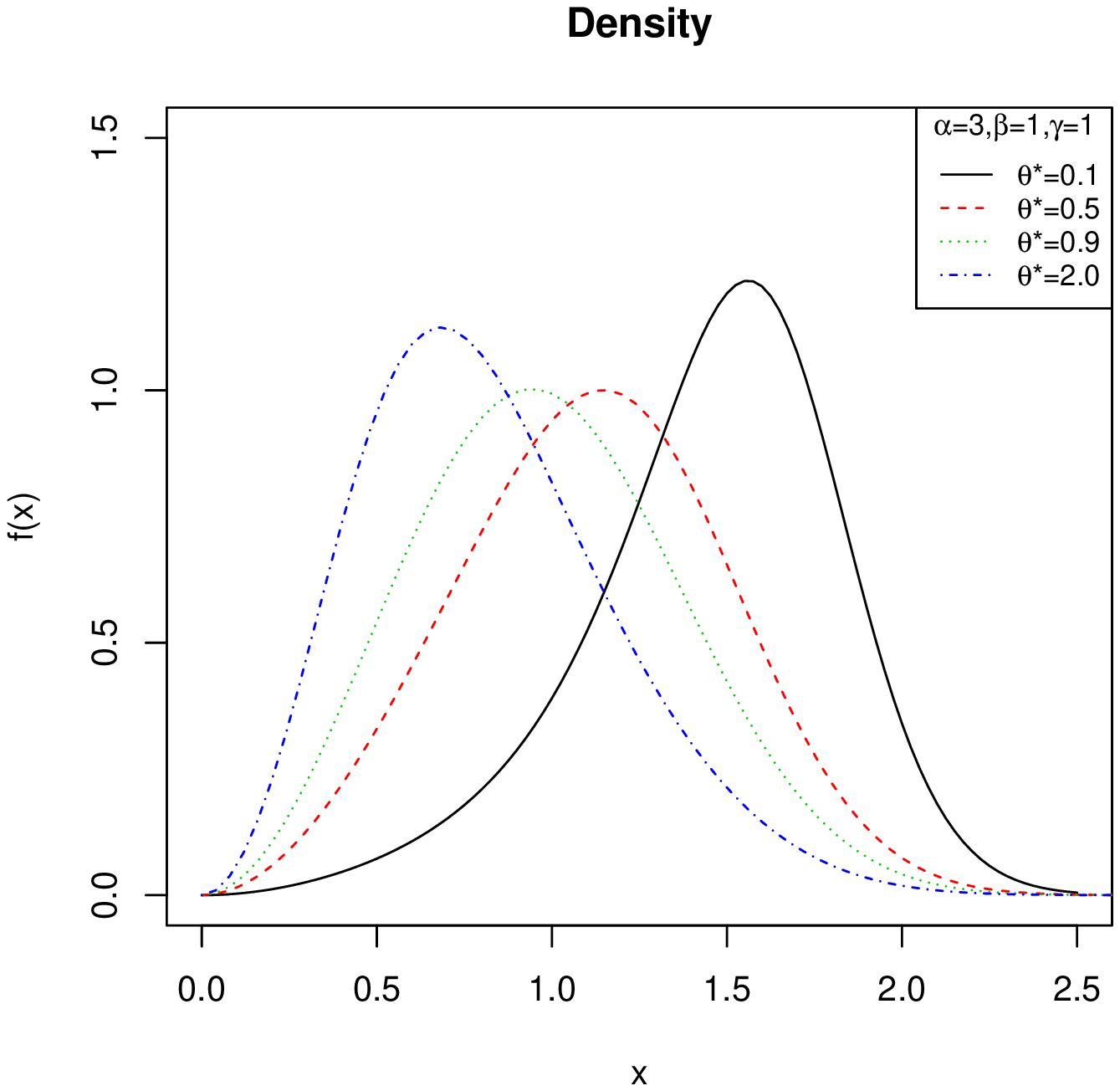}
\includegraphics[scale=0.35]{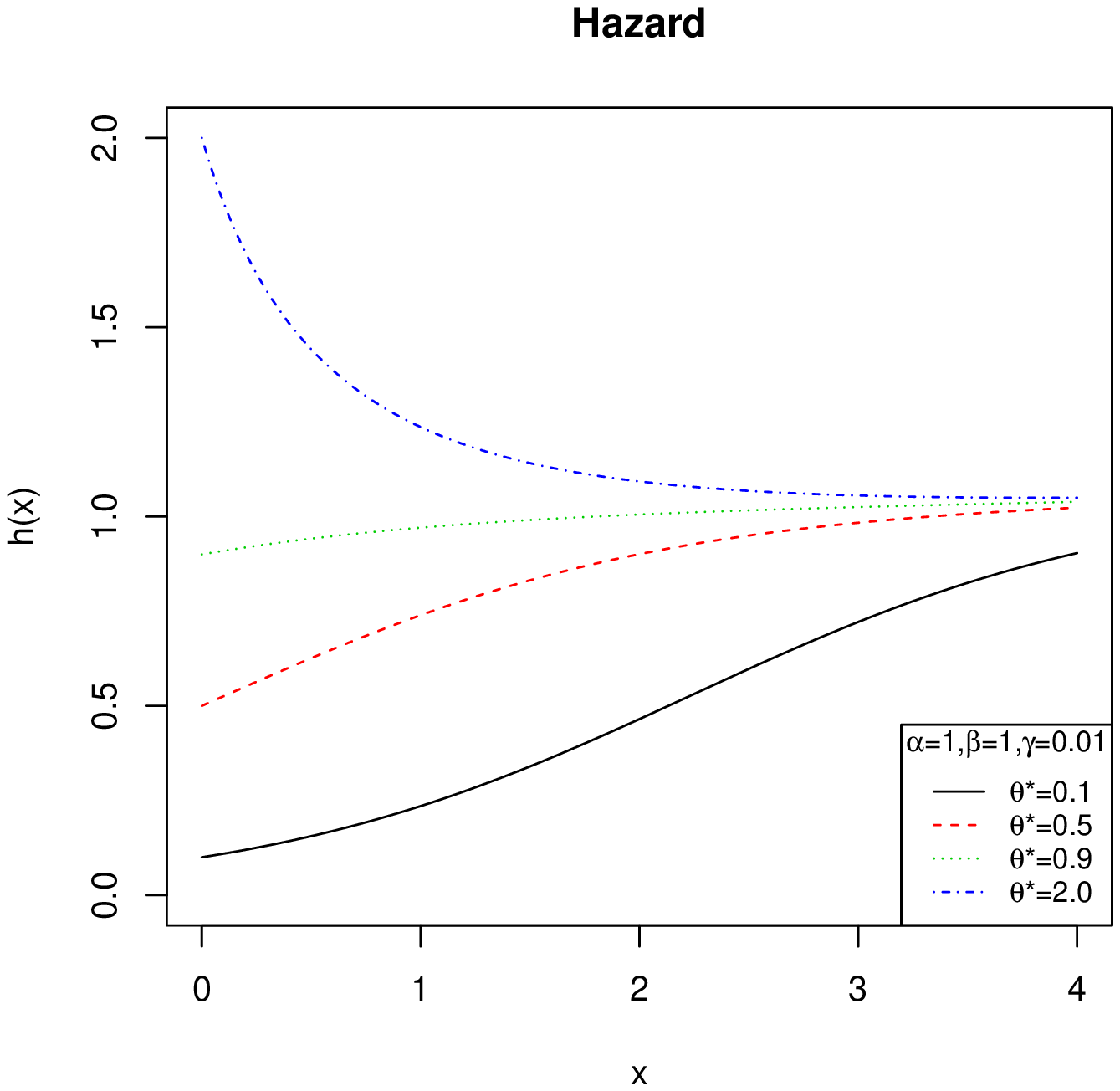}
\includegraphics[scale=0.35]{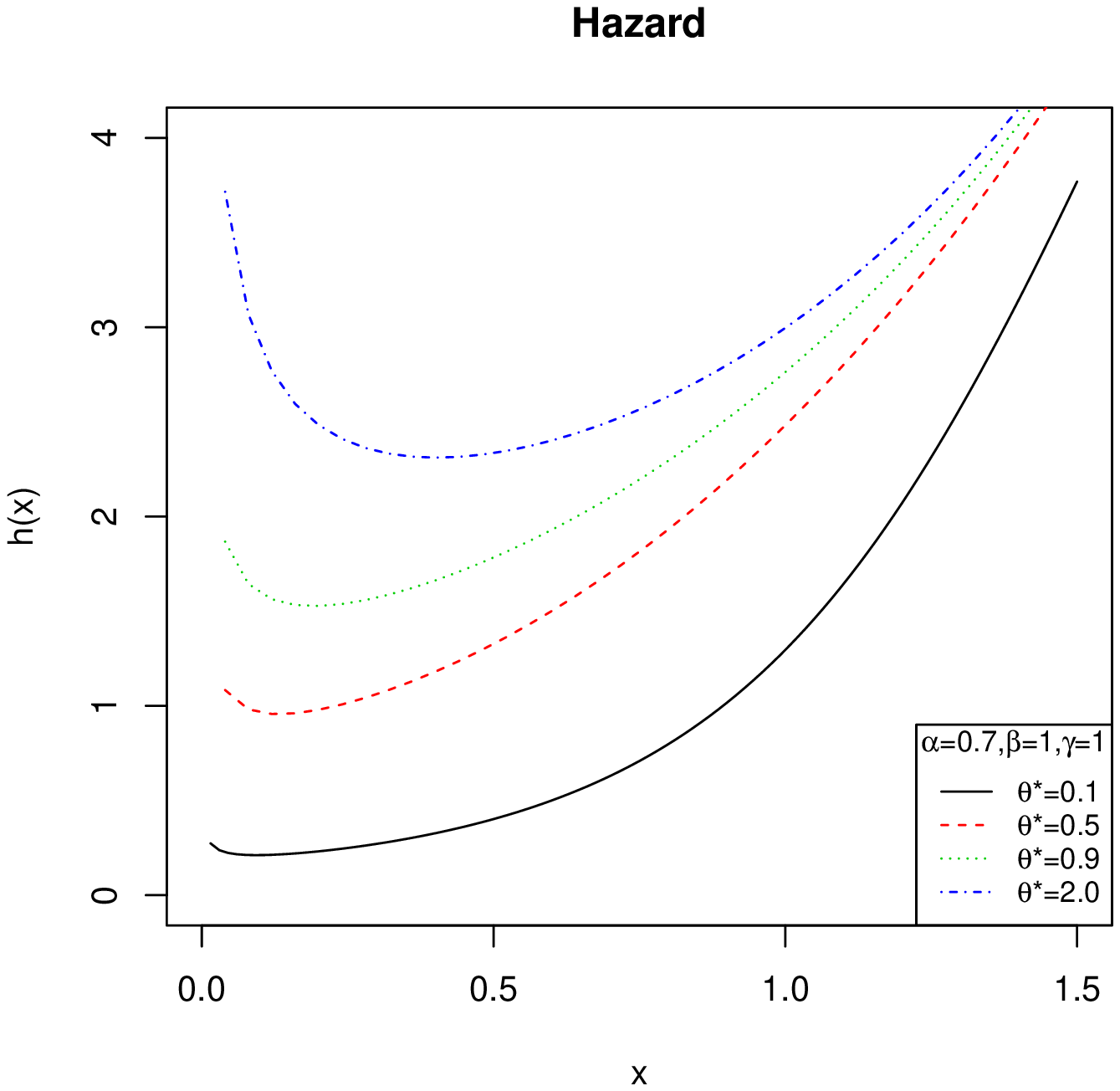}
\includegraphics[scale=0.35]{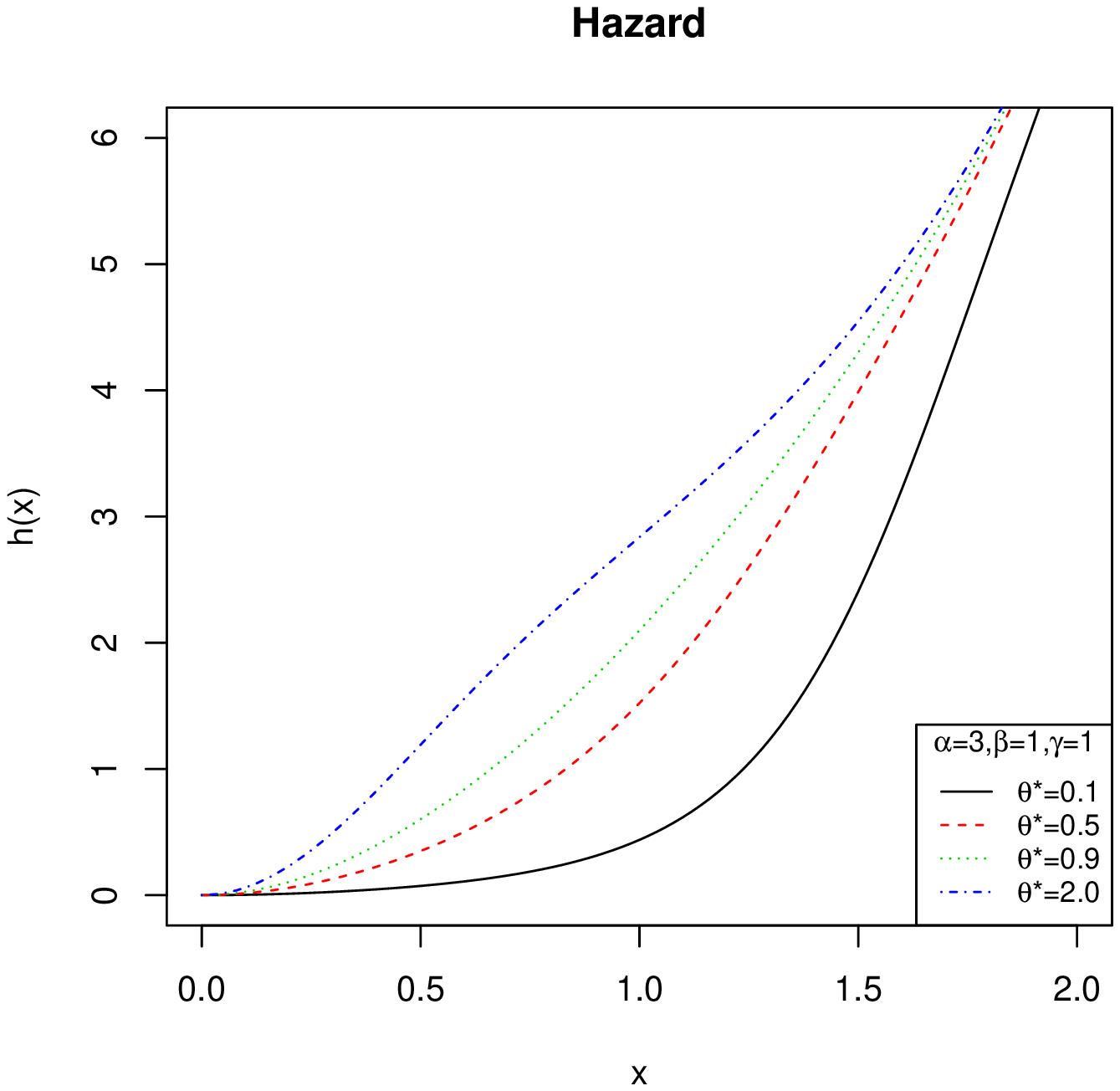}

\vspace{-0.8cm}
\caption{Plots of pdf and hazard rate function of GGG for  different values $\alpha$, $\beta $, $\gamma$ and $\theta^*$.}\label{fig.GG}
\end{figure}

\subsection{Generalized Gompertz-Poisson distribution}

The Poisson distribution (truncated at zero) is a special case of power series distributions with $a_{n}=\frac{1}{n!}$ and $C(\theta)=e^{\theta}-1$  $(\theta>0)$. The pdf and hazard rate function of generalized  Gompertz-Poisson (GGP) distribution are given respectively by
\begin{eqnarray}
f(x)&=& \theta\alpha \beta e^{\gamma x-\theta}(1-t)t^{\alpha-1} e^{\theta t^{\alpha}}, \ \ \ x>0\\
 \label{eq.hgp}
h(x)&=&\frac{\theta\alpha \beta e^{\gamma x}(1-t)t^{\alpha-1}e^{\theta t^{\alpha}}}{e^{\theta}-e^{\theta t^{\alpha}}}, \ \ \ \ \ x>0.
\end{eqnarray}

\begin{theorem} \label{thm.hgp}  Consider the GGP hazard function in (\ref{eq.hgp}). Then, for $\alpha\geq1$, the hazard function is increasing and for $0<\alpha<1$, is decreasing and bathtub shaped.
\end{theorem}

\begin{proof} See Appendix A.2.
\end{proof}

The first and second non-central moments of GGP can be computed as
\begin{eqnarray*}
&&E(X)=\frac{\alpha\beta}{e^{\theta}-1}\sum_{n=1}^{\infty} \frac{\theta^{n}}{(n-1)!}\sum_{k=0}^{\infty}\sum_{j=0}^{\infty}\frac{ A_j(-1)^{j+k}e^{\frac{\beta}{\gamma}(j+1)}(\frac{\beta}{\gamma}(j+1))^{k}}{\Gamma(k+1)(\gamma+\gamma k)^{2}},\\
&&E(X^{2})=\frac{2\alpha\beta}{e^{\theta}-1}\sum_{n=1}^{\infty}  \frac{\theta^{n}}{(n-1)!}\sum_{k=0}^{\infty}\sum_{j=0}^{\infty}\frac{ A_j(-1)^{j+k+3}e^{\frac{\beta}{\gamma}(j+1)}(\frac{\beta}{\gamma}(j+1))^{k}}{\Gamma(k+1)(\gamma+\gamma k)^{3}}.
\end{eqnarray*}
The plots of pdf and hazard rate function of GGP for different values of $\alpha$, $\beta$, $\gamma$ and $\theta$  are given in Figure \ref{fig.GP}.

\begin{figure}
\centering
\includegraphics[scale=0.35]{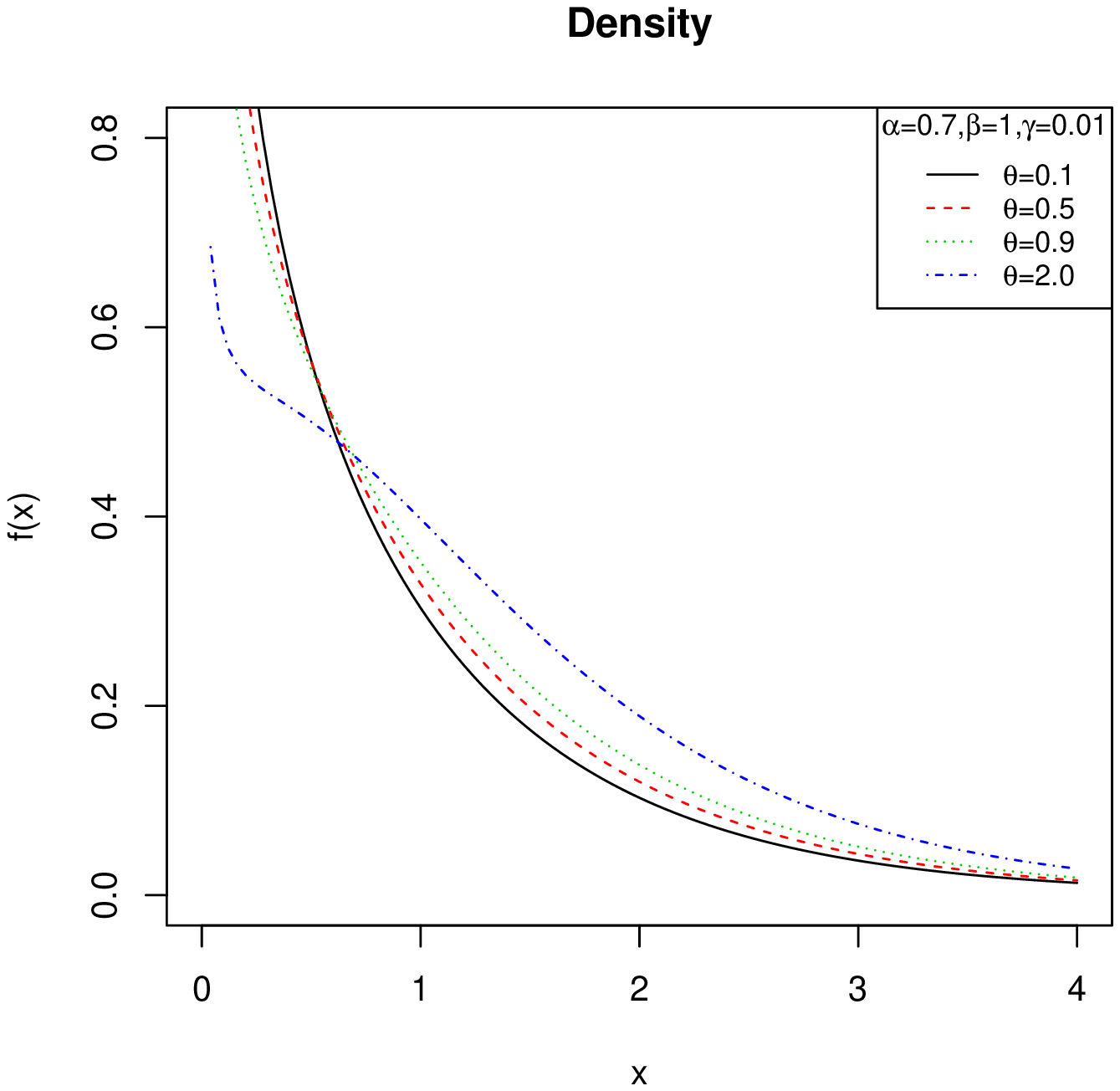}
\includegraphics[scale=0.35]{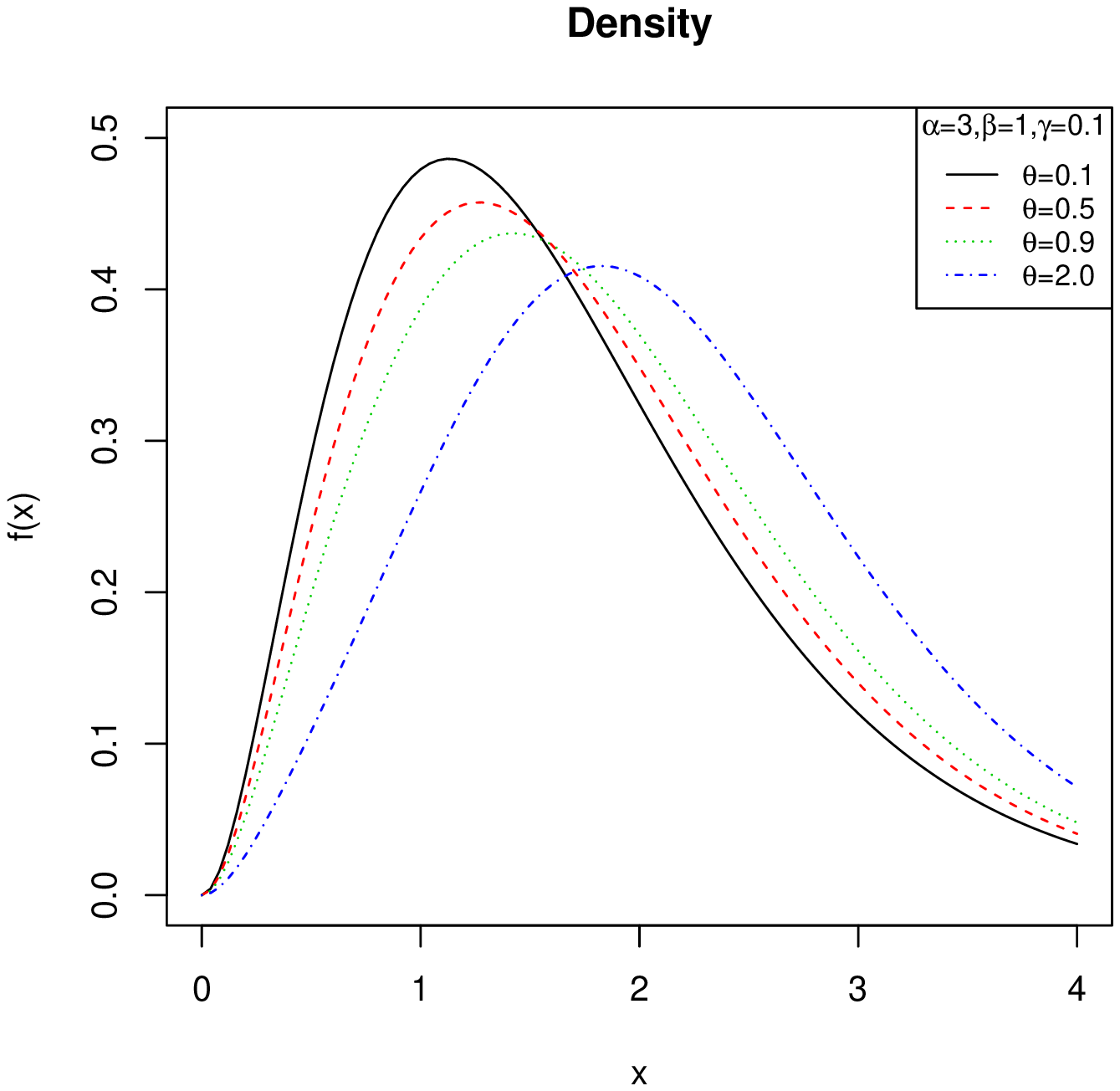}
\includegraphics[scale=0.35]{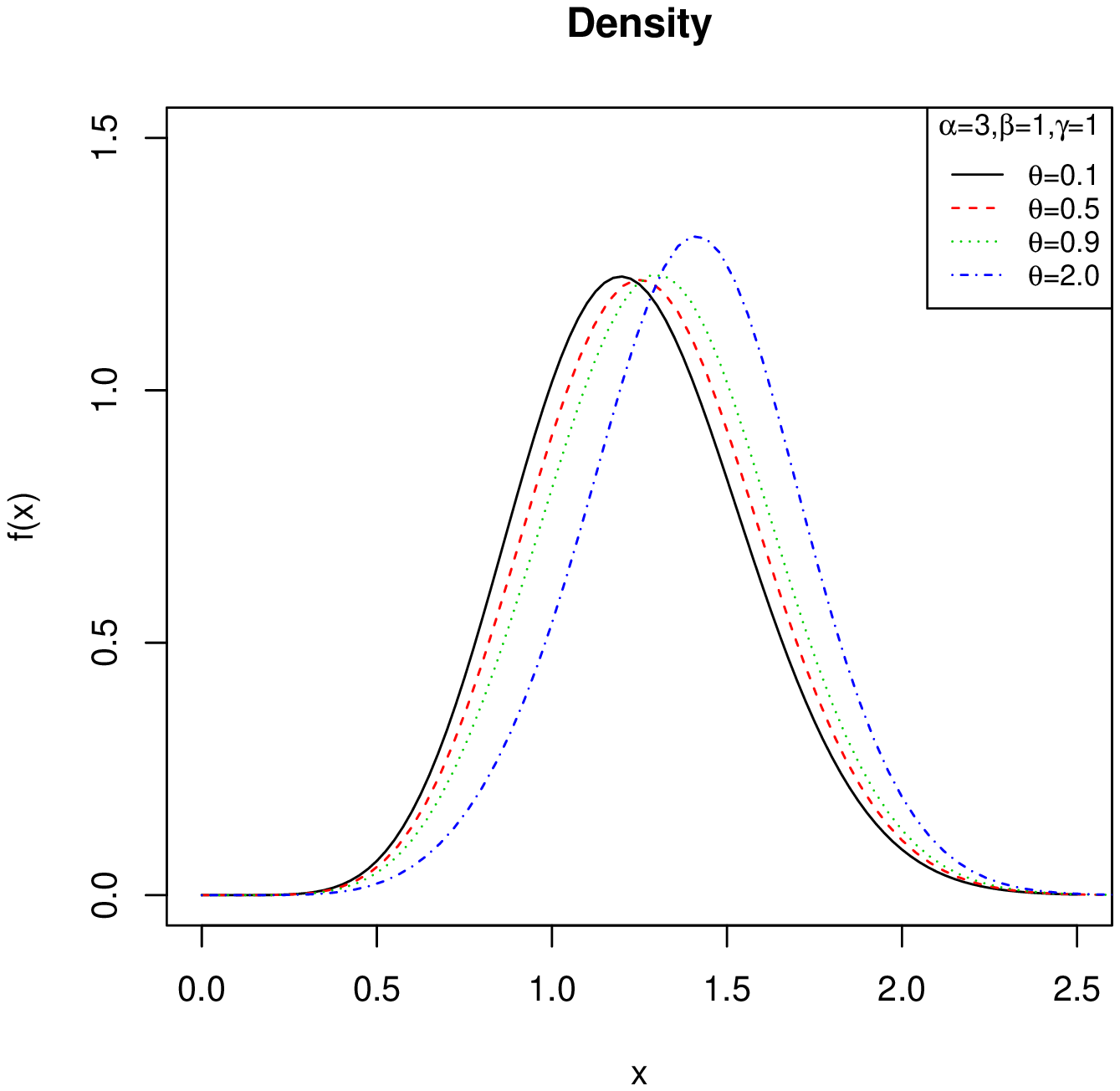}
\includegraphics[scale=0.35]{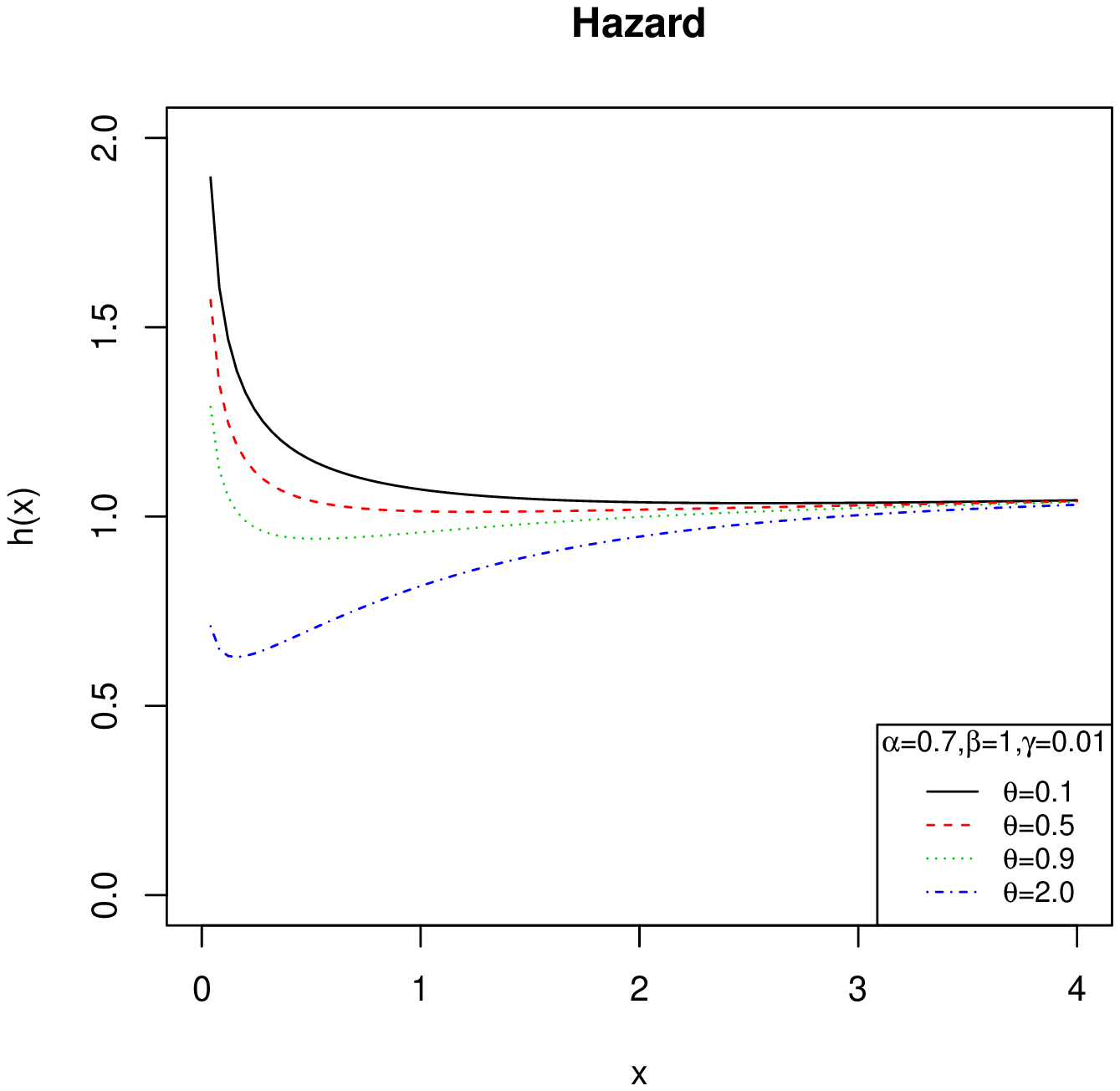}
\includegraphics[scale=0.35]{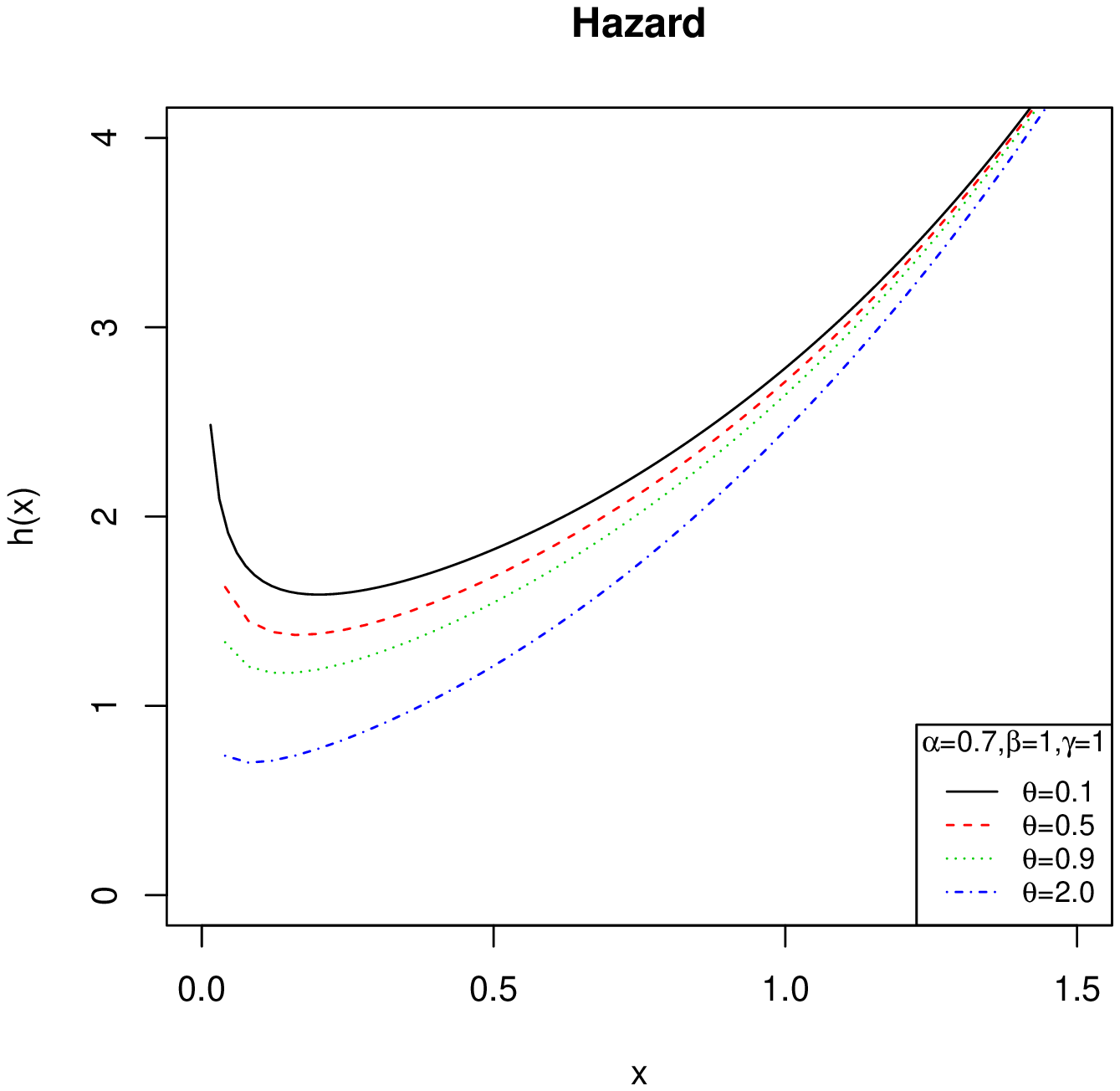}
\includegraphics[scale=0.35]{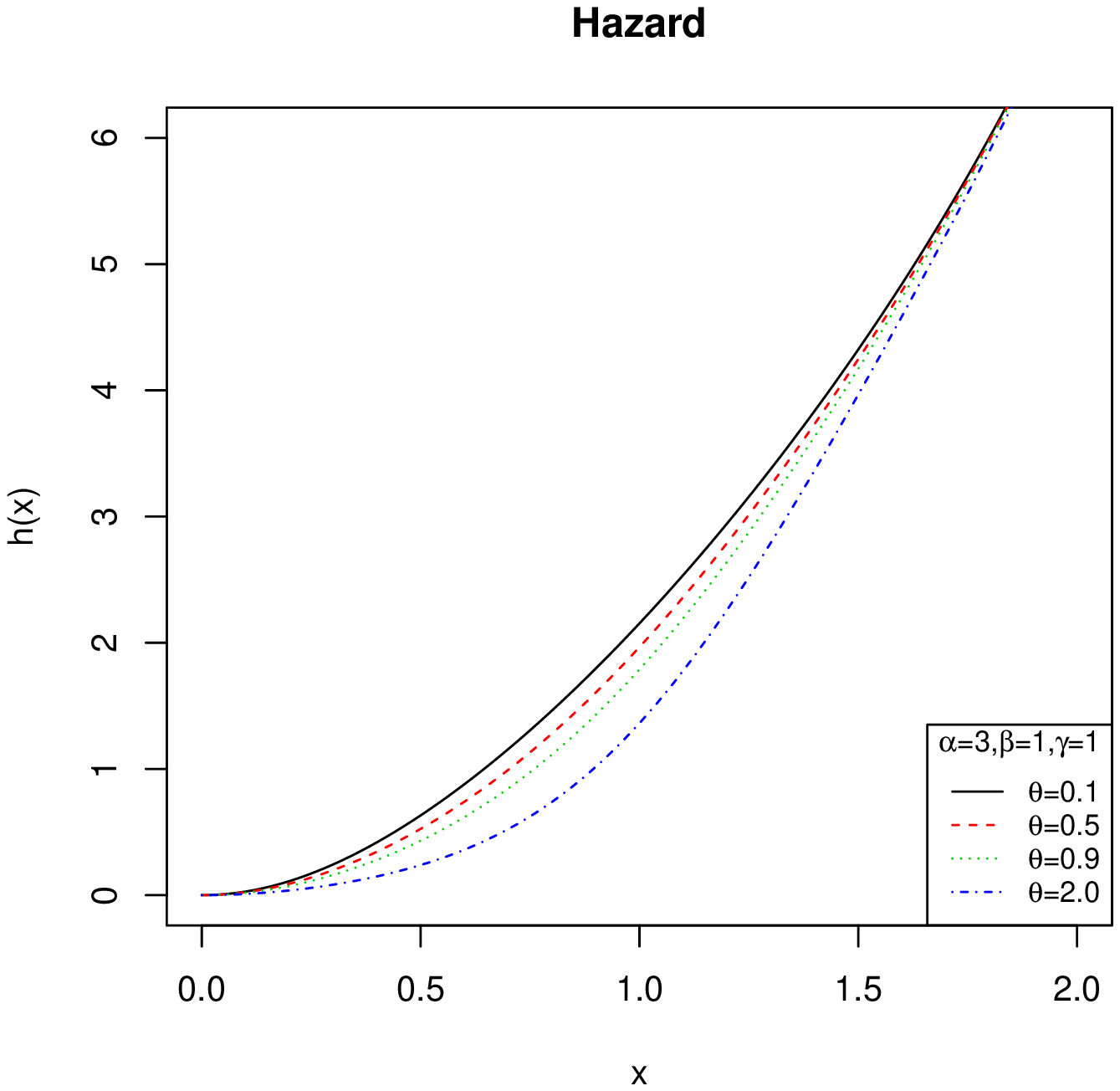}

\vspace{-0.8cm}
\caption{Plots of pdf and hazard rate function of GGP for  different values $\alpha$,  $\beta $, $\gamma$ and $\theta$.}\label{fig.GP}
\end{figure}

\subsection{Generalized Gompertz-binomial distribution}

The binomial distribution (truncated at zero) is a special case of power series distributions with $a_{n}=\binom{m}{n}$ and $C(\theta)=(\theta+1)^{m}-1 \ (\theta>0),$ where $m$ $(n\leq m)$  is the number of replicas. The pdf and hazard rate function of generalized
 Gompertz-binomial (GGB) distribution are given respectively by
\begin{eqnarray}
f(x)&=&m \theta\alpha \beta e^{\gamma x}(1-t)t^{\alpha-1}\frac{(\theta t^{\alpha}+1)^{m-1}}{(\theta+1)^{m}-1}, \ \ \ x>0,\\
 \label{eq.hgb}
h(x)&=&\frac{m \theta\alpha \beta e^{\gamma x}(1-t)t^{\alpha-1} (\theta t^{\alpha}+1)^{m-1}}{(\theta+1)^{m}-(\theta t^{\alpha}+1)^{m}}, \ \ \ x>0.
\end{eqnarray}

The plots of pdf and hazard rate function of GGB for $m=4$, and different values of $\alpha$,  $\beta$, $\gamma$ and $\theta$  are given in Figure \ref{fig.GB}. We can find that the GGP distribution can be obtained as limiting of GGB distribution if $m\theta\rightarrow \lambda>0$, when $m\rightarrow \infty$.

\begin{theorem}  Consider the GGB hazard function in (\ref{eq.hgb}). Then, for $\alpha\geq1$, the hazard function is increasing and for $0<\alpha<1$, is
decreasing and bathtub shaped.
\end{theorem}

\begin{proof} The proof is omitted, since $\theta > 0$ and therefore the proof is similar to the proof of Theorem \ref{thm.hgp}.
\end{proof}

The first and second non-central moments of GGB are given by
\begin{eqnarray*}
&&E(X)=\frac{\alpha\beta}{(\theta+1)^{m}-1}\sum_{n=1}^{\infty}\theta^{n}n \binom{m}{n}\sum_{k=0}^{\infty}\sum_{j=0}^{\infty}\frac{
  A_j(-1)^{j+k}e^{\frac{\beta}{\gamma}(j+1)}(\frac{\beta}{\gamma}(j+1))^{k}}{\Gamma(k+1)(\gamma+\gamma k)^{2}},\\
&&E(X^{2})= \frac{2\alpha\beta}{(\theta+1)^{m}-1}\sum_{n=1}^{\infty}
\theta^{n}n \binom{m}{n}\sum_{k=0}^{\infty}\sum_{j=0}^{\infty}\frac{
  A_j(-1)^{j+k+3}e^{\frac{\beta}{\gamma}(j+1)}(\frac{\beta}{\gamma}(j+1))^{k}}{\Gamma(k+1)(\gamma+\gamma k)^{3}}.                                      \end{eqnarray*}

\begin{figure}
\centering
\includegraphics[scale=0.35]{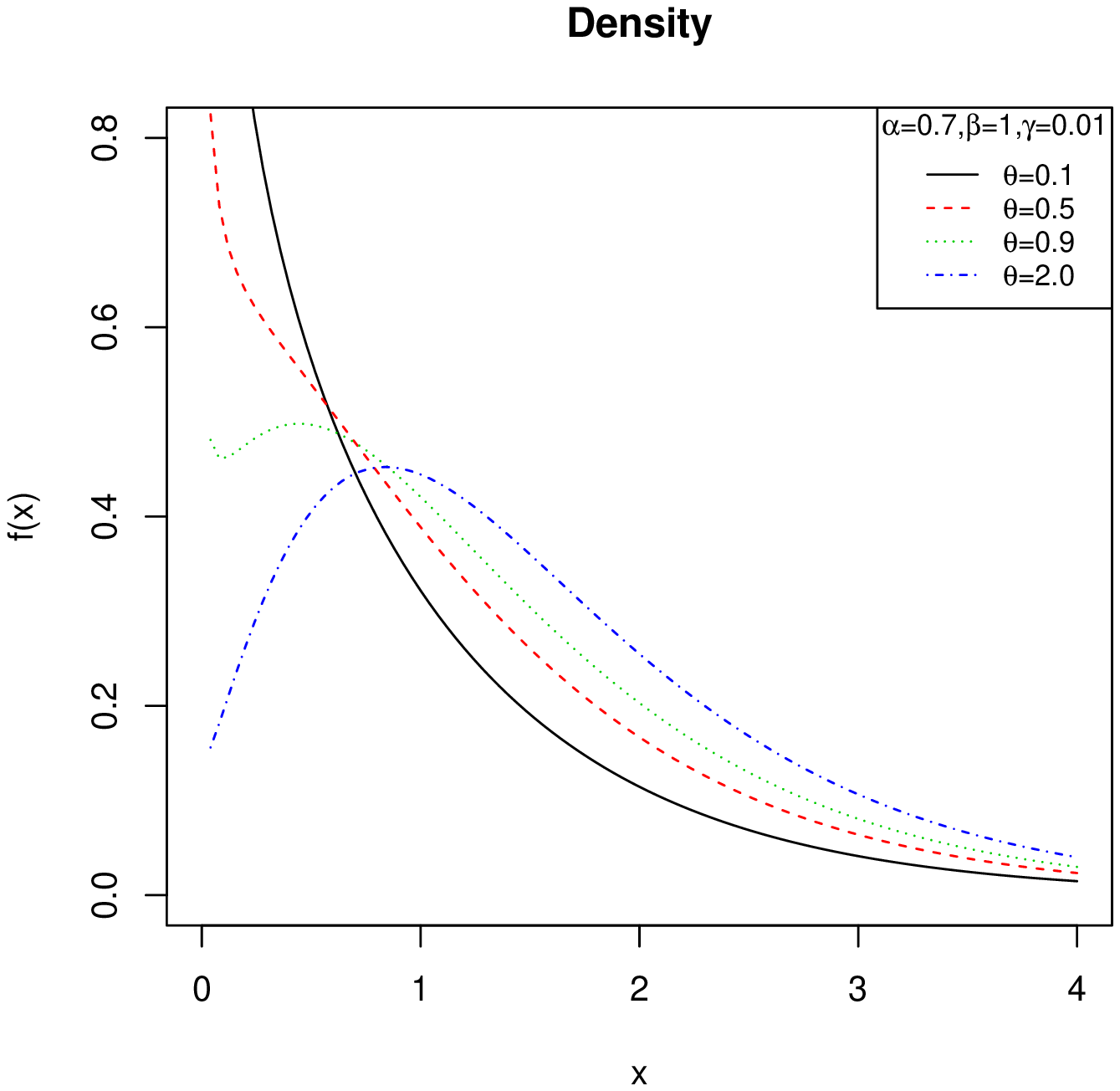}
\includegraphics[scale=0.35]{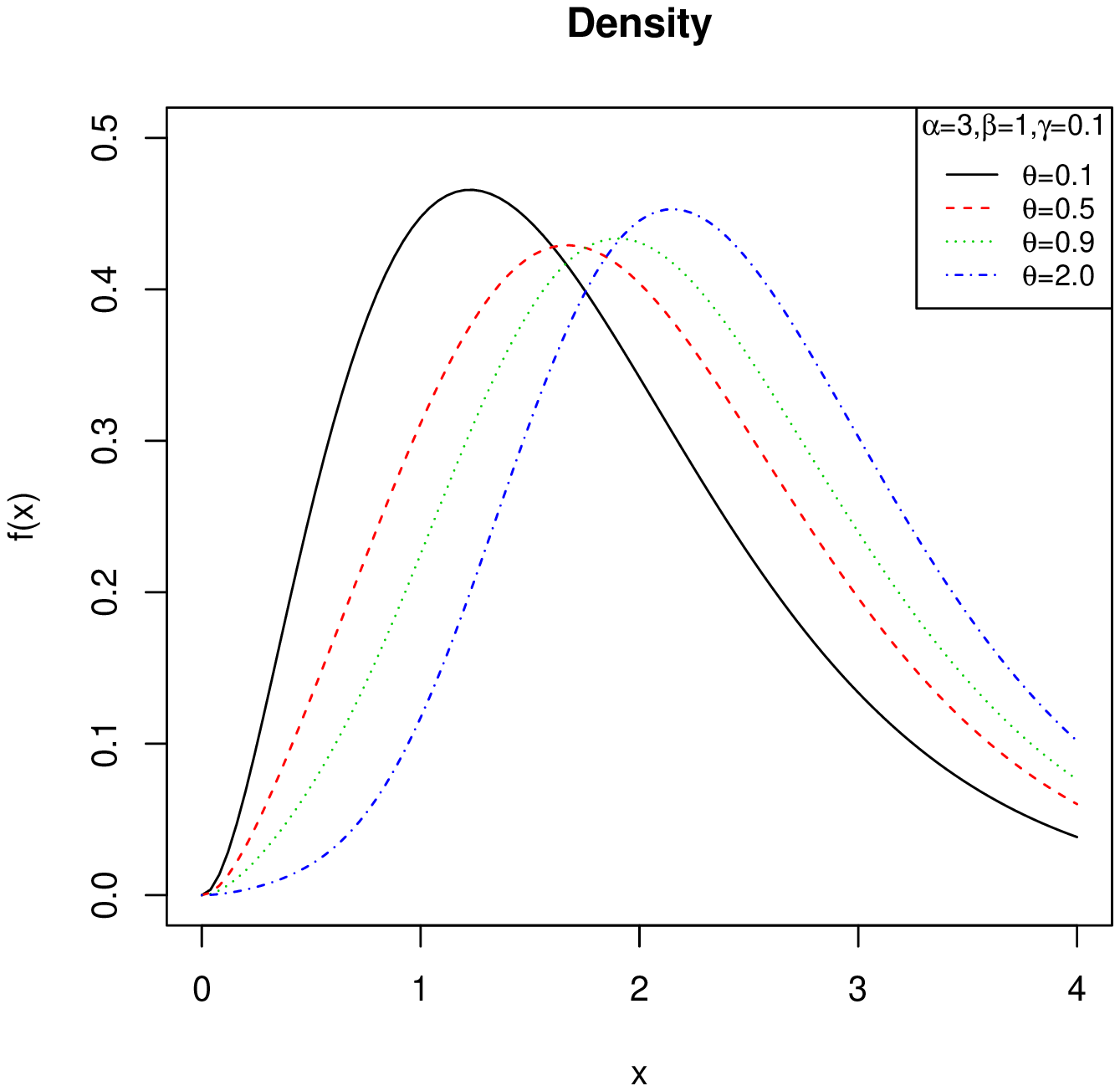}
\includegraphics[scale=0.35]{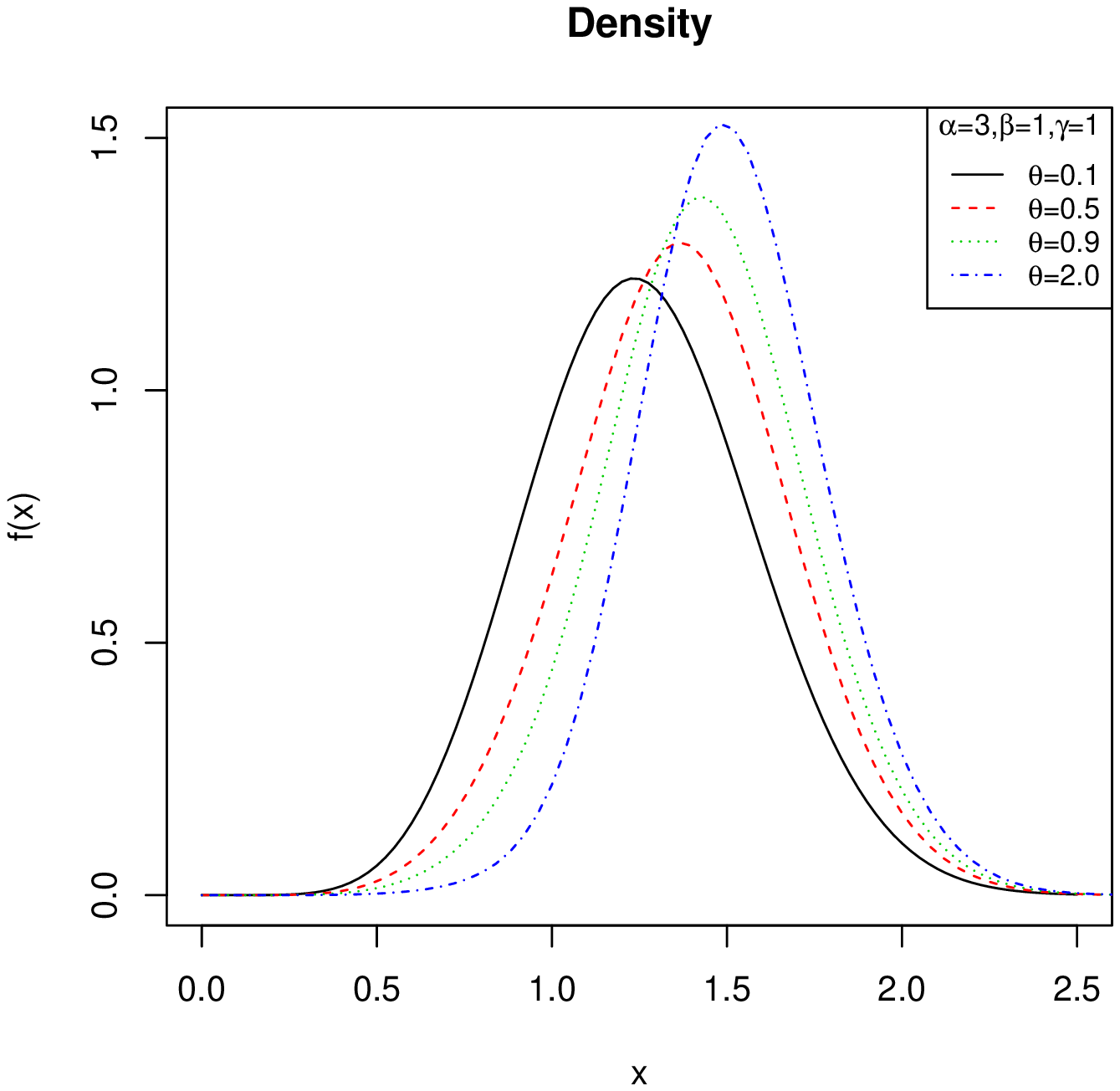}
\includegraphics[scale=0.35]{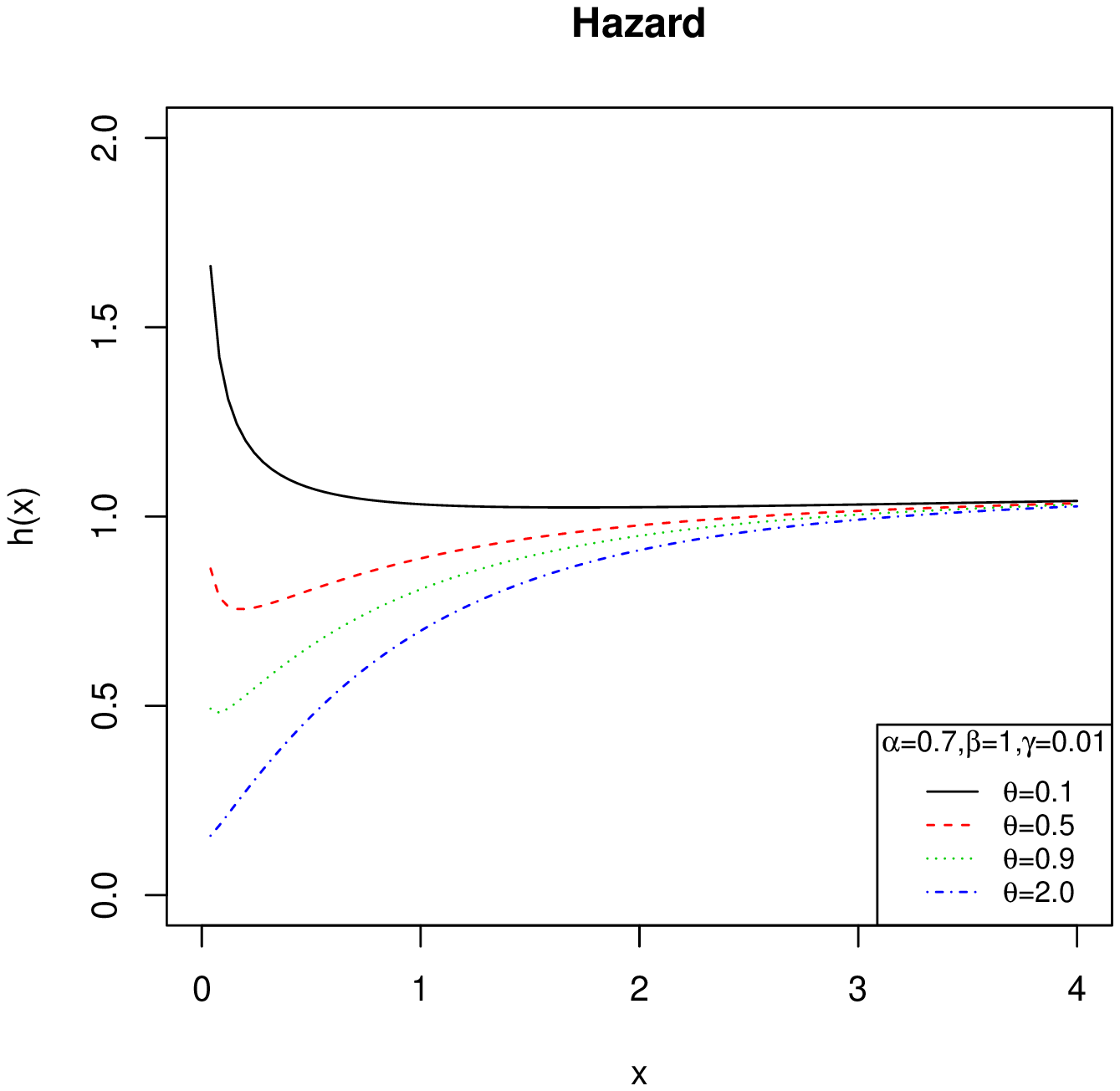}
\includegraphics[scale=0.35]{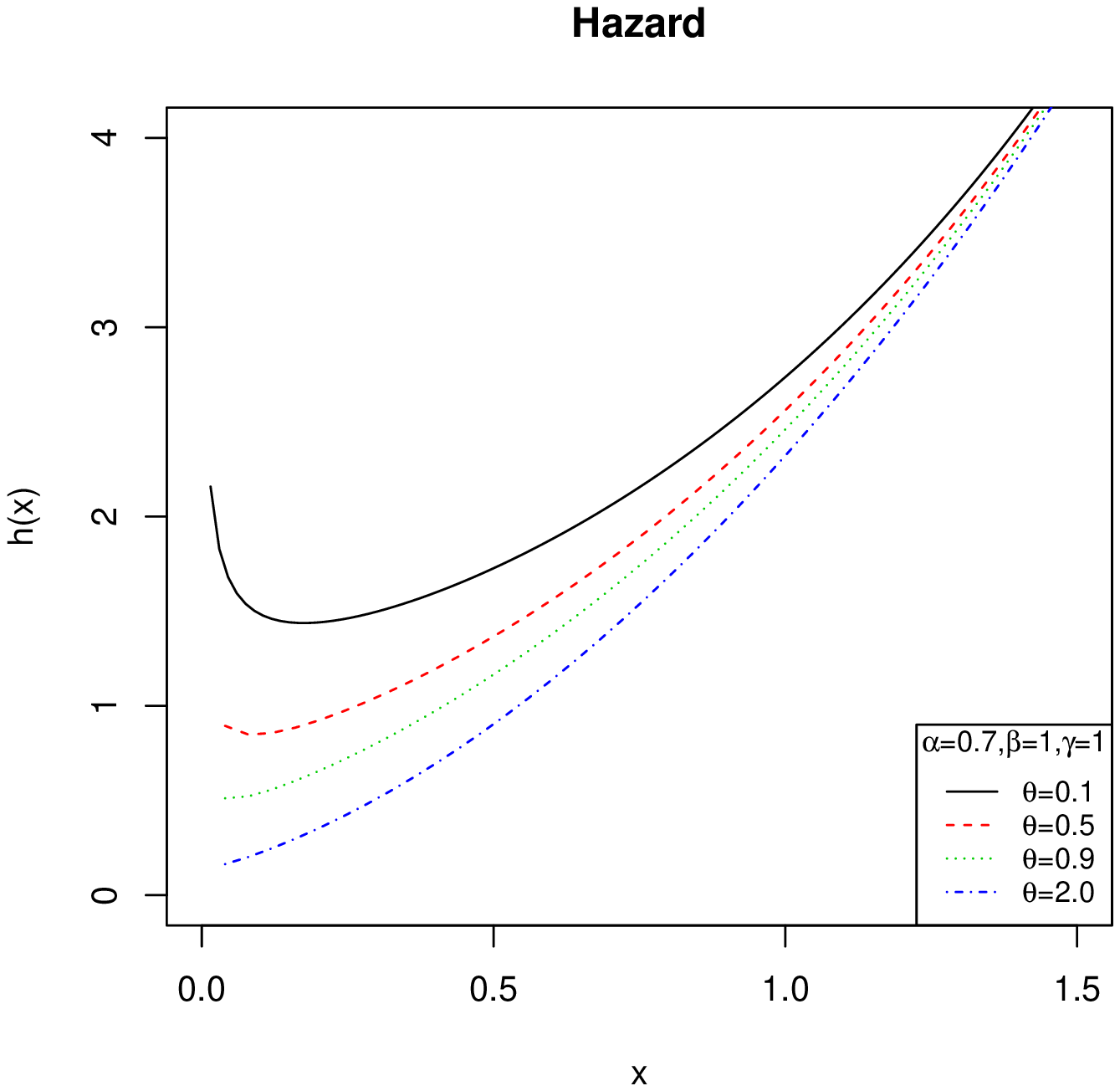}
\includegraphics[scale=0.35]{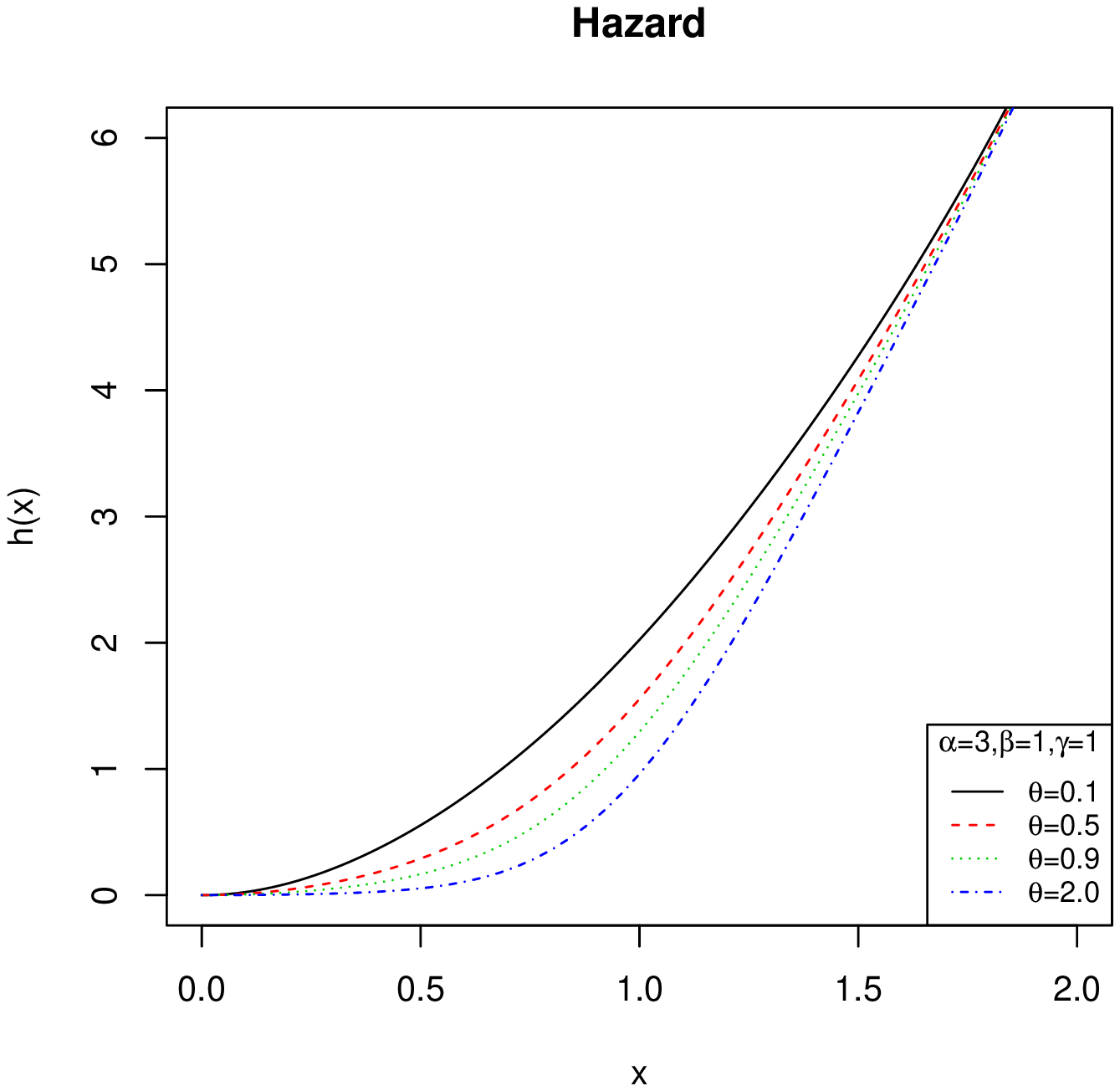}
\vspace{-0.8cm}
\caption{Plots of pdf and hazard rate function of GGB for $m=5$, and  different values $\alpha$, $\beta $, $\gamma$ and $\theta$.}\label{fig.GB}
\end{figure}

\subsection{Generalized  Gompertz-logarithmic distribution}
The logarithmic distribution (truncated at zero) is also a special case of power series distributions with $a_{n}=\frac{1}{n}$ and $C(\theta)=-\log(1-\theta) \ (0<\theta<1)$. The  pdf and hazard rate function of generalized Gompertz-logarithmic (GGL) distribution are given respectively by
\begin{eqnarray}
f(x)&=& \frac{\theta \alpha \beta e^{\gamma x}(1-t)t^{\alpha-1}}{(\theta t^{\alpha}-1)\log(1-\theta)},\ \ \ x>0,\\
\label{eq.hgl}
h(x)&=&  \frac{\theta \alpha \beta e^{\gamma x}(1-t)t^{\alpha-1}}{(\theta t^{\alpha}-1)\log(\frac{1-\theta}{1-\theta t^{\alpha}})},\ \ \ x>0.
\end{eqnarray}

The plots of pdf and hazard rate function of GGL for different values of $\alpha$, $\beta$, $\gamma$ and $\theta$  are given in Figure \ref{fig.GL}.

\begin{theorem}  Consider the GGL hazard function in (\ref{eq.hgl}). Then, for $\alpha\geq1$, the hazard function is increasing and for $0<\alpha<1$, is
decreasing and bathtub shaped.
\end{theorem}

\begin{proof} The proof is omitted, since $0<\theta<1$ and therefore the proof is similar to the proof of Theorem 1.
\end{proof}

The first and second non-central moments of GGL are
\begin{eqnarray*}
&&E(X)=\frac{\alpha\beta}{-\log(1-\theta)}\sum_{n=1}^{\infty}\theta^{n}
  \sum_{k=0}^{\infty}\sum_{j=0}^{\infty}\frac{
  A_j(-1)^{j+k}e^{\frac{\beta}{\gamma}(j+1)}(\frac{\beta}{\gamma}(j+1))^{k}}{\Gamma(k+1)(\gamma+\gamma k)^{2}},\\
&&E(X^{2})= \frac{2\alpha\beta}{-\log(1-\theta)}\sum_{n=1}^{\infty}
 \theta^{n}\sum_{k=0}^{\infty}\sum_{j=0}^{\infty}\frac{
  A_j(-1)^{j+k+3}e^{\frac{\beta}{\gamma}(j+1)}(\frac{\beta}{\gamma}(j+1))^{k}}{\Gamma(k+1)(\gamma+\gamma k)^{3}}.
\end{eqnarray*}

\begin{figure}[t]
\centering
\includegraphics[scale=0.35]{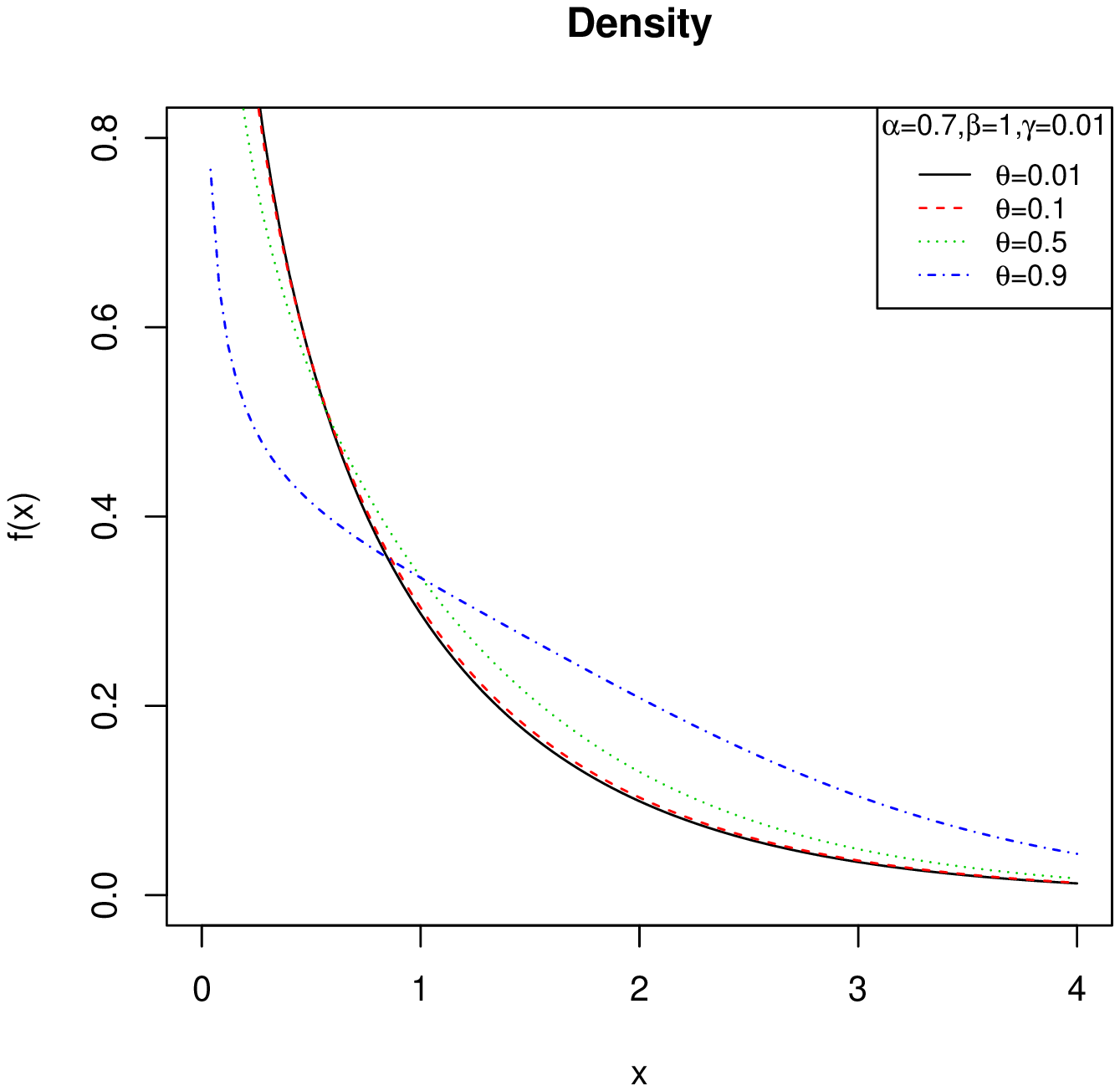}
\includegraphics[scale=0.35]{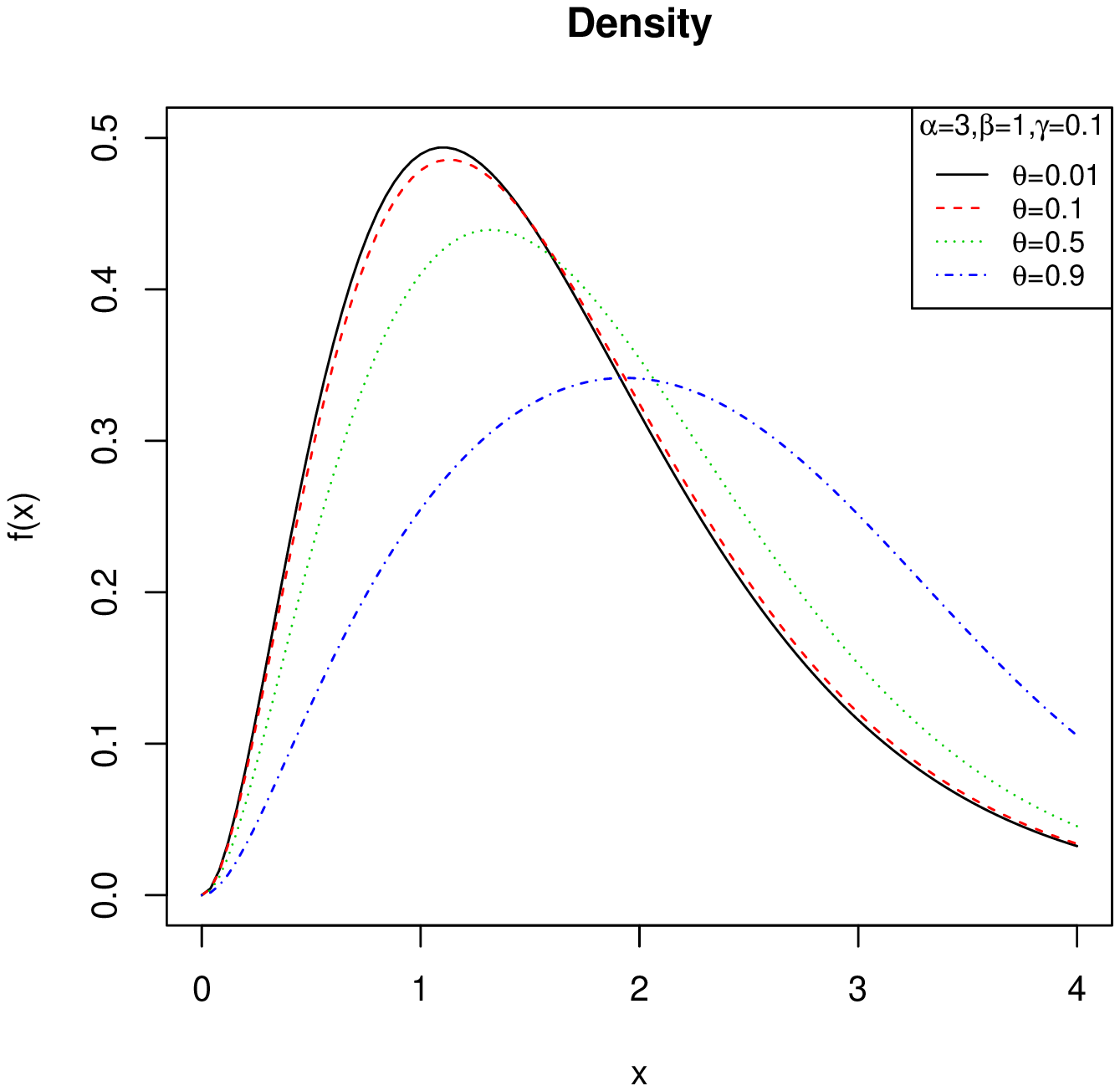}
\includegraphics[scale=0.35]{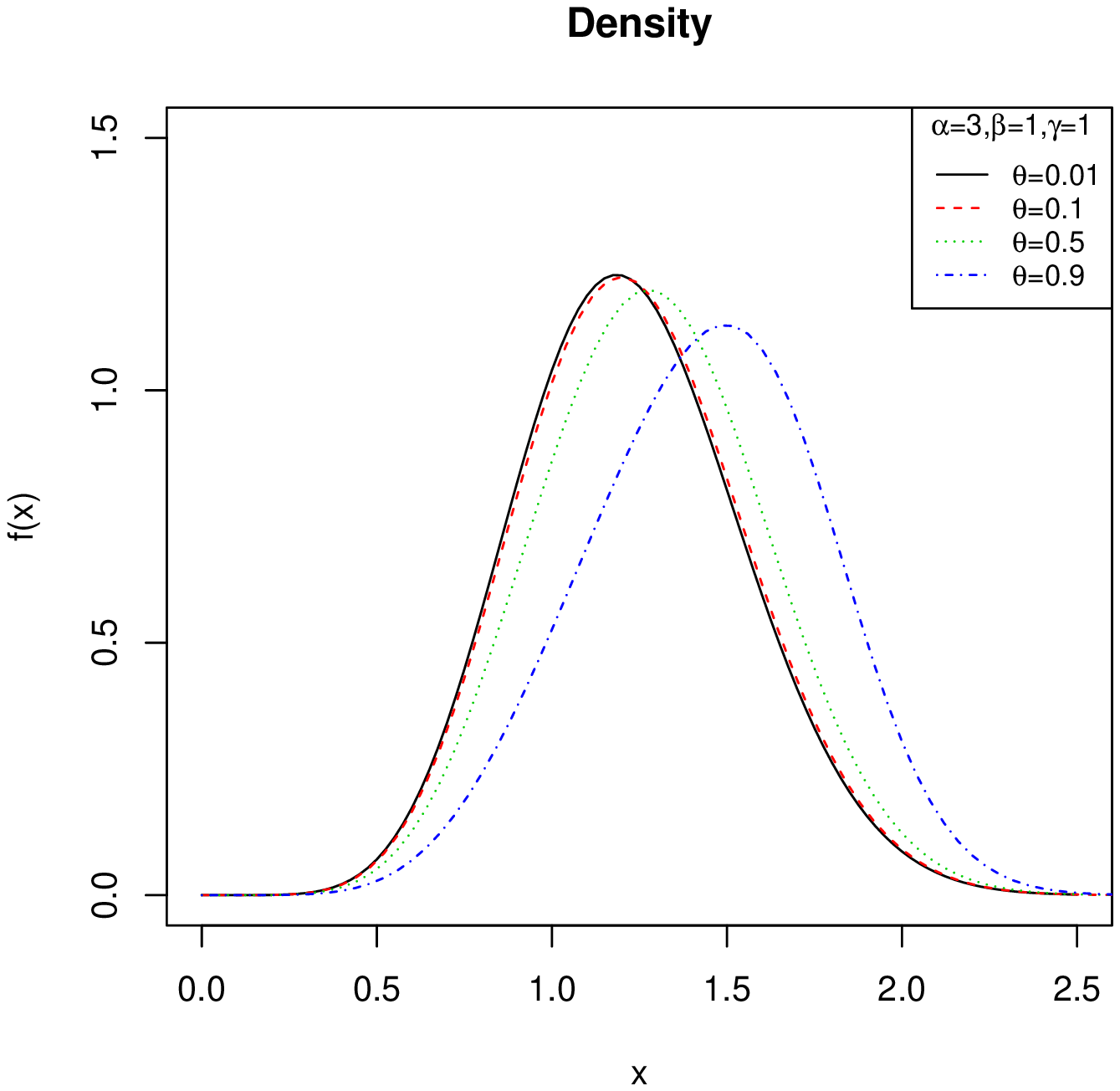}
\includegraphics[scale=0.35]{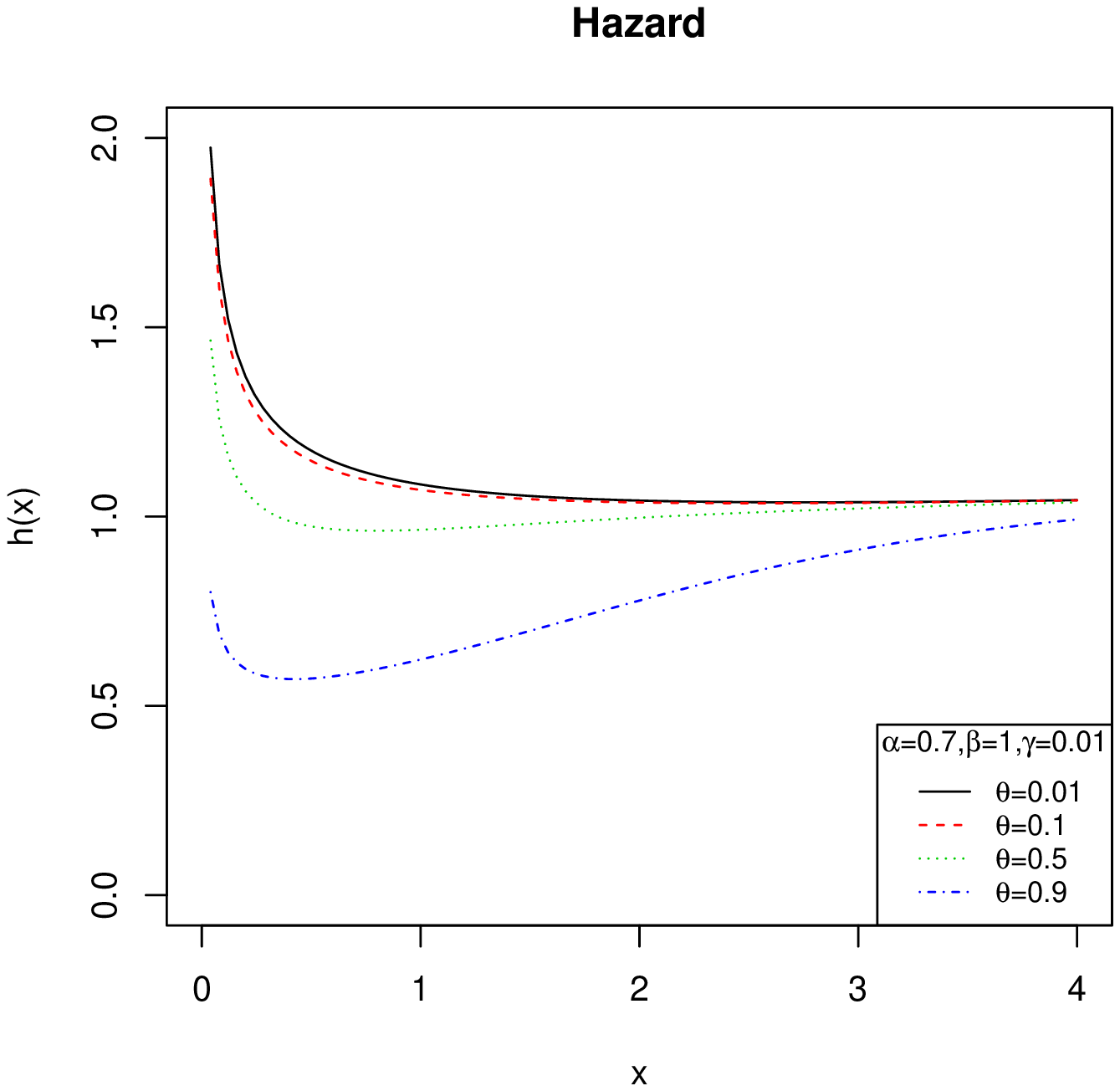}
\includegraphics[scale=0.35]{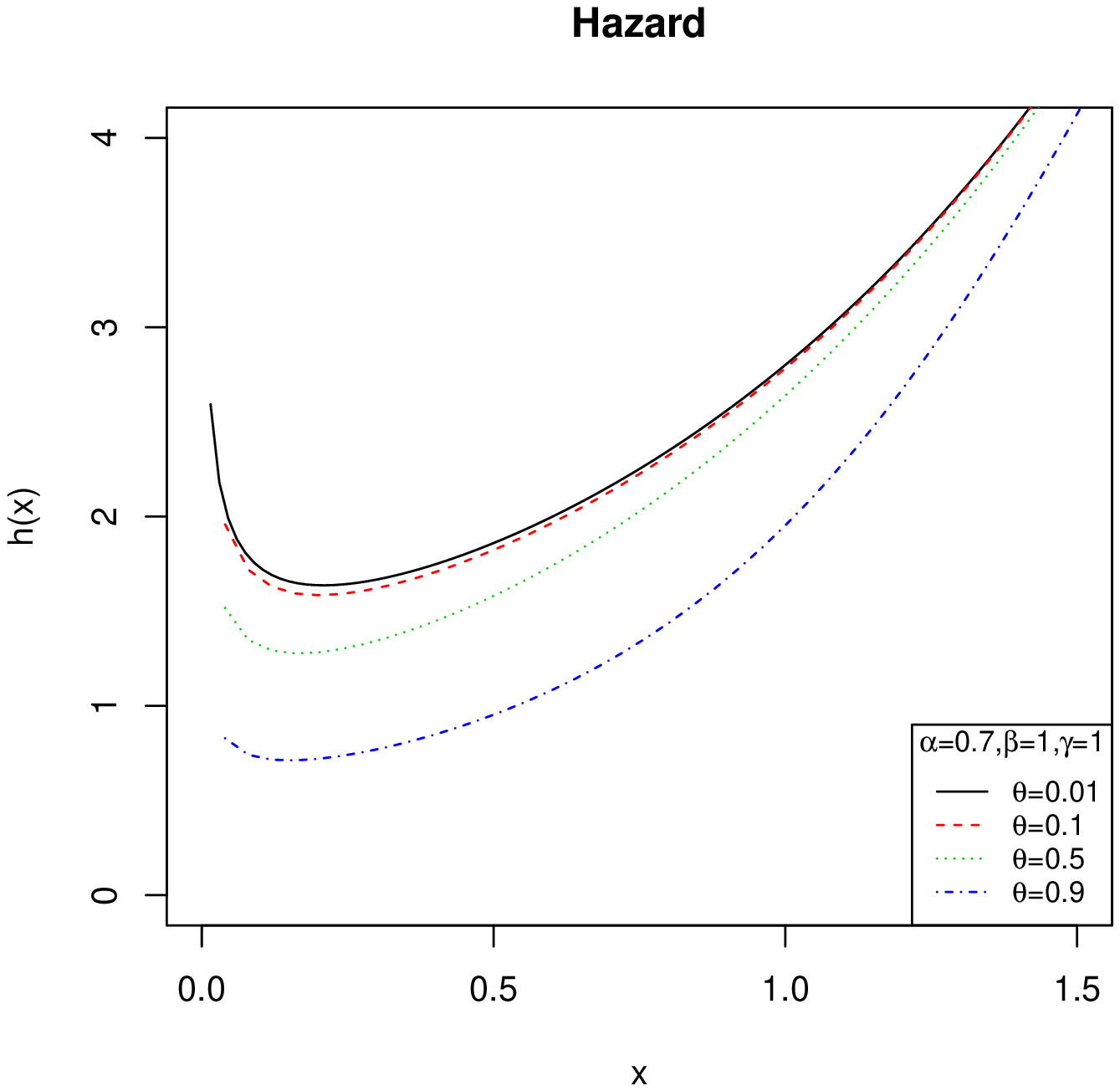}
\includegraphics[scale=0.35]{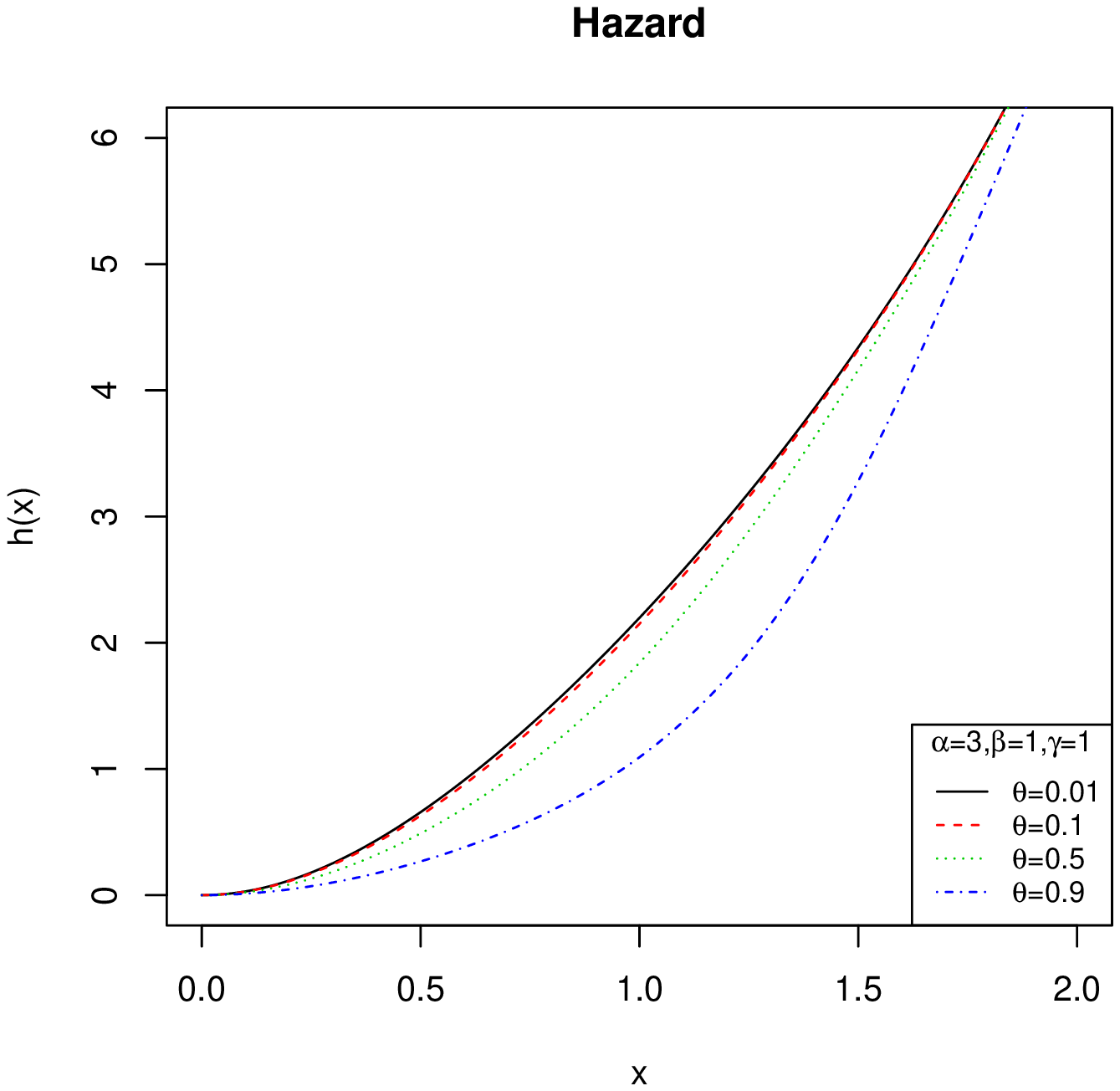}

\vspace{-0.8cm}
\caption{Plots of pdf and hazard rate function of GGL for different values $\alpha$, $\beta $, $\gamma$ and $\theta$.}\label{fig.GL}
\end{figure}

\section{Estimation and inference}
\label{sec.est}
In this section, we will derive the maximum likelihood estimators (MLE) of the unknown parameters  ${\boldsymbol \Theta}=(\alpha,\beta, \gamma, \theta)^T$ of the ${\rm GGPS}(\alpha,\beta, \gamma, \theta)$. Also, asymptotic confidence intervals of these parameters will be derived based on the Fisher information. At the end, we proposed an Expectation-Maximization (EM) algorithm for estimating the parameters.

\subsection{MLE for parameters}
 Let $X_{1},\dots, X_{n}$ be an independent random sample, with observed values  $x_{1},\dots,x_{n}$ from \break  ${\rm GGPS}(\alpha,\beta,\gamma,\theta)$ and
 $ {\boldsymbol\Theta}=(\alpha,\beta,\gamma,\theta)^{T}$  be a                                                                                                                                                                    parameter vector. The log-likelihood function is given by
 \begin{eqnarray*}
 l_{n} = l_{n}(\boldsymbol \Theta; \boldsymbol x)&=& n\log(\theta)
 +n\log(\alpha\beta)
+n\gamma \bar{x}
+ \sum_{i=1}^{n} \log(1-t_i)
+(\alpha-1)\sum_{i=1}^{n}\log(t_i)\nonumber\\
&&+ \sum_{i=1}^{n}\log(C'(\theta t_i^{\alpha})) -n\log(C(\theta)),
 \end{eqnarray*}
where $t_{i}=1-e^{-\frac{\beta}{\gamma}(e^{\gamma x_{i}}-1)}$. Therefore, the  score function is given by $U({\boldsymbol\Theta};{\boldsymbol x})=(\frac{\partial l_{n}}{\partial\alpha},\frac{\partial l_{n}}{\partial\beta},\frac{\partial l_{n}}{\partial\gamma},$ $\frac{\partial l_{n}}{\partial\theta})^{T}$, where
\begin{eqnarray}
\frac{\partial l_{n}}{\partial \alpha}&=&\frac{n}{\alpha}+\sum_{i=1}^{n}\log(t_{i})+\sum_{i=1}^{n}\frac{\theta t_{i}^{\alpha}\log(t_{i})C''(\theta t_{i}^{\alpha})}{C'(\theta t_{i}^{\alpha})}, \label{eq.la}\\
\frac{\partial l_{n}}{\partial \beta}&=&\frac{n}{\beta}-\frac{1}{\gamma}(\sum_{i=1}^{n}e^{\gamma x_{i}}-n)+(\alpha-1)\sum_{i=1}^{n}\frac{\frac{\partial t_{i}}{\partial \beta}}{t_{i}}+\sum_{i=1}^{n}\frac{\theta\frac{\partial(t_{i}^{\alpha})}{\partial \beta}C''(\theta t_{i}^{\alpha})}{C'(\theta t_{i}^{\alpha})},                                                                                                                                                          \label{eq.lb}
\\
 \frac{\partial l_{n}}{\partial \gamma}&=&n\bar{x}+\frac{\beta}{\gamma^{2}}(\sum_{i=1}^{n}e^{\gamma x_{i}}-n)-\frac{\beta}{\gamma}(\sum_{i=1}^{n}x_{i}e^{\gamma x_{i}})\nonumber\\
 &&+ (\alpha-1)\sum_{i=1}^{n}\frac{\frac{\partial t_{i}}{\partial\gamma}}{t_{i}}+\sum_{i=1}^{n}\frac{\theta\frac{\partial(t_{i}^{\alpha})}{\partial \gamma}C''(\theta t_{i}^{\alpha})}{C'(\theta t_{i}^{\alpha})},
 \label{eq.lg}\\
 \frac{\partial l_{n}}{\partial \theta}&=&\frac{n}{\theta}+\sum_{i=1}^{n}\frac{
 t_{i}^{\alpha}C''(\theta t_{i}^{\alpha})}{C'(\theta t_{i}^{\alpha})}-\frac{nC'(\theta)}{C(\theta)}.
 \label{eq.lt}
  \end{eqnarray}
The MLE of ${\boldsymbol \Theta}$, say $\hat{\boldsymbol\Theta}$, is obtained by solving the nonlinear system $U({\boldsymbol\Theta};{\boldsymbol x})={\boldsymbol 0}$. We cannot get an explicit form for this nonlinear system of equations and they can be calculated by using a numerical method, like the Newton method or the bisection method.

For each element of the power series distributions (geometric, Poisson, logarithmic and binomial), we have the following theorems for the MLE of parameters:
\begin{theorem}\label{th.lb}
Let ${\rm g}_{1}(\alpha;\beta,\gamma,\theta,x)$ denote the function on RHS of the expression in \eqref{eq.la}, where $\beta$, $\gamma$ and $\theta$ are the true values of the parameters. Then, for a given $\beta>0$, $\gamma>0$ and $\theta>0$, the roots of ${\rm g}_1(\alpha,\beta; \gamma,\theta, {\boldsymbol x}) =0$, lies in the interval
\[\left(\frac{-n}{\frac{\theta C''(\theta)}{C'(\theta)}+1}(\sum_{i=1}^{n}\log(t_i))^{-1},-n(\sum_{i=1}^{n}\log(t_i))^{-1})\right),\]
\end{theorem}
\begin{proof} See Appendix B.1.
\end{proof}

\begin{theorem}\label{th.lg}
Let ${\rm g}_{2}(\beta;\alpha,\gamma,\theta,x)$  denote the function on RHS of the expression in \eqref{eq.lg}, where $\alpha$, $\gamma$ and $\theta$ are the true values of the
parameters. Then, the equation  ${\rm g}_{2}(\beta;\alpha,\gamma,\theta,x)=0$ has at least one root.
\end{theorem}

\begin{proof} See Appendix B.2.
\end{proof}

\begin{theorem}\label{th.lt}
Let ${\rm g}_{3}(\theta;\alpha,\beta,\gamma,x)$ denote the function on RHS of the expression in \eqref{eq.lt} and $\bar{x}=n^{-1}\sum_{i=1}^{n}x_{i}$, where $\alpha$, $\beta$ and $\gamma$ are the true values of the
parameters.\\
a) The equation ${\rm g}_{3}(\theta;\alpha,\beta,\gamma,x)=0$ has at least one root  for all GGG, GGP and GGL  distributions if $\sum_{i=1}^{n}t_{i}^{\alpha}>\frac{n}{2}$.\\
b) If  ${\rm g}_{3}(p;\alpha,\beta,\gamma,x)=\frac{\partial l_{n}}{\partial p}$, where $p=\frac{\theta}{\theta+1}$ and $p\in (0,1)$ then the equation ${\rm g}_{3}(p;\alpha,\beta,\gamma,x)=0$ has at least one root for GGB  distribution if $\sum_{i=1}^{n}t_{i}^{\alpha}>\frac{n}{2}$ and  $\sum_{i=1}^{n}t_{i}^{-\alpha}>\frac{nm}{m-1}$.
\end{theorem}
\begin{proof} See Appendix B.3.
\end{proof}

Now, we derive  asymptotic confidence intervals for the parameters of GGPS distribution. It is well-known that under regularity conditions \citep[see][Section 10]{ca-be-01}, the asymptotic distribution of
$\sqrt{n}(\hat{\boldsymbol\Theta}-{\boldsymbol\Theta})$ is multivariate normal with mean ${\boldsymbol 0}$ and variance-covariance matrix $J_n^{-1}({\boldsymbol\Theta})$, 
where $J_n({\boldsymbol\Theta})=\lim_{n\rightarrow \infty} I_n({\boldsymbol\Theta})$, and $I_n({\boldsymbol\Theta})$ is the $4\times4$  observed information matrix, i.e.
\[I_n\left(\Theta \right)=-\left[ \begin{array}{cccc}
I_{\alpha \alpha } & I_{\alpha \beta } & I_{\alpha \gamma } & I_{\alpha \theta } \\
I_{\beta \alpha } & I_{\beta \beta } &I_{\beta \gamma } &I_{\beta \theta } \\
I_{\gamma \alpha } & I_{\gamma \beta } &I_{\gamma \gamma } &I_{\gamma \theta } \\
I_{\theta \alpha } & I_{\theta \beta } & I_{\theta \gamma} & I_{\theta \theta}\\
\end{array} \right],\]
whose elements are given in Appendix C. Therefore, an $100(1-\eta)$ asymptotic confidence interval for each parameter, ${\boldsymbol\Theta}_{r}$, is given by
$$
ACI_{r}=(\hat{\boldsymbol\Theta}_{r} -Z_{\eta/2}\sqrt{\hat{I}_{rr}}, \hat{\boldsymbol\Theta}_{r}+Z_{\frac{\eta}{2}}\sqrt{\hat{I}_{rr}}),
$$
 where $ \hat{I}_{rr}$ is the $(r,r)$  diagonal element of $I_{n}^{-1}(\hat{\boldsymbol\Theta})$ for $r=1,2,3,4$ and $Z_{\eta/2}$ is the quantile
$\frac{\eta}{2}$ of the standard normal distribution.

\subsection{ EM-algorithm}



The traditional methods to obtain the MLE of parameters are numerical methods by solving the equations \eqref{eq.la}-\eqref{eq.lt}, and sensitive to the initial values. Therefore, we develop an  Expectation-Maximization (EM) algorithm to obtain the MLE of parameters. It is an iterative method, and is a very powerful tool in handling the incomplete data problem
\citep{de-la-ru-77}.

We define a hypothetical complete-data distribution with a joint pdf in the form
$$
g(x,z;{\boldsymbol \Theta})=\frac{a_{z}\theta^{z}}{C(\theta)}z\alpha \beta e^{\gamma x}(1-t)t^{z\alpha-1},
$$
where $t=1-e^{\frac{-\beta}{\gamma}(e^{\gamma x}-1)}$, and  $\alpha, \beta$, $\gamma$, $\theta>0$, $x>0$ and $z\in \mathbb{N}$. Suppose ${\boldsymbol\Theta}^{(r)}=(\alpha^{(r)},\beta^{(r)},\gamma^{(r)},$ $\theta^{(r)})$ is the current estimate (in the $r$th iteration) of ${\boldsymbol\Theta}$. Then, the E-step of an EM cycle requires the expectation of $(Z|X;{\boldsymbol\Theta}^{(r)})$. The pdf of $Z$ given $X=x$ is given by
$$
g(z|x)=\frac{a_{z}\theta^{z-1}z t^{z\alpha-\alpha}}{C'(\theta t^{\alpha})},
$$
and since
\begin{eqnarray*}
C'(\theta)+\theta C''(\theta)=\sum_{z=1}^\infty  a_z z \theta ^{z-1}+\theta \sum_{z=1}^\infty a_z z(z-1) \theta ^{z-2} =\sum_{z=1}^{\infty}z^{2}a_{z}\theta^{z-1},
\end{eqnarray*}
 the expected value of $Z|X=x$ is obtained as
\begin{eqnarray}
E(Z|X=x)=
1+\frac{\theta t^{\alpha}C''(\theta t^{\alpha})}{C'(\theta t^{\alpha})}.
\end{eqnarray}

By using the MLE over $\boldsymbol\Theta$, with the missing $Z$'s replaced by their conditional expectations given
above, the M-step of EM cycle is completed.  Therefore, the log-likelihood for the complete-data is
\begin{eqnarray} \label{eq.ls}
l^{\ast}_n(\boldsymbol y,\boldsymbol \Theta)&\propto& \sum_{i=1}^{n}z_{i}\log(\theta)
+n\log(\alpha\beta)
+n\gamma \bar{x}
+\sum_{i=1}^{n}\log(1-t_i)\nonumber\\
&&+\sum_{i=1}^{n}(z_{i}\alpha-1)\log(t_i)
-n\log(C(\theta)),
\end{eqnarray}
where $\boldsymbol y= (\boldsymbol x;\boldsymbol z)$, $\boldsymbol x=(x_1,\dots,x_n)$ and ${\boldsymbol z}=(z_1,\dots,z_n)$. On differentiation of (\ref{eq.ls}) with respect to parameters $\alpha$, $\beta$, $\gamma$ and $\theta$, we obtain the components of the score function, $U({\boldsymbol y};\boldsymbol\Theta)=(\frac{\partial l^{\ast}_{n}}{\partial\alpha},\frac{\partial l^{\ast}_{n}}{\partial\beta},\frac{\partial l^{\ast}_{n}}{\partial\gamma},\frac{\partial l^{\ast}_{n}}{\partial\theta})^T$, as
\begin{eqnarray*}
\frac{\partial l^{\ast}_{n}}{\partial\alpha}&=&\frac{n}{\alpha}+\sum_{i=1}^{n}z_{i}\log[1-e^{\frac{-\beta}{\gamma}(e^{\gamma x_{i}}-1)}],\\
\frac{\partial l^{\ast}_{n}}{\partial\beta}&=&\frac{n}{\beta}-\frac{1}{\gamma}(\sum_{i=1}^{n}e^{\gamma x_{i}}-n)+\sum_{i=1}^{n}(z_{i}\alpha-1)\frac{\frac{1}{\gamma}(e^{\gamma x_{i}}-1)}{[e^{\frac{\beta}{\gamma}(e^{\gamma x_{i}}-1)}-1]},\\
\frac{\partial l^{\ast}_{n}}{\partial\gamma}&=&n\bar{x}+\frac{\beta}{\gamma^{2}}(\sum_{i=1}^{n}e^{\gamma x_{i}}-n)-\frac{\beta}{\gamma}(\sum_{i=1}^{n} x_{i}e^{\gamma x_{i}})+\sum_{i=1}^{n}(z_{i}\alpha-1)\frac{\frac{-\beta}{\gamma^{2}}(e^{\gamma x_{i}}-1)+\frac{\beta x_{i}e^{\gamma x_{i}}}{\gamma}}{[e^{\frac{\beta}{\gamma}(e^{\gamma x_{i}}-1)}-1]},\\
\frac{\partial l^{\ast}_{n}}{\partial\theta}&=&\sum_{i=1}^{n}\frac{z_{i}}{\theta}-n\frac{C'(\theta)}{C(\theta)}.
\end{eqnarray*}
From a nonlinear system of equations $U(\boldsymbol y;\boldsymbol \Theta)=0$, we obtain the iterative procedure of the EM-algorithm as
\begin{eqnarray*}
&&\hat{\alpha}^{(j+1)}=\frac{-n}{\sum_{i=1}^{n}\hat{z_{i}}^{(j)}\log[1-e^{\frac{-\hat{\beta}^{(j)}}{\hat{\gamma}^{(j)}}(e^{\hat{\gamma }^{(j)} x_{i}}-1)}]}, \qquad \hat{\theta}^{(j+1)}-\frac{C(\hat{\theta}^{(j+1)})}{nC'(\hat{\theta}^{(j+1)})}\sum_{i=1}^{n}\hat{z}_{i}^{(j)}=0,\\
&&\frac{n}{\hat{\beta}^{(j+1)}}-\frac{1}{\hat{\gamma}^{(j)}}(\sum_{i=1}^{n}e^{\hat{\gamma}^{(j)} x_{i}}-n)+\sum_{i=1}^{n}({\hat z}_{i}\hat{\alpha}^{(j)}-1)\frac{\frac{1}{\hat{\gamma}^{(j)}}(e^{\hat{\gamma}^{(j)} x_{i}}-1)}{[e^{\frac{\hat{\beta}^{(j+1)}}{\hat{\gamma}^{(j)}}(e^{\hat{\gamma}^{(j)}x_{i}}-1)}-1]}=0,
\\
&&n\bar{x}+\frac{\hat{\beta}^{(j)}}{[\hat{\gamma}^{(j+1)}]^{2}}(\sum_{i=1}^{n}e^{\hat{\gamma}^{(j+1)} x_{i}}-n)-\frac{\hat{\beta}^{(j)}}{\hat{\gamma}^{(j+1)}}(\sum_{i=1}^{n} x_{i}e^{\hat{\gamma}^{(j+1)} x_{i}})
\\
&&\quad \ \ +\sum_{i=1}^{n}({\hat z}_{i}\hat{\alpha}^{(j)}-1)\frac{\frac{-\hat{\beta}^{(j)}}{[\hat{\gamma}^{(j+1)}]^{2}}(e^{\hat{\gamma}^{(j+1)} x_{i}}-1)+\frac{\hat{\beta}^{(j)} x_{i}e^{\hat{\gamma}^{(j+1)} x_{i}}}{\hat{\gamma}^{(j+1)}}}{[e^{\frac{\hat{\beta}^{(j)}}{\hat{\gamma}^{(j+1)}}(e^{\hat{\gamma}^{(j+1)} x_{i}}-1)}-1]}=0,
\end{eqnarray*}
where $\hat{\theta}^{(j+1)}$, $\hat{\beta}^{(j+1)}$ and $\hat{\gamma}^{(j+1)}$ are found numerically. Here, for $i=1,2,\dots,n$, we have that
$$
\hat{z}_{i}^{(j)}=1+\frac{\theta^{*(j)}C''(\theta^{*(j)})}{C'(\theta^{*(j)})},
$$
where
$\theta^{*(j)}=\hat{\theta}^{(j)}[1-e^{-\frac{\hat{\beta}^{(j)}}{\hat{\gamma}^{(j)}}(e^{\hat{\gamma}^{(j)} x_{i}}-1)}]^{\hat{\alpha}^{(j)}}$.

%

We can use the results of
\cite{lou-82}
to obtain the standard errors of the estimators from the EM-algorithm. Consider
$l_{c}({\boldsymbol\Theta};{\boldsymbol x})=E(I_{c}({\boldsymbol\Theta};{\boldsymbol y})|{\boldsymbol x})$,
 where
$I_{c}({\boldsymbol\Theta};{\boldsymbol y})=-[\frac{\partial U({\boldsymbol y};{\boldsymbol\Theta})}{\partial{\boldsymbol\Theta}}]$
is the $4\times 4$ observed information matrix.If
$
l_{m}({\boldsymbol\Theta};{\boldsymbol x})=Var[U({\boldsymbol y};{\boldsymbol\Theta})|{\boldsymbol x}]$, then, we obtain the observed information as
$$
I(\hat{{\boldsymbol\Theta}};{\boldsymbol x})=l_{c}(\hat{\boldsymbol\Theta};{\boldsymbol x})-l_{m}(\hat{\boldsymbol\Theta};{\boldsymbol x}).
$$
The standard errors of the MLEs of the EM-algorithm are the square root of the diagonal elements of the $I(\hat{\boldsymbol\Theta};{\boldsymbol x})$. The computation of these matrices are too long and tedious. Therefore, we did not present the details. Reader can see \cite{ma-ja-12} how to calculate these values.

\section{Simulation study}

We performed a simulation  in order to investigate the proposed estimator of
$\alpha$, $\beta$, $\gamma$ and $\theta$ of the proposed EM-scheme. We generated 1000 samples of size $n$
from the GGG distribution with $\beta=1$ and $\gamma =0.1$. Then,
the averages of estimators (AE), standard error of estimators (SEE), and averages of standard errors (ASE) of MLEs of the EM-algorithm determined though the Fisher information matrix are calculated. The results are given in Table \ref{tab.sim}. We can find that

\noindent (i) convergence has been achieved in all cases and this emphasizes the numerical stability of the EM-algorithm,

\noindent (ii) the differences between the average estimates and the true values are almost small,

\noindent (iii) the standard errors of the MLEs decrease when the sample size increases.

\label{sec.sim}

\begin{table}[ht]
\begin{center}

\caption{The average MLEs, standard error of estimators and averages of standard errors for the GGG distribution.} \label{tab.sim}

\begin{tabular}{|c|cc|cccc|cccc|cccc|} \hline
 & \multicolumn{2}{|c|}{parameter} & \multicolumn{4}{|c|}{AE} & \multicolumn{4}{|c|}{SEE} & \multicolumn{4}{|c|}{ASE} \\ \hline
$n$ & $\alpha $ & $\theta $ & $\hat\alpha $ & $\hat\beta $ & $\hat\gamma $ & $\hat\theta $ & $\hat\alpha $ & $\hat\beta $ & $\hat\gamma $ & $\hat\theta $ & $\hat\alpha $ & $\hat\beta $ & $\hat\gamma $ & $\hat\theta $ \\ \hline
50 & 0.5 & 0.2 & 0.491 & 0.961 & 0.149 & 0.204 & 0.114 & 0.338 & 0.265 & 0.195 & 0.173 & 0.731 & 0.437 & 0.782 \\
 & 0.5 & 0.5 & 0.540 & 0.831 & 0.182 & 0.389 & 0.160 & 0.337 & 0.260 & 0.263 & 0.210 & 0.689 & 0.421 & 0.817 \\
 & 0.5 & 0.8 & 0.652 & 0.735 & 0.154 & 0.684 & 0.304 & 0.377 & 0.273 & 0.335 & 0.309 & 0.671 & 0.422 & 0.896 \\
 & 1.0 & 0.2 & 0.988 & 0.972 & 0.129 & 0.206 & 0.275 & 0.319 & 0.191 & 0.209 & 0.356 & 0.925 & 0.436 & 0.939 \\
 & 1.0 & 0.5 & 1.027 & 0.852 & 0.147 & 0.402 & 0.345 & 0.352 & 0.226 & 0.283 & 0.408 & 0.873 & 0.430 & 0.902 \\
 & 1.0 & 0.8 & 1.210 & 0.711 & 0.178 & 0.745 & 0.553 & 0.365 & 0.230 & 0.342 & 0.568 & 0.799 & 0.433 & 0.898 \\
 & 2.0 & 0.2 & 1.969 & 0.990 & 0.084 & 0.216 & 0.545 & 0.305 & 0.151 & 0.228 & 0.766 & 1.135 & 0.422 & 0.902 \\
 & 2.0 & 0.5 & 1.957 & 0.842 & 0.113 & 0.487 & 0.608 & 0.334 & 0.192 & 0.277 & 0.820 & 1.061 & 0.431 & 0.963 \\
 & 2.0 & 0.8 & 2.024 & 0.713 & 0.161 & 0.756 & 0.715 & 0.396 & 0.202 & 0.353 & 1.143 & 0.873 & 0.402 & 0.973 \\ \hline

100 & 0.5 & 0.2 & 0.491 & 0.977 & 0.081 & 0.212 & 0.084 & 0.252 & 0.171 & 0.179 & 0.125 & 0.514 & 0.283 & 0.561 \\
 & 0.5 & 0.5 & 0.528 & 0.883 & 0.109 & 0.549 & 0.124 & 0.275 & 0.178 & 0.247 & 0.155 & 0.504 & 0.275 & 0.567 \\
 & 0.5 & 0.8 & 0.602 & 0.793 & 0.136 & 0.769 & 0.215 & 0.323 & 0.194 & 0.299 & 0.220 & 0.466 & 0.259 & 0.522 \\
 & 1.0 & 0.2 & 0.974 & 0.997 & 0.102 & 0.226 & 0.195 & 0.242 & 0.129 & 0.206 & 0.251 & 0.645 & 0.280 & 0.767 \\
 & 1.0 & 0.5 & 1.030 & 0.875 & 0.113 & 0.517 & 0.262 & 0.291 & 0.155 & 0.270 & 0.298 & 0.651 & 0.295 & 0.843 \\
 & 1.0 & 0.8 & 1.113 & 0.899 & 0.117 & 0.846 & 0.412 & 0.342 & 0.177 & 0.331 & 0.400 & 0.600 & 0.287 & 0.781 \\
 & 2.0 & 0.2 & 1.952 & 0.995 & 0.138 & 0.221 & 0.424 & 0.237 & 0.117 & 0.209 & 0.524 & 0.922 & 0.321 & 0.992 \\
 & 2.0 & 0.5 & 2.004 & 0.885 & 0.110 & 0.518 & 0.493 & 0.283 & 0.131 & 0.274 & 0.601 & 0.873 & 0.321 & 0.966 \\
 & 2.0 & 0.8 & 2.028 & 0.981 & 0.104 & 0.819 & 0.605 & 0.350 & 0.155 & 0.339 & 0.816 & 0.717 & 0.289 & 0.946 \\ \hline
\end{tabular}

\end{center}
\end{table}

\section{Real examples}
\label{sec.exa}
In this Section, we consider two real data sets and  fit the Gompertz, GGG, GGP, GGB (with $m=5$), and GGL distributions. The first data set is negatively skewed,  and the second data set is positively skewed, and we show that the proposed distributions  fit both positively skewed and  negatively skewed data well. For each data, the MLE of parameters (with standard deviations) for the distributions are obtained. To test the goodness-of-fit of the distributions, we calculated the maximized log-likelihood, the Kolmogorov-Smirnov (K-S) statistic with its respective p-value, the AIC (Akaike Information Criterion), AICC (AIC with correction), BIC (Bayesian Information Criterion), CM (Cramer-von Mises statistic) and AD (Anderson-Darling statistic) for the six distributions. Here, the significance level is 0.10. To show that the likelihood equations have a unique solution in the parameters,
we plot the profile log-likelihood functions of $\beta$, $\gamma$, $\alpha$ and $\theta$ for the six distributions.

First, we consider the data  consisting of the strengths of 1.5 cm glass fibers  given in
\cite{sm-na-87}
 and measured at the National Physical Laboratory, England. This data is also studied by \cite{ba-sa-co-10}
and is given in Table \ref{tab.data1}.

The results are given in Table \ref{table.EX1} and show that the GGG distribution yields the best fit among the GGP, GGB, GGL, GG and Gompertz distributions. Also, the GGG, GGP, and GGB distribution are better than GG distribution. The plots of the pdfs (together with the data histogram)  and cdfs in Figure \ref{plot.EX} confirm this conclusion. Figures \ref{plot.pro1}  show the profile log-likelihood functions of $\beta$, $\gamma$, $\alpha$ and $\theta$ for the six distributions.

\begin{table}
\begin{center}
\caption{The strengths of glass fibers.}\label{tab.data1}
\begin{tabular}{l} \hline
            0.55, 0.93, 1.25, 1.36, 1.49, 1.52, 1.58, 1.61, 1.64, 1.68, 1.73, 1.81, 2.00, 0.74, 1.04, 1.27,\\

            1.39, 1.49, 1.53, 1.59, 1.61, 1.66, 1.68, 1.76, 1.82, 2.01, 0.77, 1.11, 1.28, 1.42, 1.50, 1.54,\\

            1.60, 1.62, 1.66, 1.69, 1.76, 1.84, 2.24, 0.81, 1.13, 1.29, 1.48, 1.50, 1.55, 1.61, 1.62, 1.66,\\

            1.70, 1.77, 1.84, 0.84, 1.24, 1.30, 1.48, 1.51, 1.55, 1.61, 1.63, 1.67, 1.70, 1.78, 1.89 \\ \hline
\end{tabular}

\caption{The phosphorus concentration in the leaves.}\label{tab.data2}
\begin{tabular}{l} \hline
0.22, 0.17, 0.11,  0.10, 0.15, 0.06, 0.05, 0.07, 0.12, 0.09, 0.23, 0.25, 0.23, 0.24, 0.20, 0.08\\

0.11, 0.12, 0.10,  0.06, 0.20, 0.17, 0.20, 0.11, 0.16, 0.09, 0.10, 0.12, 0.12, 0.10, 0.09, 0.17\\

0.19, 0.21, 0.18, 0.26, 0.19,  0.17, 0.18, 0.20, 0.24, 0.19, 0.21, 0.22, 0.17, 0.08, 0.08, 0.06\\

0.09, 0.22, 0.23, 0.22, 0.19,  0.27, 0.16, 0.28, 0.11, 0.10, 0.20, 0.12, 0.15, 0.08, 0.12, 0.09\\

0.14, 0.07, 0.09, 0.05, 0.06,  0.11, 0.16, 0.20, 0.25, 0.16, 0.13, 0.11, 0.11, 0.11, 0.08, 0.22\\

0.11, 0.13, 0.12, 0.15, 0.12,  0.11, 0.11, 0.15, 0.10, 0.15, 0.17, 0.14, 0.12, 0.18, 0.14, 0.18\\

0.13, 0.12, 0.14, 0.09, 0.10,  0.13, 0.09, 0.11, 0.11, 0.14, 0.07, 0.07, 0.19, 0.17, 0.18, 0.16\\

0.19, 0.15, 0.07, 0.09, 0.17,  0.10, 0.08, 0.15, 0.21, 0.16, 0.08, 0.10, 0.06, 0.08, 0.12, 0.13\\ \hline
\end{tabular}

\end{center}
\end{table}

As a second example, we consider a data set from
\cite{fo-fr-07},
 who studied the soil fertility influence and the characterization
of the biologic fixation of ${\text N}_{2}$ for the {\it Dimorphandra wilsonii rizz growth}. For 128 plants, they made measures of the
phosphorus concentration in the leaves. This data is also studied by \cite{si-bo-di-co-13} and is given in Table \ref{tab.data2}.
Figures \ref{plot.pro2}  show the profile log-likelihood functions of $\beta$, $\gamma$, $\alpha$ and $\theta$ for the six distributions.

The results are given in Table \ref{table.EX2}. Since the estimation of parameter $\theta$ for GGP, GGB, and GGL is close to zero, the estimations of parameters for these distributions are equal to  the estimations of parameters for GG distribution. In fact, The limiting distribution of GGPS when $\theta\rightarrow 0^{+}$ is a GG distribution (see Proposition \ref{prop.1}). Therefore, the value of maximized log-likelihood, $\log(L)$, are equal for these four distributions.   The plots of the pdfs (together with the data histogram)  and cdfs in Figure \ref{plot.EX2} confirm these conclusions. Note that the estimations of parameters for GGG distribution are not equal to  the estimations of parameters for GG distribution. But the $\log(L)$'s are equal for these distributions. However, from Table \ref{table.EX2} also we can conclude that the GG distribution is simpler than other distribution because it has three parameter but GGG, GGP, GGB, and GGL have four parameter. Note that GG is a special case of GGPS family.

\begin{table}
\begin{center}
\caption{Parameter estimates (with std.), K-S statistic,
\textit{p}-value, AIC, AICC and BIC for the first data set.}\label{table.EX1}

\begin{tabular}{|l|cccccc|} \hline
Distribution & Gompertz & GG & GGG & GGP & GGB & GGL \\ \hline
$\hat{\beta }$  & 0.0088  & 0.0356  & 0.7320  & 0.1404  & 0.1032  & 0.1705  \\
$s.e.(\hat{\beta })$ & 0.0043 & 0.0402 & 0.2484 & 0.1368 & 0.1039 & 0.2571 \\ \hline
$\hat{\gamma }$ & 3.6474  & 2.8834 & 1.3499 & 2.1928 & 2.3489 & 2.1502 \\
$s.e.(\hat{\gamma })$ & 0.2992 & 0.6346 & 0.3290 & 0.5867 & 0.6010 & 0.7667 \\ \hline
$\hat{\alpha }$ & --- & 1.6059 & 2.1853 & 1.6205 & 1.5999 & 2.2177 \\
$s.e.(\hat{\alpha })$ & --- & 0.6540 & 1.2470 & 0.9998 & 0.9081 & 1.3905 \\ \hline
$\hat{\theta }$ & --- & --- & 0.9546 & 2.6078 & 0.6558 & 0.8890 \\
$s.e.(\hat{\theta })$ & --- & --- & 0.0556 & 1.6313 & 0.5689 & 0.2467 \\ \hline
$-{\log  (L)\ }$ & 14.8081 & 14.1452 & 12.0529 & 13.0486 & 13.2670 & 13.6398 \\
K-S & 0.1268 & 0.1318 & 0.0993 & 0.1131 & 0.1167 & 0.1353 \\
p-value & 0.2636 & 0.2239 & 0.5629 & 0.3961 & 0.3570 & 0.1992 \\
AIC & 33.6162 & 34.2904 & 32.1059 & 34.0971 & 34.5340 & 35.2796 \\
AICC & 33.8162 & 34.6972 & 32.7956 & 34.78678 & 35.2236 & 35.9692 \\
BIC & 37.9025 & 40.7198 & 40.6784 & 42.6696 & 43.1065 & 43.8521 \\
CM & 0.1616 & 0.1564 & 0.0792 & 0.1088 & 0.1172 & 0.1542 \\
AD & 0.9062 & 0.8864 & 0.5103 & 0.6605 & 0.7012 & 0.8331 \\ \hline

\end{tabular}
\end{center}
%
\begin{center}
\caption{Parameter estimates (with std.), K-S statistic,
\textit{p}-value, AIC, AICC and BIC for the second data set.}\label{table.EX2}

\begin{tabular}{|l|cccccc|} \hline
Distribution & Gompertz & GG & GGG & GGP & GGB & GGL \\ \hline
 $\hat{\beta }$  & 1.3231 & 13.3618 & 10.8956 & 13.3618 & 13.3618 & 13.3618 \\
$s.e.(\hat{\beta })$  & 0.2797 & 4.5733 & 8.4255 & 5.8585 & 6.3389 & 7.3125 \\ \hline
 $\hat{\gamma }$  & 15.3586 & 3.1500 & 4.0158 & 3.1500 & 3.1500 & 3.1500 \\
$s.e.(\hat{\gamma })$  & 1.3642 & 2.1865 & 3.6448 & 2.4884 & 2.6095 & 2.5024 \\ \hline
 $\hat{\alpha }$  &  ---  & 6.0906 & 5.4236 & 6.0906 & 6.0906 & 6.0905 \\
$s.e.(\hat{\alpha })$  &  ---  & 2.4312 & 2.8804 & 2.6246 & 2.7055 & 2.8251 \\ \hline
 $\hat{\theta }$  &  ---  &  ---  & -0.3429 & $1.0\times {10}^{-8}$ & $1.0\times {10}^{-8}$ & $1.0\times {10}^{-8}$ \\
$s.e.(\hat{\theta })$  &  ---  &  ---  & 1.2797 & 0.8151 & 0.2441 & 0.6333 \\ \hline
 $-\log (L)\; $  & -184.5971 & -197.1326 & -197.1811 & -197.1326 & -197.1326 & -197.1326 \\
K-S  & 0.1169 & 0.0923 & 0.0898 & 0.0923 & 0.0923 & 0.0923 \\
p-value  & 0.06022 & 0.2259 & 0.2523 & 0.2259 & 0.2259 & 0.2259 \\
AIC  & -365.1943 & -388.2653 & -386.3623 & -386.2653 & -386.2653 & -386.2653 \\
AICC  & -365.0983 & -388.0717 & -386.0371 & -385.9401 & -385.9401 & -385.9401 \\
BIC  & -359.4902 & -379.7092 & -374.9542 & -374.8571 & -374.8571 & -374.8571 \\
CM & 0.3343 & 0.1379 & 0.1356 & 0.1379 & 0.1379 & 0.1379 \\
AD & 2.3291 & 0.7730 & 0.7646 & 0.7730 & 0.7730 & 0.7730 \\ \hline

\end{tabular}

\end{center}
\end{table}

\begin{figure}[th]
\centering
\includegraphics[width=7.75cm,height=7.75cm]{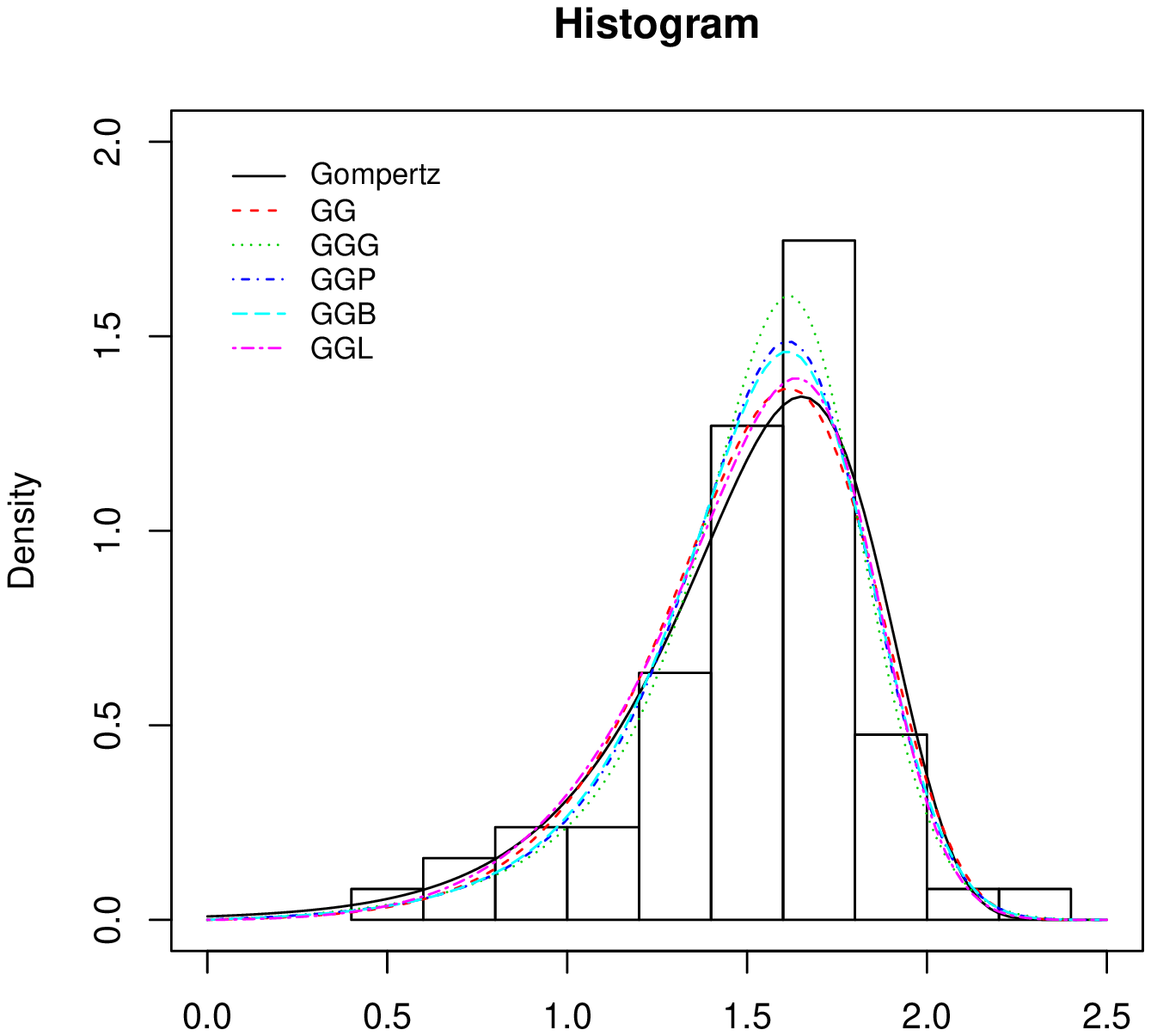}
\includegraphics[width=7.75cm,height=7.75cm]{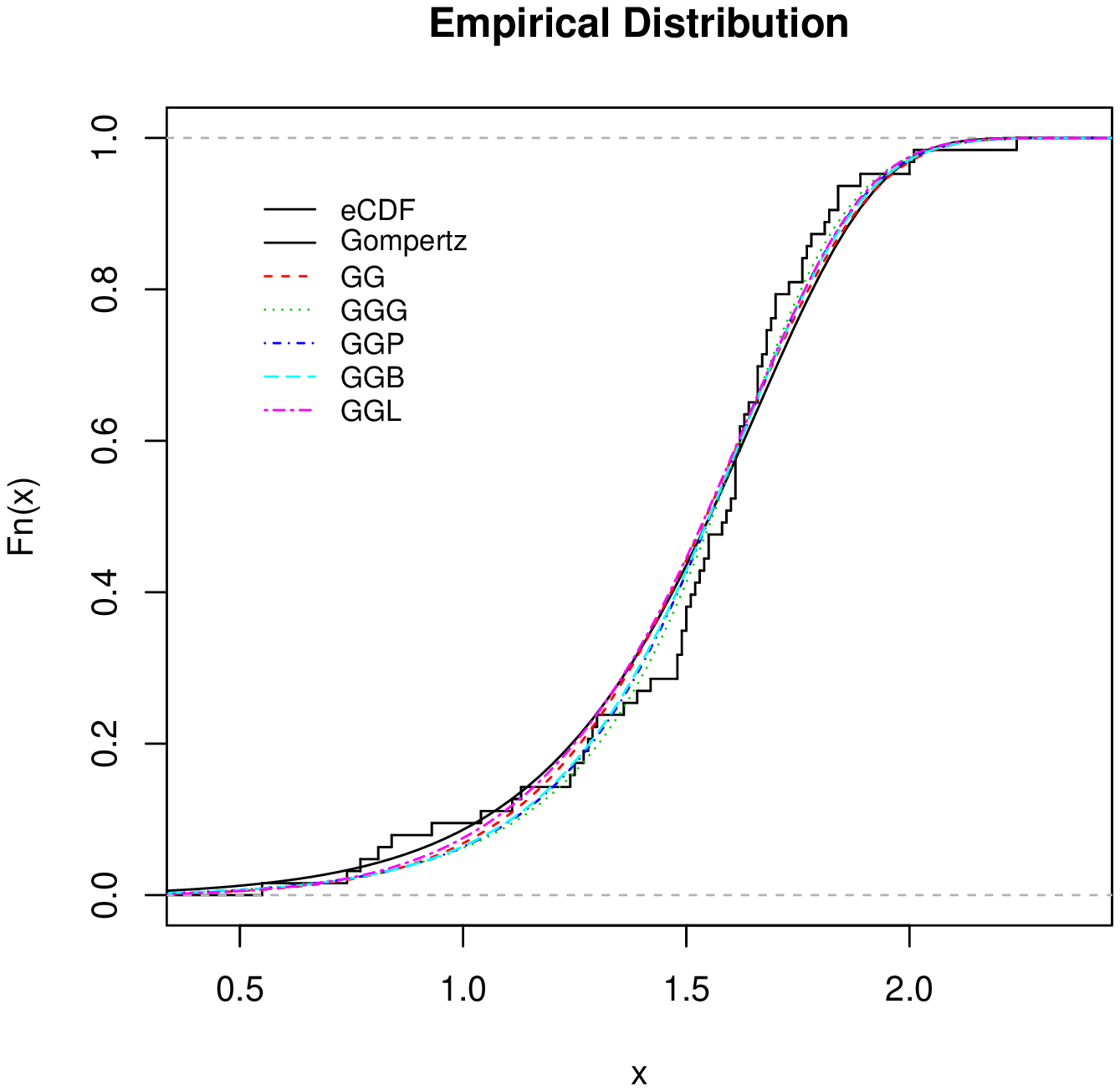}
\vspace{-0.8cm}
\caption{Plots (pdf and cdf) of  fitted Gompertz, generalized  Gompertz, GGG, GGP, GGB and GGL distributions for the first data set.}\label{plot.EX}
%
\includegraphics[width=7.75cm,height=7.75cm]{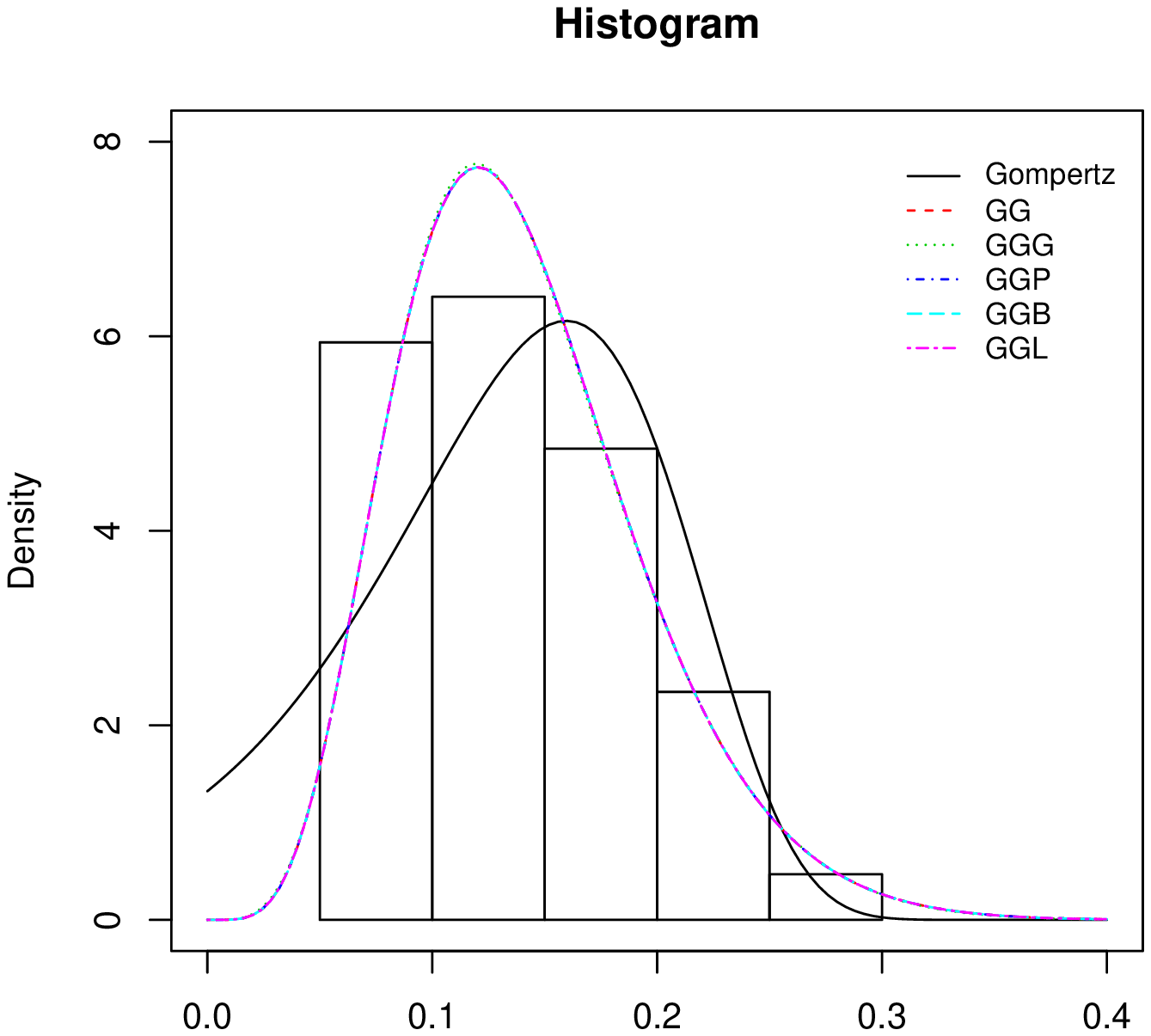}
\includegraphics[width=7.75cm,height=7.75cm]{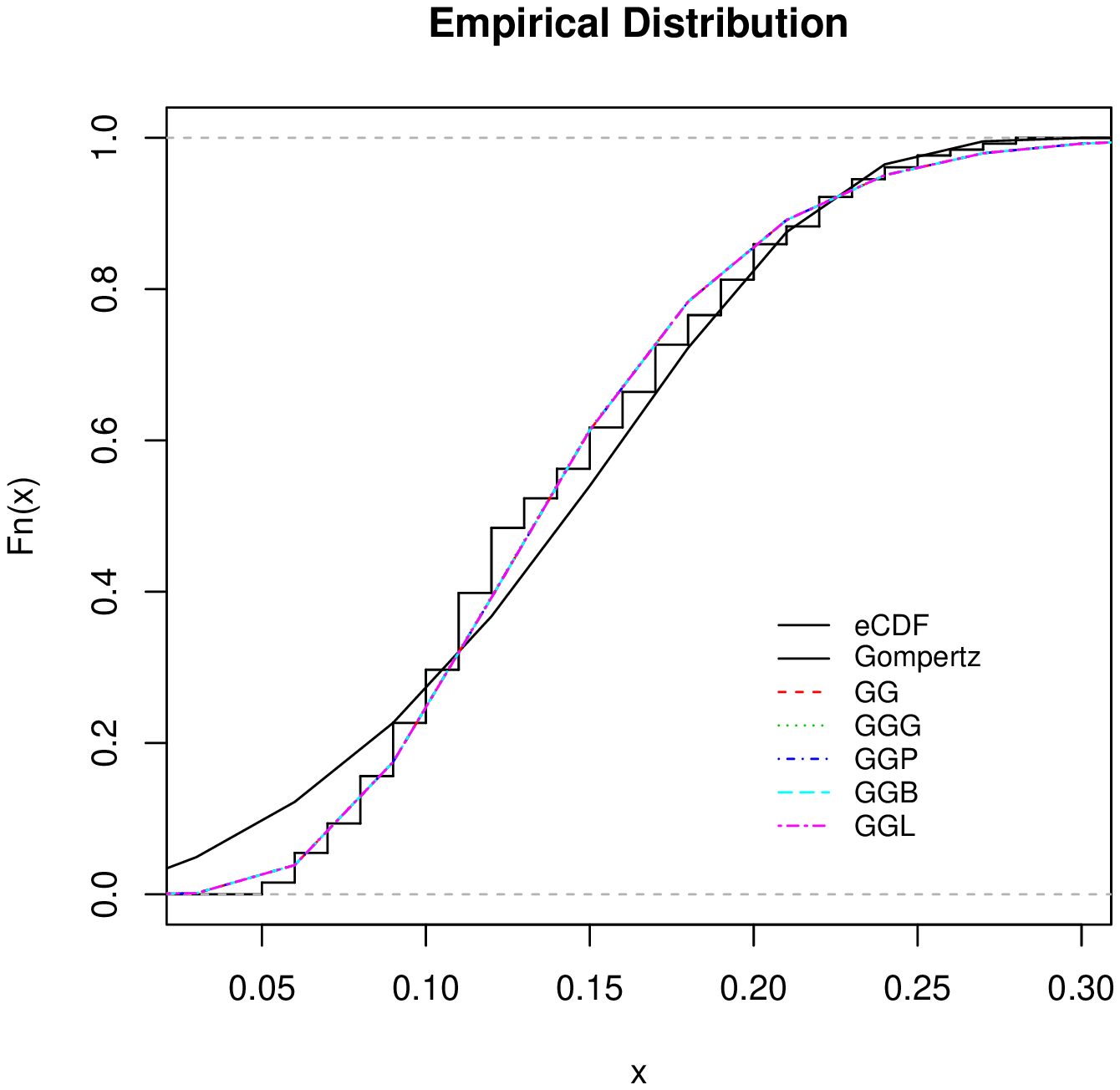}
\vspace{-0.8cm}
\caption{Plots (pdf and cdf) of  fitted Gompertz, generalized  Gompertz, GGG, GGP, GGB and GGL distributions for the second data set.}\label{plot.EX2}
\end{figure}

\begin{figure}
\centering
\includegraphics[width=7.75cm,height=7.75cm]{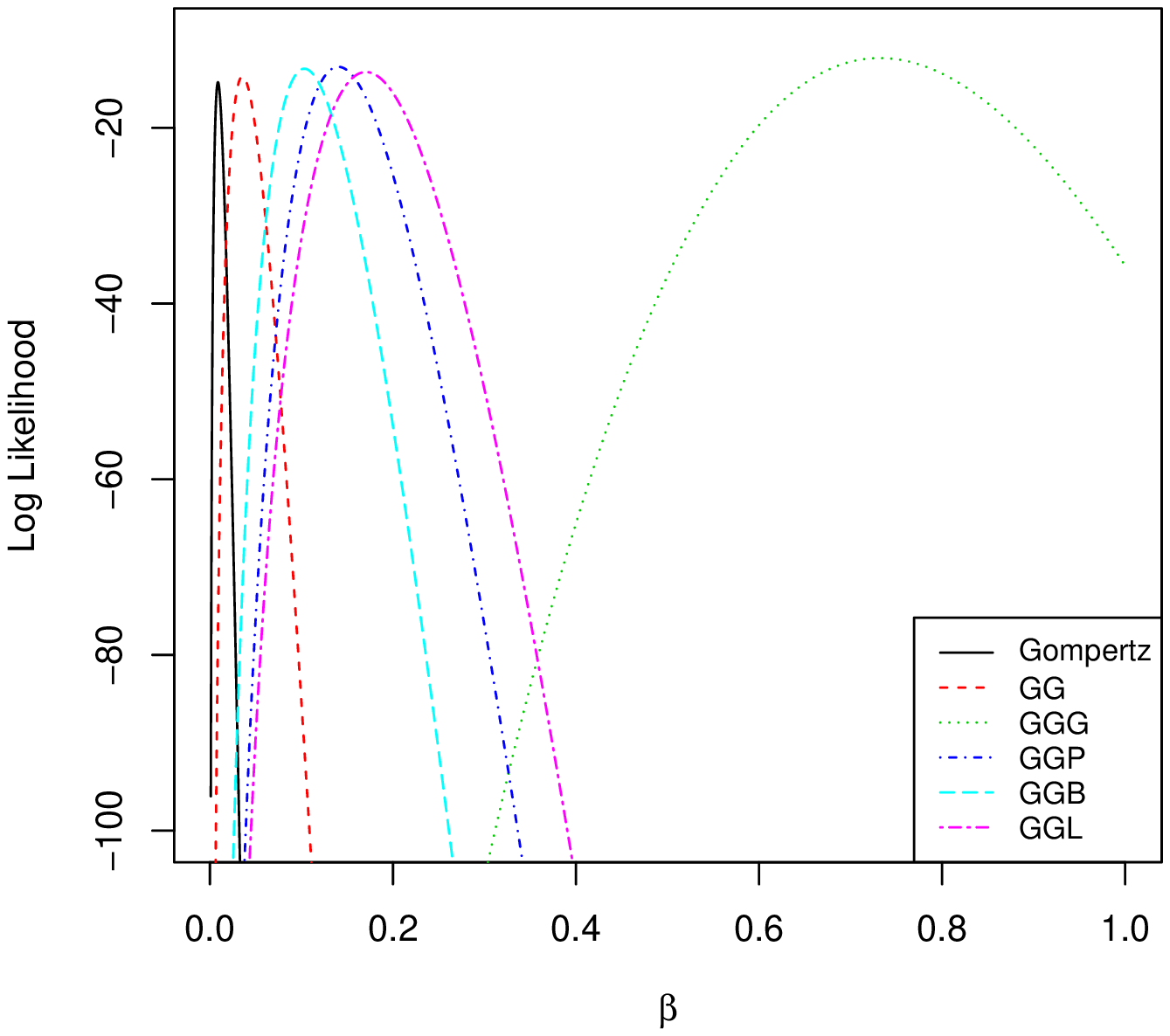}
\includegraphics[width=7.75cm,height=7.75cm]{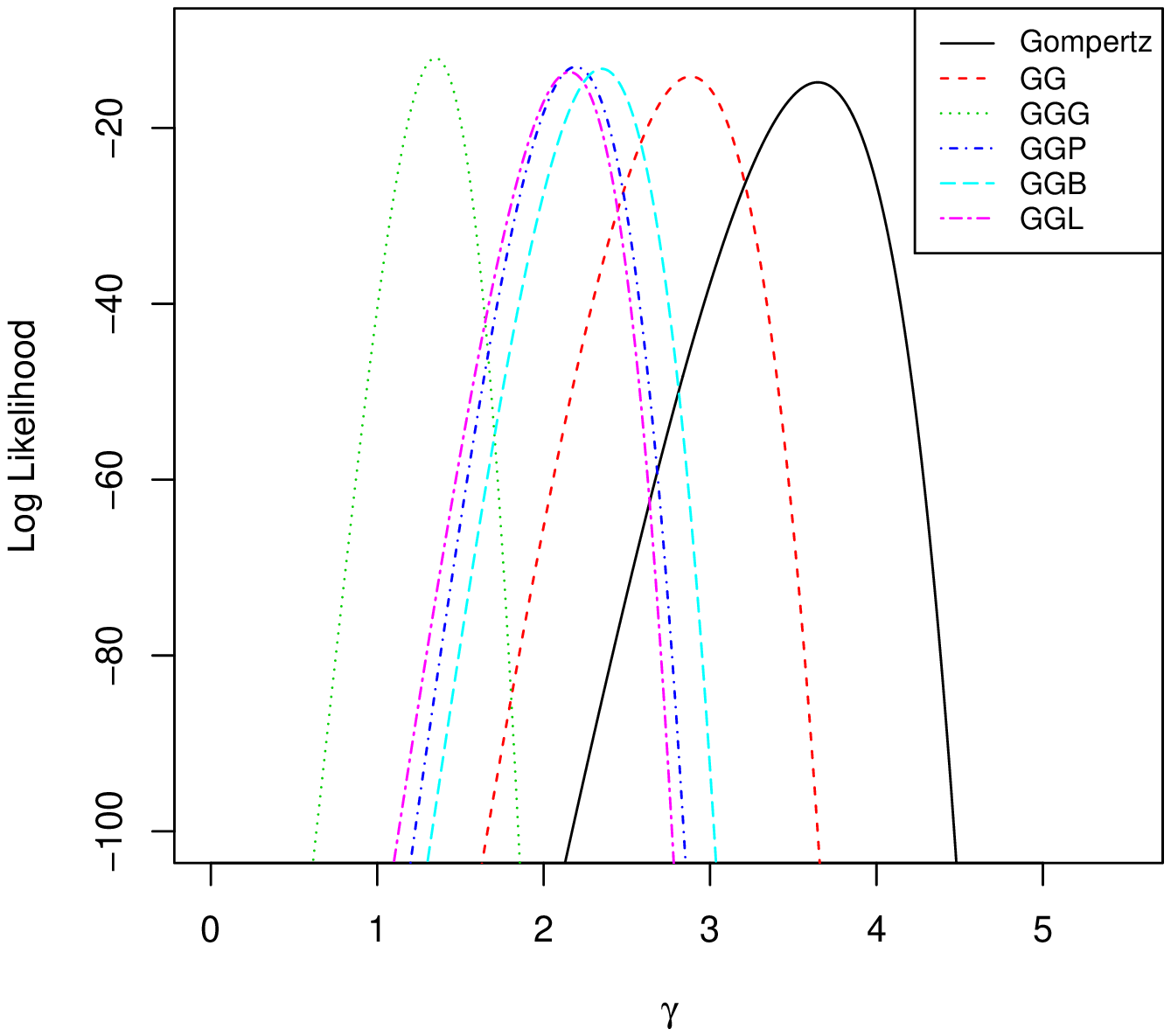}
\includegraphics[width=7.75cm,height=7.75cm]{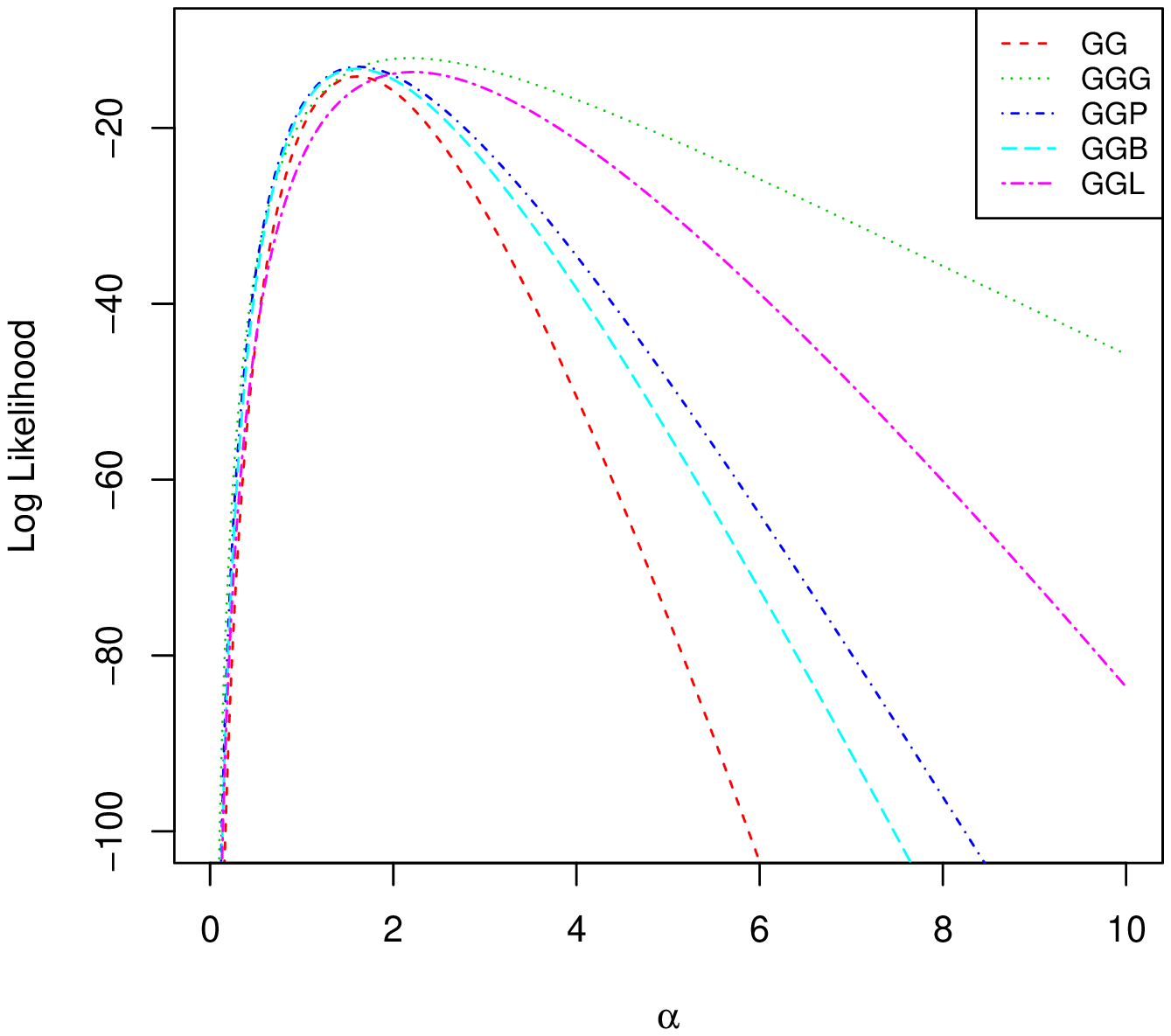}
\includegraphics[width=7.75cm,height=7.75cm]{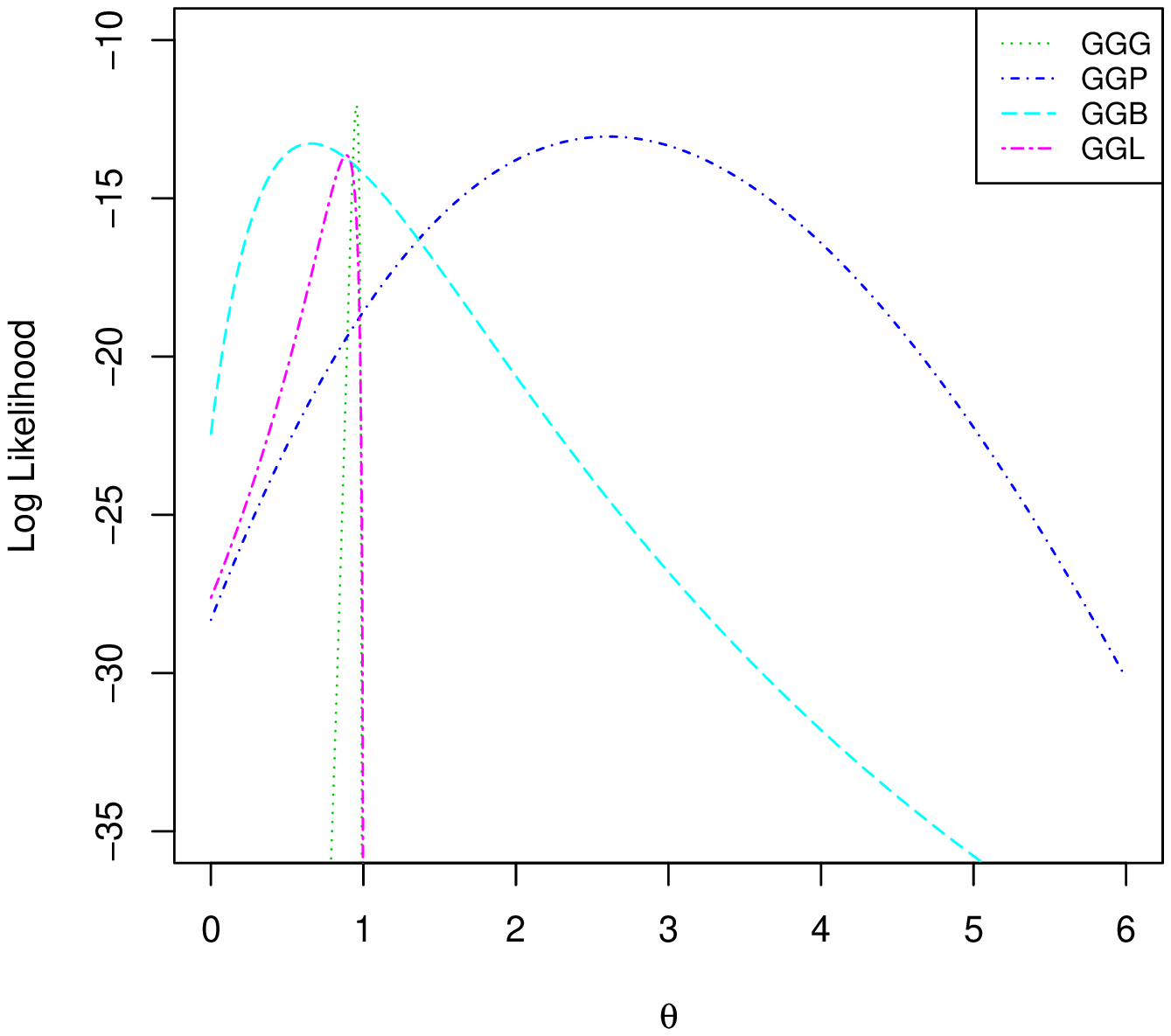}
\vspace{-0.8cm}
\caption{The profile log-likelihood functions for  Gompertz, generalized  Gompertz, GGG, GGP, GGB and GGL distributions for the first data set. }\label{plot.pro1}
\end{figure}

\begin{figure}
\centering
\includegraphics[width=7.75cm,height=7.75cm]{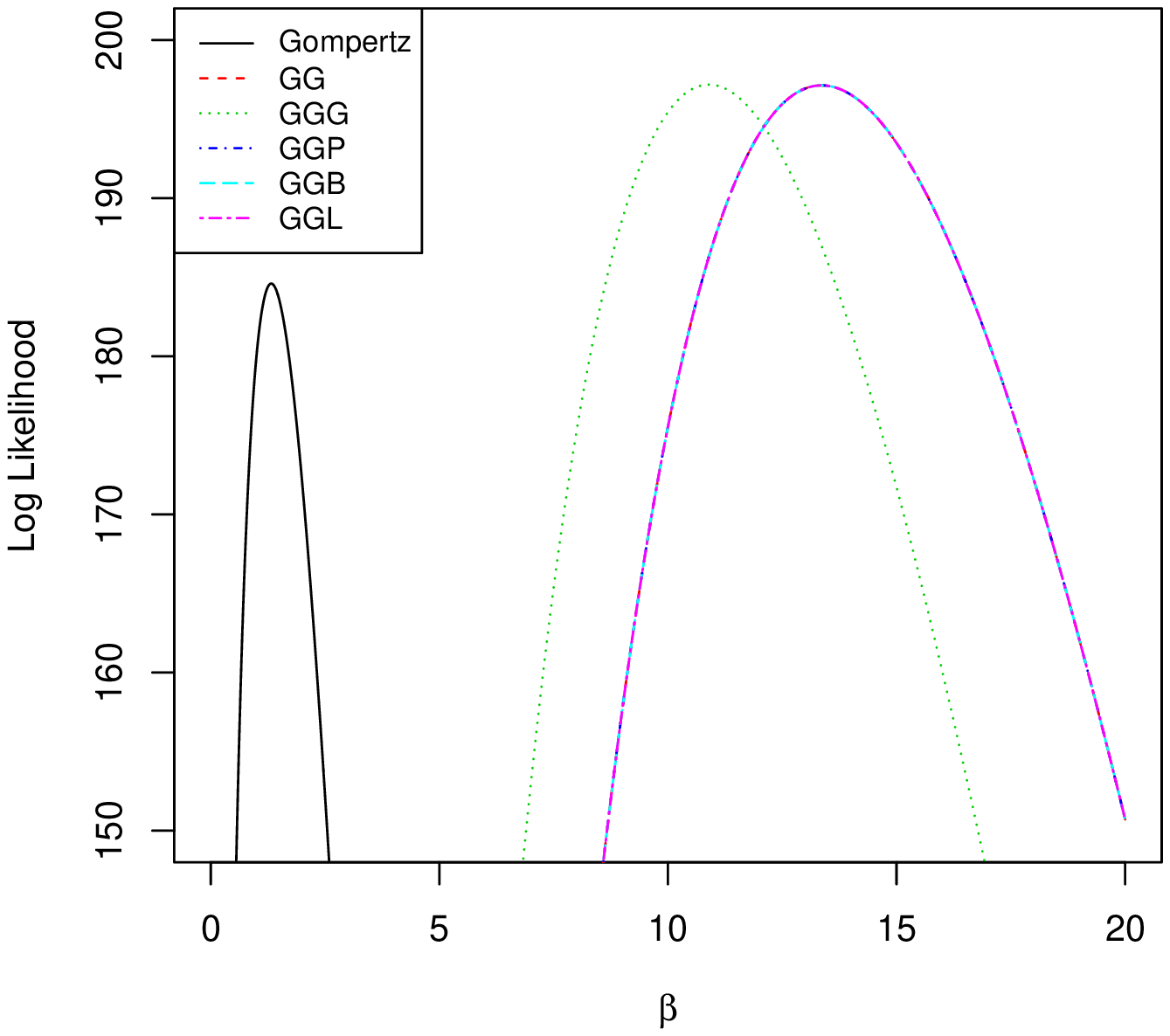}
\includegraphics[width=7.75cm,height=7.75cm]{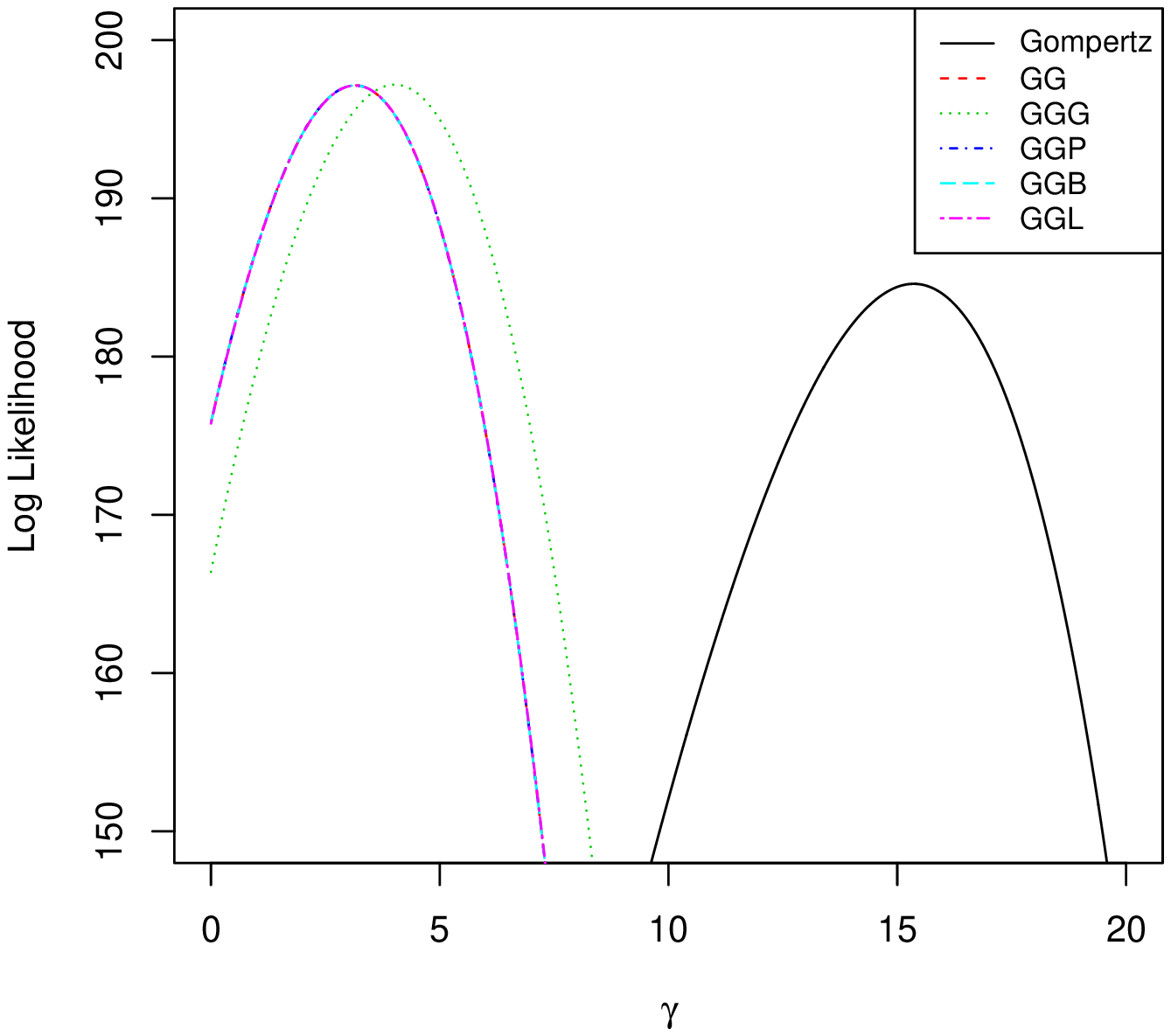}
\includegraphics[width=7.75cm,height=7.75cm]{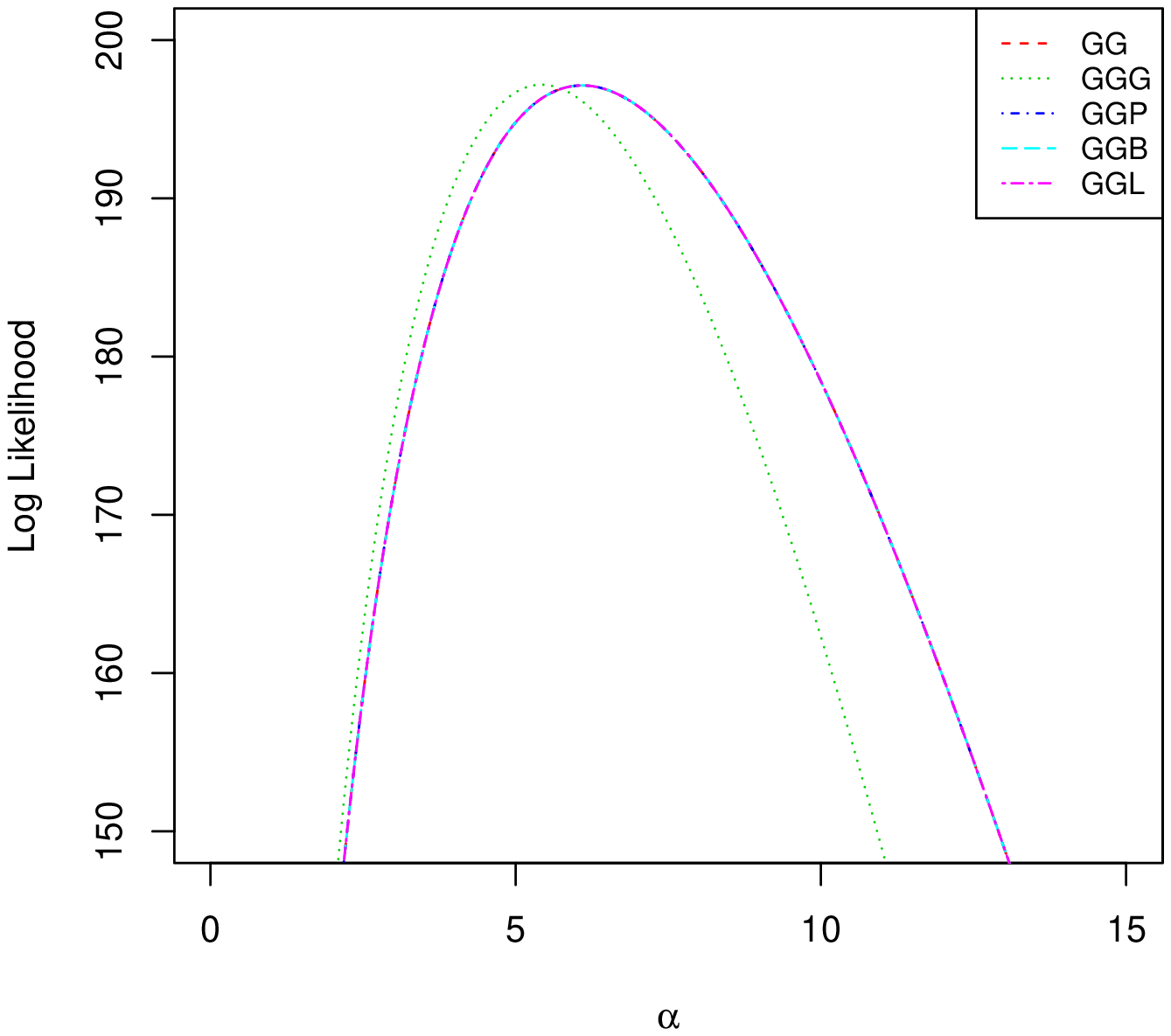}
\includegraphics[width=7.75cm,height=7.75cm]{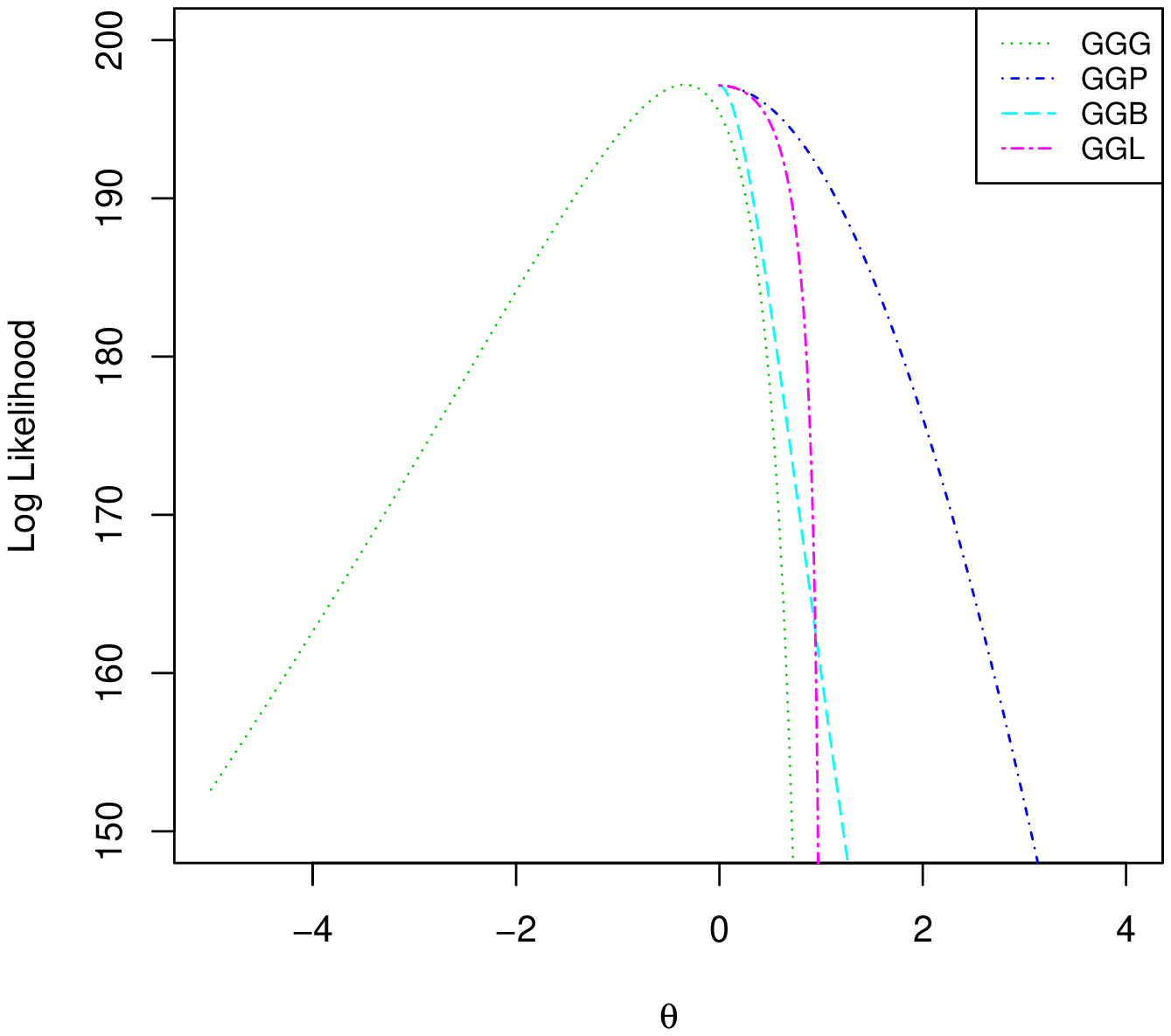}
\vspace{-0.8cm}
\caption{The profile log-likelihood functions for  Gompertz, generalized  Gompertz, GGG, GGP, GGB and GGL distributions for the second data set.}\label{plot.pro2}
\end{figure}

\newpage

\section*{ Appendix }

\subsection*{A.}
We demonstrate those parameter intervals for which the hazard function is decreasing, increasing and bathtub
shaped, and in order to do so, we follow closely a theorem given by
\cite{glaser-80}.
 Define the function $\tau(x)=\frac{-f'(x)}{f(x)}$  where $f'(x)$
 denotes the first derivative of $f(x)$ in \eqref{fGP}. To simplify, we consider $u=1-\exp(\frac{-\theta}{\gamma}(e^{\gamma x}-1))$.
\subsubsection*{A.1}
Consider the GGG hazard function in \eqref{eq.hGG}, then we define
$$
\tau(u)=\frac{-f'(u)}{f(u)}=\frac{1-\alpha}{u}+\frac{2\alpha\theta u^{\alpha-1}}{1-\theta u^{\alpha}}.
$$
If $\alpha\geq 1$, then $\tau'(u)>0$, and $h(.)$ is an increasing function.
If $0<\alpha< 1$, then
$$\lim_{u\rightarrow 0}\tau'(u)=-\infty, \;\;
\lim_{u\rightarrow 1}\tau'(u)=\frac{2\alpha \theta^{2}}{(1-\theta)^{2}}+(\alpha-1)(1-\frac{1}{(1-\theta)^{2}})>0.
$$
Since the limits have different signs, the equation
$\tau'(u)=0$ has at least one root.
Also, we can show that $\tau''(u)>0$.
Therefore, the equation $\tau'(u)=0$ has one root. Thus the hazard function is decreasing and bathtub shaped in this case.

\subsubsection*{A.2}
The  GGP hazard rate is given by
$
h(u)={\theta\alpha\beta u^{\alpha-1}e^{\theta u^{\alpha}}}/({e^{\theta}-e^{\theta u^{\alpha}}}).
$
We define $\eta(u)=\log[h(u)]$. 
Then, its first derivative is
$$
\eta'(u)=\frac{\alpha-1}{u}+\alpha\theta e^{\theta}\frac{u^{\alpha-1}}{e^{\theta}-e^{\theta u^{\alpha}}}.
$$
It is clearly for $\alpha\geq 1$, $\eta'(u)>0$ and  $h(u)$ is increasing function.
If  $0<\alpha<1$, then
$$\lim_{u\rightarrow 0}\eta'(u)=-\infty,\;\; \lim_{u\rightarrow 1}\eta'(u)=0,$$
So the equation  $\tau'(u)=0$ has at least one root. Also, we can show that $\tau''(u)>0$.
It implies that equation $\eta'(u)=0$ has a one root and the hazard rate increase and bathtub shaped.

\subsection*{B.}
\subsubsection*{B.1}
Let $w_1(\alpha)=\sum_{i=1}^{n}\frac{\theta t_i^{\alpha} \log(t_i)C''(\theta t_i^{\alpha})}{C'(\theta t_i^{\alpha})}=\frac{\partial}{\partial\alpha}\sum_{i=1}^{n}\log(C'(\theta t_i^{\alpha}))$.
For GGG,
$$
w_{1}(\alpha)=2\theta \sum_{i=1}^{n}\frac{t_{i}^{\alpha}\log t_{i} }{1-\theta t_{i}^{\alpha}},\ \ \ \    \frac{\partial w_{1}(\alpha)}{\partial \alpha}=2\theta \sum_{i=1}^{n}t_{i}^{\alpha}[\frac{\log t_{i} }{1-\theta t_{i}^{\alpha}}]^{2}>0.
$$
For GGP,
$$
w_{1}(\alpha)=\theta \sum_{i=1}^{n}t_{i}^{\alpha}\log t_{i}, \ \ \ \ \frac{\partial w_{1}(\alpha)}{\partial \alpha}=\theta \sum_{i=1}^{n}t_{i}^{\alpha}[\log t_{i} ]^{2}>0.
$$
For GGL,
$$
w_{1}(\alpha)=\theta \sum_{i=1}^{n}\frac{t_{i}^{\alpha}\log t_{i} }{1-\theta t_{i}^{\alpha}},\ \ \ \ \frac{\partial w_{1}(\alpha)}{\partial \alpha}=\theta \sum_{i=1}^{n}t_{i}^{\alpha}[\frac{\log t_{i} }{1-\theta t_{i}^{\alpha}}]^{2}>0.
$$
For GGB,
$$
w_{1}(\alpha)=(m-1)\theta\sum_{i=1}^{n}\frac{t_{i}^{\alpha}\log t_{i} }{1+\theta t_{i}^{\alpha}},\ \ \ \ \frac{\partial w_{1}(\alpha)}{\partial \alpha}=(m-1)\theta \sum_{i=1}^{n}t_{i}^{\alpha}[\frac{\log t_{i} }{1+\theta t_{i}^{\alpha}}]^{2}>0.
$$
Therefore, $w_1(\alpha)$ is strictly increasing in $\alpha$ and
$$\lim_{\alpha\rightarrow0^+} {\rm g}_1(\alpha;\beta,\gamma,\theta,x)=\infty, \qquad
\lim_{\alpha\rightarrow\infty} {\rm g}_1(\alpha;\beta,\gamma,\theta,x)=\sum_{i=1}^{n}\log(t_i).$$
Also,
$${\rm g}_1(\alpha;\beta,\gamma,\theta,x)<\frac{n}{\alpha}+\sum_{i=1}^{n}\log(t_i),\qquad
{\rm g}_1(\alpha;\beta,\gamma,\theta,x)>\frac{n}{\alpha}+(\frac{\theta C''(\theta)}{C'(\theta)}+1)\sum_{i=1}^{n}\log(t_i).$$
Hence, ${\rm g}_1(\alpha;\beta,\gamma,\theta,x)<0$ when $\frac{n}{\alpha}+\sum_{i=1}^{n}\log(t_i)<0$,
and ${\rm g}_1(\alpha;\beta,\gamma,\theta,x)>0$ when $\frac{n}{\alpha}+(\frac{\theta C''(\theta)}{C'(\theta)}+1)\sum_{i=1}^{n}\log(t_i)>0$. The proof is completed.

\subsubsection*{B.2}
It can be easily shown that
$$
\lim_{\beta\rightarrow0^+} {\rm g}_2(\beta;\alpha,\gamma,\theta,x)=\infty,      \ \ \ \ \ \       \lim_{\beta\rightarrow\infty} {\rm g}_2(\beta;\alpha,\gamma,\theta,x)=\frac{-1}{\gamma}\sum_{i=1}^{n}(e^{\gamma x_{i}}-1).
$$
Since the limits have different signs, the equation ${\rm g}_2(\beta;\alpha,\gamma,\theta,x)=0$ has at least one root with respect to $\beta$ for fixed values $\alpha$, $\gamma$ and $\theta$.  The proof is completed.

\subsubsection*{B.3}

\noindent a) For GGP, it is clear that
\[ \mathop{\lim }_{\theta \rightarrow 0}  {\rm g}_{3}(\theta;\alpha,\beta,\gamma,x)=\sum_{i=1}^{n}t_{i}^{\alpha}-\frac{n}{2}, \ \ \ \ \
\mathop{\lim }_{\theta \rightarrow \infty} {\rm g}_{3}(\theta;\alpha,\beta,\gamma,x)=-\infty.\]
Therefore, the equation ${\rm g}_{3}(\theta;\alpha,\beta,\gamma,x)=0$ has at least one root for $\theta>0$, if $\sum_{i=1}^{n}t_{i}^{\alpha}-\frac{n}{2}>0$  or $ \sum_{i=1}^{n}t_{i}^{\alpha}>\frac{n}{2}$.

\noindent b) For GGG, it is clear that
\[ \mathop{\lim }_{\theta \rightarrow \infty} {\rm g}_{3}(\theta;\alpha,\beta,\gamma,x)=-\infty, \ \ \ \ \
 \mathop{\lim }_{\theta \rightarrow 0^+} {\rm g}_{3}(\theta;\alpha,\beta,\gamma,x)=-n+2\sum_{i=1}^{n}t_{i}^{\alpha}.\]
Therefore, the
equation  ${\rm g}_{3}(\theta,\beta,\gamma,x)=0$ has at least one root for $0<\theta<1$, if  $-n+2\sum_{i=1}^{n}t_{i}^{\alpha} >0$ or $\sum_{i=1}^{n}t_{i}^{\alpha} >\frac{n}{2}$.

\noindent For GGL, it is clear that
 \[ \mathop{\lim }_{\theta \rightarrow 0} {\rm g}_{3}(\theta;\alpha,\beta,\gamma,x)=\sum_{i=1}^{n}t_{i}^{\alpha}-\frac{n}{2},  \ \ \ \ \
 \mathop{\lim }_{\theta \rightarrow 1} {\rm g}_{3}(\theta;\alpha,\beta,\gamma,x)=-\infty.\]
 Therefore, the equation ${\rm g}_{3}(\theta;\alpha,\beta,\gamma,x)=0$ has at least one root for $0<\theta<1$, if $\sum_{i=1}^{n}t_{i}^{\alpha}-\frac{n}{2}>0$  or $ \sum_{i=1}^{n}t_{i}^{\alpha}>\frac{n}{2}$.

\noindent For GGB, it is clear that
\[ \mathop{\lim }_{p \rightarrow 0} {\rm g}_{3}(p;\alpha,\beta,\gamma,x)=\sum_{i=1}^{n}t_{i}^{\alpha}(m-1)-\frac{n(m-1)}{2}, \ \ \
 \mathop{\lim }_{p \rightarrow 0} {\rm g}_{3}(p;\alpha,\beta,\gamma,x)=\sum_{i=1}^{n}\frac{-m+1+m t_{i}^{\alpha}}{t_{i}},\]
 Therefore, the equation ${\rm g}_{3}(p;\alpha,\beta,\gamma,x)=0$ has at least one root for $0<p<1$, if $\sum_{i=1}^{n}t_{i}^{\alpha}(m-1)-\frac{n(m-1)}{2}>0$  and  $\sum_{i=1}^{n}\frac{-m+1+m t_{i}^{\alpha}}{t_{i}^{\alpha}}<0$ or $\sum_{i=1}^{n}t_{i}^{\alpha}>\frac{n}{2}$ and $\sum_{i=1}^{n} t_{i}^{-\alpha}>\frac{nm}{1-m}$.

\subsection*{C.}
Consider $t_{i}=1-e^{-\frac{\beta}{\gamma}(e^{\gamma x_{i}}-1)}$. Then, the elements of $4\times 4$  observed information matrix $I_{n}(\Theta)$  are given by
\begin{eqnarray*}
 I_{\alpha\alpha}&=& \frac{\partial^{2}l_{n}}{\partial \alpha^{2}}=\frac{-n}{\alpha^{2}}+\theta\sum_{i=1}^{n}t_{i}^{\alpha} [\log(t_i)]^{2}[\frac{C''(\theta t_{i}^{\alpha})}{C'(\theta t_{i}^{\alpha})}+\theta t_{i}^{\alpha}\frac{C'''(\theta t_{i}^{\alpha})C'(\theta t_{i}^{\alpha})-(C''(\theta t_{i}^{\alpha}))^{2}}{(C'(\theta t_{i}^{\alpha}))^{2}} ],\nonumber
\\
 I_{\alpha\beta}&=& \frac{\partial^{2}l_{n}}{\partial \alpha \partial \beta }=\sum_{i=1}^{n}[\frac{e^{\gamma x_i}-1}{\gamma}]+\frac{\theta}{\gamma}\sum_{i=1}^{n}t_{i}^{\alpha} [e^{\gamma x_i}-1][(\alpha \log(t_i)+1)\frac{C''(\theta t_{i}^{\alpha})}{C'(\theta t_{i}^{\alpha})}
\\
&&+\alpha\theta t_{i}^{\alpha}\log(t_i)\frac{C'''(\theta t_{i}^{\alpha})C'(\theta t_{i}^{\alpha})-(C''(\theta t_{i}^{\alpha}))^{2}}{(C'(\theta t_{i}^{\alpha}))^{2}} ],
\\
 I_{\alpha\gamma}&=& \frac{\partial^{2}l_{n}}{\partial \alpha \partial \gamma }=\beta \sum_{i=1}^{n}[\frac{e^{\gamma x_i}(\gamma x_i -1)+1}{\gamma^{2}}]+\frac{\theta\beta}{\gamma^{2}}\sum_{i=1}^{n} [e^{\gamma x_i(\gamma x_i -1)}+1][(\alpha \log(t_i)+1)\frac{C''(\theta t_{i}^{\alpha})}{C'(\theta t_{i}^{\alpha})}
\\
&&+\alpha\theta t_{i}^{\alpha}\log(t_i)\frac{C'''(\theta t_{i}^{\alpha})C'(\theta t_{i}^{\alpha})-(C''(\theta t_{i}^{\alpha}))^{2}}{(C'(\theta t_{i}^{\alpha}))^{2}} ],
\\
 I_{\alpha\theta}&=& \frac{\partial^{2}l_{n}}{\partial \alpha \partial \theta }= \sum_{i=1}^{n}t_{i}^{\alpha} \log(t_i)[\frac{C''(\theta t_{i}^{\alpha})}{C'(\theta t_{i}^{\alpha})}+\theta t_{i}^{\alpha}\frac{C'''(\theta t_{i}^{\alpha})C'(\theta t_{i}^{\alpha})-(C''(\theta t_{i}^{\alpha}))^{2}}{(C'(\theta t_{i}^{\alpha}))^{2}} ],
\\
 I_{\beta\beta}&=& \frac{\partial^{2}l_{n}}{\partial \beta^{2}}=\frac{-n}{\beta^{2}}+\theta \alpha^{2}\sum_{i=1}^{n}t_{i}^{\alpha} [\frac{e^{\gamma x_i}-1}{\gamma}]^{2}[\frac{C''(\theta t_{i}^{\alpha})}{C'(\theta t_{i}^{\alpha})}+\theta t_{i}^{\alpha}\frac{C'''(\theta t_{i}^{\alpha})C'(\theta t_{i}^{\alpha})-(C''(\theta t_{i}^{\alpha}))^{2}}{(C'(\theta t_{i}^{\alpha}))^{2}}],
 \\
  I_{\beta\gamma}&=& \frac{\partial^{2}l_{n}}{\partial \beta \partial \gamma }=\frac{(\alpha-2)}{\gamma^{2}}\sum_{i=1}^{n}(e^{\gamma x_i}(\gamma x_i -1)+1)+\alpha\theta \sum_{i=1}^{n}\frac{t_i}{\gamma^{2}}(e^{\gamma x_i}(\gamma x_i -1)+1)[\frac{C''(\theta t_{i}^{\alpha})}{C'(\theta t_{i}^{\alpha})}
\\
 &&+\frac{\beta^{2}}{\gamma}(e^{\gamma x_i}-1)\frac{C'''(\theta t_{i}^{\alpha})C'(\theta t_{i}^{\alpha})-(C''(\theta t_{i}^{\alpha}))^{2}}{(C'(\theta t_{i}^{\alpha})))^{2}}]
\\
 I_{\beta\theta}&=& \frac{\partial^{2}l_{n}}{\partial \beta \partial \theta}=\sum_{i=1}^{n}t_{i}^{2\alpha}[\frac{C'''(\theta t_{i}^{\alpha})C'(\theta t_{i}^{\alpha})-(C''(\theta t_{i}^{\alpha}))^{2}}{(C'(\theta t_{i}^{\alpha}))^{2}} ],
\\
 I_{\gamma\gamma}&=&  \frac{\partial^{2}l_{n}}{\partial \gamma^{2}}=\frac{2\beta}{\gamma^{3}}\sum_{i=1}^{n}[e^{\gamma x_i}(\gamma x_i -1)+1]
 +(\alpha-1)\beta \sum_{i=1}^{n}[\frac{-2}{\gamma^{3}}(e^{\gamma x_i}(\gamma x_i -1)+1)+\frac{x_i^{2}e^{\gamma x_i}}{\gamma^{3}}]
 \\
&&+\alpha\beta\theta\sum_{i=1}^{n}[\frac{-2}{\gamma^{3}}(e^{\gamma x_i}(\gamma x_i -1)+1)t_{i}^{\alpha}\frac{C''(\theta t_{i}^{\alpha})}{C'(\theta t_{i}^{\alpha})}+\frac{t_{i}^{\alpha}x_i^{2}e^{\gamma x_i}}{\gamma}\frac{C''(\theta t_{i}^{\alpha})}{C'(\theta t_{i}^{\alpha})}
\\
&&+\frac{\alpha\beta t_{i}^{\alpha}}{\gamma^{4}}(e^{\gamma x_i}(\gamma x_i -1)+1)^{2}\frac{C''(\theta t_{i}^{\alpha})}{C'(\theta t_{i}^{\alpha})}
\\
&&+\frac{\alpha\beta t_{i}^{2\alpha}}{\gamma^{4}}(e^{\gamma x_i}(\gamma x_i -1)+1)^{2}\frac{C'''(\theta t_{i}^{\alpha})C'(\theta t_{i}^{\alpha})-(C''(\theta t_{i}^{\alpha}))^{2}}{(C'(\theta t_{i}^{\alpha})))^{2}}],
\\
 I_{\theta\gamma}&=& \frac{\partial^{2}l_{n}}{\partial \theta \partial \gamma }= \alpha\beta\sum_{i=1}^{n}\frac{t_{i}^{\alpha}}{\gamma^{2}}[e^{\gamma x_i}(\gamma x_i -1)+1][\frac{C''(\theta t_{i}^{\alpha})}{C'(\theta t_{i}^{\alpha})}+\theta t_{i}^{\alpha}\frac{C'''(\theta t_{i}^{\alpha})C'(\theta t_{i}^{\alpha})-(C''(\theta t_{i}^{\alpha}))^{2}}{(C'(\theta t_{i}^{\alpha}))^{2}}],
\\
 I_{\theta\theta}&=& \frac{\partial^{2}l_{n}}{\partial \theta^{2}}=\frac{-n}{\theta^{2}}+\sum_{i=1}^{n}t_{i}^{2\alpha}[\frac{C'''(\theta t_{i}^{\alpha})C'(\theta t_{i}^{\alpha})-(C''(\theta t_{i}^{\alpha}))^{2}}{(C'(\theta t_{i}^{\alpha}))^{2}} ]-n[\frac{C''(\theta)C'(\theta)-(C'(\theta))^{2}}{(C'(\theta ))^{2}} ],
 \end{eqnarray*}

\section*{Acknowledgements}
The authors would like to thank the anonymous  referees  for many helpful comments and  suggestions.

\bibliographystyle{apa}

\end{document}